\documentclass[11pt]{article}

\usepackage{amsmath}
\usepackage{verbatim,color,amssymb}
\usepackage{amsthm}					
\usepackage{natbib}
\usepackage{multirow}
\usepackage{setspace}
\usepackage{enumitem}
\usepackage{graphicx}
\usepackage{subfigure}
\usepackage{lineno}
\usepackage{array,booktabs}
\usepackage{ bbold }
\usepackage{url}
\usepackage{graphicx} 
\usepackage{booktabs} 
\usepackage{pgfplots}
\pgfplotsset{/pgf/number format/use comma,compat=newest}

\usepackage{multibib}
\newcites{latex}{References}

\usepackage{lscape}

\usepackage{subfig}
\usepackage{setspace}
\usepackage{tikz}
\usetikzlibrary{arrows}
\usepackage{multirow}

\renewcommand{\baselinestretch}{1.5}

\newcommand{\bkappa}{\mbox{\boldmath $\kappa$}}

\newcommand{\etam}{\mbox{\boldmath $\eta$}}
\newcommand{\bmu}{\mbox{\boldmath $\mu$}}

\newcommand{\bxi}{\mbox{\boldmath $\xi$}}

\newcommand{\btheta}{\mbox{\boldmath $\theta$}}
\newcommand{\bbeta}{\mbox{\boldmath $\beta$}}

\newcommand{\bSigma}{\mbox{\boldmath $\Sigma$}}
\newcommand{\balpha}{\mbox{\boldmath $\alpha$}}

\newcommand{\bPsi}{\mbox{\boldmath $\Psi$}}

\newtheorem{Thm}{\underline{\bf Theorem}}
\newtheorem{Assmp}{\underline{\bf Assumption}}

\newtheorem*{Proof*}{Proof}

\newtheorem{Rem}{\underline{\bf Remark}}

\newtheorem{Lem}{\underline{\bf Lemma}}

\newtheorem{SLem}{Lemma S\ignorespaces}

\newtheorem{SDef}{Definition S\ignorespaces}

\def\eE{\mathbb{E}}

\def\A{{\cal A}}

\def\B{{\cal B}}
\def\C{{\cal C}}
\def\D{{\cal D}}

\def\F{{\cal F}}
\def\G{{\cal G}}
\def\H{{\cal H}}

\def\S{{\cal S}}

\def\V{{\cal V}}

\def\diag{\hbox{diag}}

\def\diag{\hbox{diag}}
\def\log{\hbox{log}}

\def\Ga{\hbox{Ga}}

\def\Normal{\hbox{Normal}}

\def\Unif{\hbox{Unif}}

\def\eig{\hbox{eigen}}

\def\bone{{\mathbf 1}}

\def\eE{\mathbb{E}}

\def\b1e{{\mathbf e}}

\def\bA{{\mathbf A}}

\def\pd{{\mathrm d}}
\def\bD{{\mathbf D}}

\def\bI{{\mathbf I}}

\def\bJ{{\mathbf J}}
\def\bK{{\mathbf K}}

\def\bM{{\mathbf M}}

\def\bq{{\mathbf q}}

\def\bS{{\mathbf S}}

\def\bv{{\mathbf v}}

\def\bx{{\mathbf x}}

\def\by{{\mathbf y}}

\def\bS{{\mathbf S}}

\renewcommand\footnoterule{\kern-3pt \hrule \textwidth 2in \kern 2.6pt}

\def\boxit#1{\vbox{\hrule\hbox{\vrule\kern6pt \vbox{\kern6pt \textcolor{blue}{#1}\kern6pt}\kern6pt\vrule}\hrule}}

\def\authorfootnote#1{{\let\thefootnote\relax\footnotetext{#1}}}

\begin{document}


\title{Double soft-thresholded model for multi-group scalar on vector-valued image regression}

\def\spacingset#1{\renewcommand{\baselinestretch}%
{#1}\small\normalsize} \spacingset{1}
\author{Arkaprava Roy\footnote{
Department of Biostatistics, University of Florida, Gainesville, FL email: ark007@ufl.edu}, 
Zhou Lan\footnote{Center for Outcomes Research and Evaluation, Yale University, New Haven, CT, email: zhou.lan@yale.edu}
\\ and \\
for the Alzheimer’s Disease Neuroimaging Initiative\footnote{Data used in the preparation of this article were obtained from the Alzheimer's Disease Neuroimaging Initiative
(ADNI) database (adni.loni.usc.edu). As such, the investigators within the ADNI contributed to the design
and implementation of ADNI and/or provided data but did not participate in analysis or writing of this report.
A complete listing of ADNI investigators can be found at:
{\tt http://adni.loni.usc.edu/wp-content/uploads/how\_to\_apply/ADNI\_Acknowledgement\_List.pdf}}}


\maketitle
\vskip 8mm

{{\bf Abstract:}} 
In this paper, we develop a novel spatial variable selection method for scalar on vector-valued image regression in a multi-group setting.
Here, `vector-valued image' refers to the imaging datasets that contain vector-valued information at each pixel/voxel location, such as in RGB color images,
multimodal medical images, DTI imaging, etc.
The focus of this work is to identify the spatial locations in the image having an important effect on the scalar outcome measure.
Specifically, the overall effect of each voxel is of interest.
We thus develop a novel shrinkage prior by soft-thresholding the $\ell_2$ norm of a latent multivariate Gaussian process.
It allows us to estimate sparse and piecewise-smooth spatially varying vector-valued regression coefficient function.
Motivated by the real data, we further develop a double soft-thresholding based framework when there are multiple pre-specified subgroups.
For posterior inference, an efficient MCMC algorithm is developed.
We compute the posterior contraction rate for parameter estimation and also establish consistency for variable selection of the proposed Bayesian model, assuming that the true regression coefficients are H\"older smooth.
Finally, we demonstrate the advantages of the proposed method in simulation studies and further illustrate in an ADNI dataset for modeling MMSE scores based on DTI-based vector-valued imaging markers.

\baselineskip=12pt

\date{}

\vskip 8mm
\baselineskip=12pt
\noindent{\bf Keywords}: 
Image regression,
Multi-group,
Neuroimaging,
Posterior consistency,
scalar on image regression,
Soft-thresholding,
Variable selection.


\par\medskip\noindent

\clearpage\pagebreak\newpage
\pagenumbering{arabic}
\newlength{\gnat}
\setlength{\gnat}{17pt}
\baselineskip=\gnat

\section{Introduction}
Many modern imaging applications collect vector-valued image data, where each spatial location contains more than one scalar information.
The vector-valued images are more appealing in recent biomedical applications.
For example, in biomedical imaging, vector-valued imaging markers are routinely collected by multimodal imaging methods.
As an example, modern medical scanning devices such as SPECT/CT compute single-photon emission computed tomography (SPECT) and computed tomography (CT) together at each spatial location \citep{ehrhardt2013vector}.
Apart from multimodal imaging, 
there are now several medical imaging methods that collect RGB-based color images \citep{bajaj20013d,lange2005automatic,qi2011comparative, jones2017bayesian}.
Hence, the contrast at each imaging location is represented by a vector instead of a scalar. 
For some immunohistochemistry (IHC) assay \citep{feldman2014tissue}, the RGB information can be further converted into the underlying staining (e.g, hematoxylin and eosin staining) revealing more biomedically informative information.

Our paper is particularly motivated by a medical imaging application called
the diffusion tensor imaging (DTI) \citep{soares2013hitchhiker}.
In this imaging technique, various types of vector-valued imaging markers are assessed to describe white matter tissue properties in the brain. The main purpose of DTI is to identify the white matter anatomical structure of the whole brain, derived from diffusion tensor signals.
A quantity to describe the spatial specific anatomical structure is $\bM(\bv)$, a $3\times 3$ positive definite (pd) matrix, called the diffusion tensor at location $\bv$.
The eigenvectors of these matrices represent the diffusion directions, $[\b1e_1\ \b1e_2\ \b1e_3]$ and the squared roots of the eigenvalues $[l_1\ l_2\ l_3]$ are the corresponding semi-axis lengths \citep[][Section 1.2.3]{zhou2010statistical}. 
In many recent works, the principal eigenvector $\bM(\bv)$ is used to describe the structural information, known as diffusion direction \citep{wong2016fiber}.
Several works have established the clinical significance of summary measurements derived from diffusion tensors.
In the usual route of inference using DTI data, white matter structural connectivity profiles are first derived by reconstructing the fiber tracts using tractography algorithms on the DTI-derived estimates \citep{wong2016fiber}. 
Subsequently, a low-resolution summary of the network, called connectome \citep{sporns2011human} is usually constructed using these profiles assuming some parcellation of the brain.
In clinical applications, their associations with the subject's covariates are studied.
However, such summarization may produce biases, which could lead to inefficient statistical inference.
Thus, in this work, we focus on analyzing the effects of the DTI-derived principal diffusion direction markers directly on the performance-based mini-mental state examination (MMSE) score. 
Due to a large number of predictors in an imaging dataset, it is prudent to employ regularized regression techniques for reasonable inference.
On the other hand, the spatial dependence among the image predictors further prompts us to assume that the regression effects of the neighboring voxels should be close.

Hence, the primary challenges require to be addressed are: 1)
handling imaging predictors having vector-valued information at each voxel location, 2) tackling complex spatial dependence of the data, and 3) identifying regression effects that are simultaneously sparse and continuous, where sparsity is defined in terms of the number of non-zero smooth pieces in it.
We let $\bD_i(\bv)$ be a $q$-dimensional vector-valued predictor for $i$-$th$ subject at location $\bv$, and it can be applied to any vector-valued image data.
In our DTI application, $\bD_i(\bv)$ are principal diffusion direction i.e. first eigenvector of $\bM(\bv)$.
Since the components within $\bD_i(\bv)$ can only be put together to illustrate the diffusion direction within a voxel, each sub-component of the multidimensional vector-valued predictor $\bD_i(\bv)$ can not be of direct interest.
Therefore, identifying the spatial locations $\bv$ having important overall effects on the scalar outcomes is our primary goal.

To characterize a possible relation between a scalar outcome and an imaging predictor, scalar on image regression models are widely popular \citep{Goldsmith,Li,Wang,kang2018scalar}.
Following the traditional scalar on image regression setting, the effect of a voxel $\bv$ can be quantified by $\bD_i(\bv)^T\bbeta(\bv)$, where $\bbeta(\bv)$ stands for the regression effect of the $\bv$-$th$ voxel and $<\bx,\by>$ stand for the inner product between $\bx$ and $\by$.
This can be rewritten as $\bD_i(\bv)^T\bbeta(\bv)=\|\bD_i(\bv)\|_2\|\bbeta(\bv)\|_2\cos(\theta_{\bD_i(\bv),\bbeta(\bv)}),$ where $\theta_{\bD_i(\bv),\bbeta(\bv)}$ is the geometric angle between the two vectors $\bD_i(\bv)$ and $\bbeta(\bv)$ following the well-known dot-product result that for two $q$-dimensional vectors $\bx$ and $\by$, we have $<\bx,\by>=\bx^T\by=\|\bx\|_{2}\|\by\|_{2}\cos(\theta_{\bx,\by})$, where $\|\bx\|_2=\sqrt{\sum_{i=1}^qx^2_i}$ and $\theta_{\bx,\by}$ is the angle between $\bx$ and $\by$ \citep{lipschutz2009schaum}.
Here $\|\bbeta(\bv)\|_2\cos(\theta_{\bD_i(\bv),\bbeta(\bv)})$ can be interpreted as the rate of change of $\bD_i(\bv)^T\bbeta(\bv)$ for a unit change in $\|\bD_i(\bv)\|_2$.

Thus, $\|\bbeta(\bv)\|_2$ can be considered as the universal magnitude that the voxel $\bv$ contributes to the unit change in the scalar outcome. In contrast to this term, 
the other part $\cos(\theta_{\bD_i(\bv),\bbeta(\bv)})\in [-1,1]$ describes the standardized contribution which varies from one subject to the other depending on their own $\bD_i(\bv)$, referred to as individual contribution. For example, in medical imaging applications, $\|\bbeta(\bv)\|_2$ may be large at a certain voxel, but $\bD_i(\bv)^T\bbeta(\bv)$ can vary across the subjects. In Figure \ref{fig:illu}, we provide a graphical representation to illustrate the relationship between individual contribution and universal magnitude: for the same individual contribution ($\cos(\theta_{\bD_i(\bv),\bbeta(\bv)})$), the universal magnitude dominates the contribution to the scalar outcome.

\begin{figure}
  \centering
\includegraphics[width=120mm]{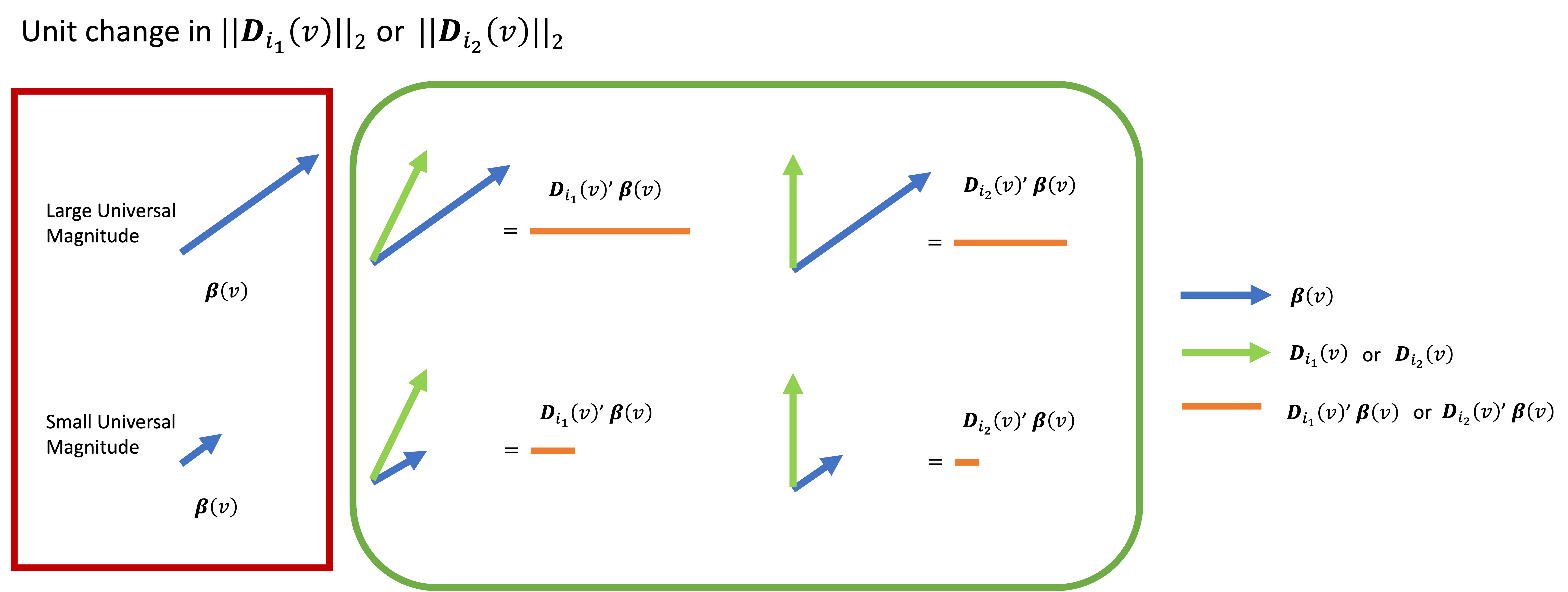}
    \caption{Graphical illustration of universal magnitude: For a given voxel $\bv$, with same individual contributions ($\bD_{i_1}(\bv)$ or $\bD_{i_2}(\bv)$), the coefficient ($\bbeta(\bv)$) with large universal magnitude ($\|\bbeta(\bv)\|_2$) produces stronger contribution (upper panel), compared to the one with small universal magnitude (lower panel).}
    \label{fig:illu}
\end{figure}


Following the above interpretations, the norm $\|\bbeta(\bv)\|_2$ here is the overall effect that tells us whether the voxel $\bv$ is an important predictor for the outcome in general or not, and thus, we put a shrinkage prior on it for selecting the important voxels.
If it is 0, then none of the components in $\bD_i(\bv)$ have any predictive importance, and thus we should not consider that voxel (region) in our predictive model.
Furthermore, variable selection for each of the individual coefficient functions may not be possible due to the dependence among the components in $\bD_i(\bv)$.
However, it is evident that the set of predictors at each voxel naturally forms a group.
Hence, a reasonable model for our regression problem can be formulated using group-LASSO \citep{yuan2006model}.
However, the spatial dependence among the predictors prevent the usage of traditional group-LASSO.
Hence, we propose a novel spatial variable selection prior for $\bbeta(\cdot)$ through a soft-thresholding transformation on the $\ell_2$ norm of a latent $q$-dimensional spatially varying function $\widetilde{\bbeta}(\cdot)$
and the latent spatially varying function is assumed to follow a Gaussian process with a scaled squared exponential kernel \citep{van2009adaptive}.

{In general, sparse estimation is an important statistical problem, and there are considerable numbers of competing methods. 
In a frequentist setting, different types of regularization or penalty functions are used for sparse estimation \citep{lasso,fused}. 
On the other hand, in a Bayesian framework, several types of sparsity-inducing priors have been introduced such as the traditional spike and slab prior \citep{spikeslab}, the horseshoe prior \citep{Carvalho}, normal-gamma prior \citep{Griffin}, double-Pareto prior \citep{Armagan} or Dirichlet-Laplace prior \citep{Bhattacharya}. However, these priors are not suitable for an imaging predictor when the regression effect is expected to be spatially smooth.


\cite{Goldsmith} and \cite{Li} proposed priors that account for spatial dependence and sparsity for scalar on image regression models separately. To a related problem, \cite{Wang} proposed a penalty based on the total variation. 
Spatial dependence, however, is still not fully incorporated in their approach. 
Recently, \cite{roy2021spatial} developed a prior that is suitable to estimate piece-wise sparse signals.
In \cite{kang2018scalar}, a soft-thresholded Gaussian process prior was proposed for a related problem. 
Soft-thresholding induced sparse estimation is widely popular in various application areas \citep{fan2001variable,han2006noise, rothman2009generalized,kusupati2020soft}.
Thresholding-based Bayesian modeling approaches are gaining increasing popularity \citep{ni2019bayesian, cai2020bayesian, wu2022bayesian,chen2023bayesian}.
Our modeling approach also relies on a novel soft-thresholding transformation.
Furthermore, we extend the model for a multi-group setting to efficiently characterize the group structure of the data.
For handling large imaging data, we consider low-rank approximations \cite{higdon,nychka2015multiresolution} and further develop a Metropolis-Hastings-based Markov chain Monte Carlo (MCMC) sampling algorithm.
We, too, consider using low-rank approximations for faster computation and in addition, implement a gradient-based Hamiltonian Monte Carlo (HMC) sampling algorithm \citep{neal2011mcmc} for efficient posterior computation.


Furthermore, we establish posterior contraction rates, which provide us with new insights for soft-thresholded estimates unlike \cite{kang2018scalar} which only proves consistency. 
Besides, there are several other aspects different from their paper, starting from the assumptions on the predictor process.
Specifically, our assumption is to control the observed correlation among the spatial points, that are uncorrelated under the true correlation. 
In the derivation of theoretical support, we also do not consider the thresholding parameter to be fixed, unlike \cite{kang2018scalar} where it was kept fixed in the theoretical part.
For image regression, the predictor process plays an important role.
Our assumptions on the predictor process are motivated by the developments in the sparse covariance matrix estimation theory \citep{bickel2008covariance, bickel2008regularized}.
Under those assumptions, we are able to establish strong consistency results with respect to the following distance $d(\bbeta,\bbeta_0)=\int_{\bv}\|\bbeta(\bv)-\bbeta_0(\bv)\|_2\pd\bv$. 
We also establish the spatial variable selection consistency results. 
For adaptive estimation, we consider the scaled square exponential kernel-based Gaussian process prior for our latent coefficient function $\widetilde{\bbeta}(\cdot)$ and borrow a few results from \cite{van2009adaptive}.

In summary, The major contributions include a) providing a soft-thresholded $\ell_2$-norm prior allowing variable selection for vector-value imaging predictors; b) establishing posterior contraction rates; c) developing a double-thresholded framework for multi-group inference without specification of the reference group. The rest of the article is organized as follows. In the next section, we present our proposed modeling framework and describe some of its properties. In Section~\ref{sec::theo}, we study the large sample properties of our proposed method. Detailed descriptions of the priors and the posterior computation steps are in Section~\ref{sec::comp}.  In Section~\ref{sec::simu}, we compare the performance of our proposed method to other competing methods. 
We present our ADNI data application in Section~\ref{sec::real}. Finally, in Section~\ref{sec::discussion}, we discuss possible extensions and other future directions.

\section{Modeling framework}
\label{sec::modeling}
In many real data applications for scientific and clinical investigations (e.g., ADNI), the data are collected from several heterogeneous groups.
We thus introduce our scalar on vector-valued image regression model for a multi-group dataset, described as follows.
Let, $y_i$ be the scalar response for $i$-$th$ subject and $\bD_i(\bv_j)$ be the $q$-dimensional spatially distributed imaging predictor which is collected at the $d$-dimensional spatial location $\bv_j$ where $j=1,\ldots,p$. Our goal is to identify important predictors $\bD_i(\bv_j)$ driving the scalar response $y_i$.
We assume that the set of spatial locations $\{\bv_j\}_{j=1}^p$ are from a compact closed region $\B\in\mathbb{R}^d$ to allow us to define a spatial process with indices of $\{\bv_j\}_{j=1}^p$.
Let us further assume that there are $G$ many groups. Let us further define the sets $\{\G_1,\ldots,\G_G\}$, where $\G_g$ stands for the set of subject indices that belong to the $g$-$th$ group.
To this end, we introduce our proposed scalar on vector-valued image regression model as,
\begin{align}
    y_{i}=&b_{0,g}+p^{-1/2}\sum_{j=1}^p\bD_{i}(\bv_j)^T\bbeta_g(\bv_j) + e_{i}, \textrm{ for } i\in g\label{eq: multimodel}\\
    e_{i}\sim&\Normal(0, \sigma^2)\nonumber
\end{align}
where $b_{0,g}$'s stand for group-specific intercepts, and $\bbeta_g(\bv_j)$'s are the group-specific spatially varying vector-valued regression coefficient to characterize effects of $\bD_i(\bv_j)$ on $y_i$ when $i$ belongs to $g$-$th$ group.
Here, $\bD_i(\bv_j)^T\bbeta_g(\bv_j)$ stands for the effect $\bv_j$-$th$ voxel for $g$-$th$ group.
The normalizing constant $p^{-1/2}$ is to scale down the total effects of massive imaging predictors following \cite{kang2018scalar}. 
We maintain group-specific separate intercepts $b_{0,g}$'s in the model.
This specification allows flexible interpretation of the coefficients without setting the reference group.
Returning to the decomposition, $\bD_i(\bv)^T\bbeta_g(\bv)=\|\bD_i(\bv)\|_2\|\bbeta_g(\bv)\|_2\cos(\theta_{\bD_i(\bv),\bbeta_g(\bv)}),$ we can see that the part $\|\bbeta_g(\bv)\|_2\cos(\theta_{\bD_i(\bv),\bbeta_g(\bv)})$ is the effect of $\bv$-$th$ spatial location for $i$-$th$ subject, and it varies with the data vector's direction.
Specifically, the part $\cos(\theta_{\bD_i(\bv),\bbeta_g(\bv)})$ controls the variation across the subjects, as well as the sign. 
For some individuals, the overall contribution may be positive and for some, it may be negative.
This may also be of clinical importance to know, for which individuals the effect of $\bv$-$th$ location is positive and for whom it is negative based on the estimated effect $\hat{\bbeta}_g(\bv)$. We discuss these directions in the following subsections.
The error variance $\sigma^2$ is kept the same for all the groups, which may be relaxed to set group-specific variances.
However, in this paper, we do not consider that.

Different from \textit{traditional} scalar on image regression where the predictor at each location is a scalar, our imaging predictor has vector-valued entries at each spatial location.
In the next subsection, we first introduce our proposed soft-thresholding transformation-based characterization to model sparse vector-valued spatially varying regression coefficients.
Building on that, we subsequently propose our prior model for $\bbeta_g(\cdot)$ in the multi-group case, suitable for estimating sparse and piece-wise smooth coefficients.

\subsection{Soft-Thresholded $\ell_2$ Norm transformation}
The regression coefficients $\bbeta_g$'s are vector-valued, and they are assumed to be piecewise smooth and sparse to accommodate the desired properties in high dimensional scalar on image regression.
We propose a soft-thresholding transformation based prior model for $\bbeta_g$. 

For simplicity, we first set $G=1$ and propose our soft-thresholding prior in a single-group case, with $\bbeta(\cdot)$ being the vector-valued regression coefficient.
We first define our soft-thresholding function for a $q$-dimensional vector $\bx$ as $h_{\lambda}(\bx)=\left(1-\frac{\lambda}{\|\bx\|_2}\right)_{+}\bx$, where if $x>0$ then $x_{+}=x$, otherwise $x_{+}=0$}. 
Since it puts the thresholding on the $\ell_2$ norm, we call this transformation Soft-Thresholded $\ell_2$ Norm (ST2N) transformation.
Figure~\ref{fig::ST2N} depicts how ST2N thresholds the original vector in a 2-dimensional case.
Finally, we model $\bbeta(\bv)$ as the ST2N transformation of a multivariate Gaussian process (GP) $\widetilde{\bbeta}(\bv)$ such that $\bbeta(\bv)=h_{\lambda}(\widetilde{\bbeta}(\bv))$.
Hence, our final prior for $\bbeta(\bv)$ is termed as ST2N-GP.

Here $Q_{\lambda}(\bv)=\Pi(\|\widetilde{\bbeta}(\bv)\|<\lambda)$ is the prior probability associated with the event $\bbeta(\bv)=0$. 
We thus have $Q_{\lambda_1}(\bv)<Q_{\lambda_2}(\bv)$ for $\lambda_1<\lambda_2$. 
Thus, greater thresholding induces greater prior mass at zero. 
Our motivating dataset further has a multi-group structure.
Hence, in the following subsection, we propose a novel double soft-thresholding transformation to accommodate multi-group data.
In the Lemma \ref{lem::lipch}, we show continuity of the proposed soft-thresholding transformation that ensures continuity of the transformed coefficient $\bbeta(\cdot)$ when the underlying latent function $\widetilde{\bbeta}(\cdot)$ is smooth.
Continuity is a desired property for $\bbeta(\cdot)$ as discussed before.

\begin{figure}
    \centering
    \includegraphics[width=100mm]{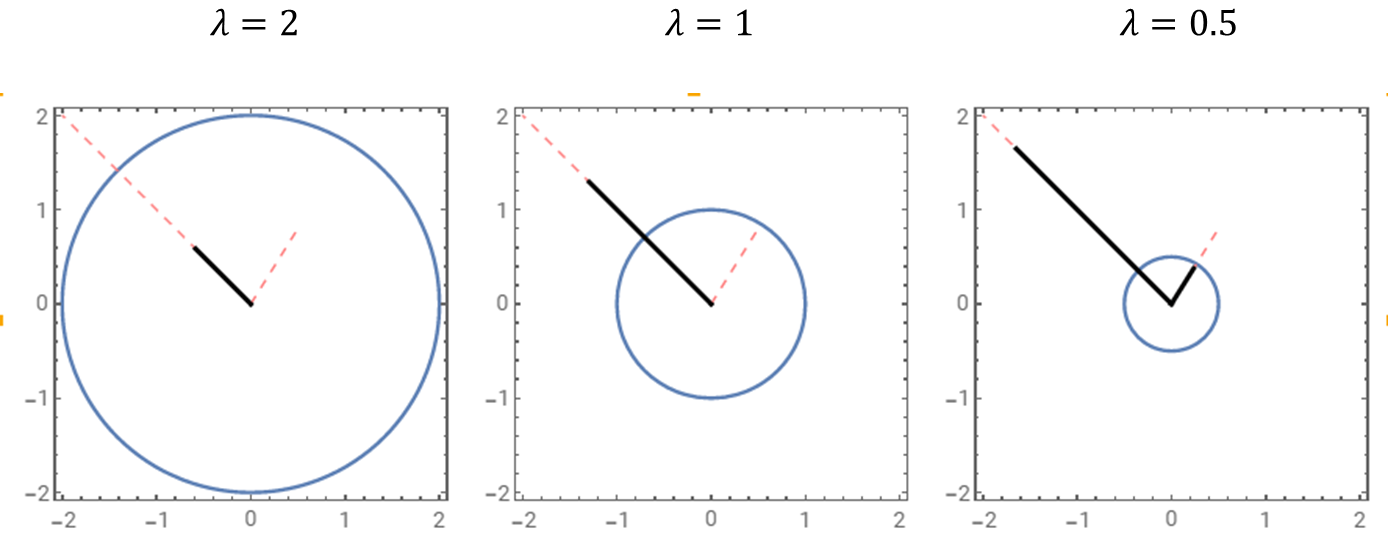}
    \caption{ST2N in action: The effect of the thresholding parameter $\lambda$ is illustrated for its three possible values ($\lambda=2,1,0.5$), and the vectors are in 2-D ($q=2$). The radius of the blue circle in each panel represents the value of $\lambda$. The two dashed red lines in each panel are typical examples of the original vectors $\bx$ and the solid black lines stand for their ST2N transformation, $h_{\lambda}(\bx)$. If the norm of original vector larger than $\lambda$ (exceed the circle), it keeps its original direction with reduced norm; otherwise, it is then zero vector.
    }
    \label{fig::ST2N}
\end{figure}
\begin{Lem}[Lipschitz Continuity]
The ST2N transformation is Lipschitz continuous.
\label{lem::lipch}
\end{Lem}
The proof is in the Supplementary Materials. We show Lipschitz continuity, both under $\ell_2$ and $\ell_{\infty}$ norms.
By construction, ST2N transformation only shrinks the magnitude but does not alter the direction (positive or negative). 
We thus have $\|h_{\lambda}(\bx)\|_2=(\|\bx\|_2-\lambda)_{+}$, but $\frac{h_{\lambda}(\bx)}{\|h_{\lambda}(\bx)\|_2}=\frac{\bx}{\|\bx\|_2}$ whenever $\|\bx\|_2>\lambda$.

\subsection{Double soft-thresholding for multi-group}
\label{sec: DSoftthres}

Relying on our proposed ST2N transformation, we now propose our prior model for $\bbeta_g(\cdot)$.
Our inferential objective here is to identify the spatial locations having an important effect on the scalar response.
For different groups, the set of important spatial locations can be different.
While modeling multi-group data, one of the groups is set as a reference in typical clinical studies.
However, the choice of a reference group may not be easy in certain cases.
For example, if the groups represent different stages of a less-studied disease, it is then ambiguous to select which one should be the reference group.
Our real data application also represents such a scenario, where we consider three different stages of cognitive impairment as different clinical groups.
In the below paragraph, we introduce the details of the proposed multi-group model which is motivated to address these above-mentioned issues. 

A straightforward model for $\bbeta_g(\cdot)$ using ST2N characterization is by letting, $\bbeta_g(\cdot)=h_{\lambda}(\widetilde{\bbeta}_g(\cdot))$ where we maintain separate group specific latent processes.
Here, maintaining the same thresholding parameter does not dispose of any flexibility, as the thresholding parameter, $\lambda$ and the latent smooth coefficient function $\widetilde{\bbeta}_g(\bv)$ are not separately identifiable, as discussed in Lemma 2 of \cite{kang2018scalar}.
Furthermore, a subset of those locations may turn out to be equally important for all the groups.
Hence, we decompose the latent process $\widetilde{\bbeta}_g(\cdot)$ into \textit{shared} and \textit{group-specific} parts as $\widetilde{\bbeta}_g(\cdot)=\widetilde{\bbeta}(\cdot) +\balpha_g(\cdot)$,  $\widetilde{\bbeta}(\cdot)$ is the shared part and $\balpha_g(\cdot)$ is the group-specific part which is again assumed to be sparse. 
The above decomposition into shared and group-specific components is not identifiable. 
However, the shared part facilitates centering the latent processes $\widetilde{\bbeta}_g(\bv)$'s and, thus, improves computational performance. Then the sparse, piece-wise smooth spatially varying coefficient $\balpha_g(\cdot)$ captures the deviation of the group $g$ from the centered effect.
We thus further apply ST2N characterization to $\balpha_g(\cdot)$'s. 

Although the latent process and the thresholding parameter are not separately identifiable as discussed above, different thresholding parameters induce different levels of shrinkage.
It brings more flexibility to the model fitting.
Hence, we set group-specific separate thresholding parameters to allow different levels of sparsity in $\balpha_g(\bv)$'s.
To summarize, we propose the following model for our multi-group setting,
\begin{align}
    \bbeta_g(\bv) =&h_{\lambda}(\widetilde{\bbeta}_g(\bv)) = h_{\lambda}(\widetilde{\bbeta}(\bv) +\balpha_g(\bv))\nonumber\\
     \balpha_g(\bv) =& h_{\lambda_g}(\widetilde{\balpha}_g(\bv)),\label{eq::multiheirar}
\end{align}
where $\widetilde{\balpha}_g(\cdot)$ stands for the latent process which is thresholded using $\lambda_g$ to obtain $\balpha_g(\cdot)$.
For notational clarity, we represent the latent non-regularized processes with a tilde.
It is evident that the shared and group-specific components of $\widetilde{\bbeta}_g(\cdot)$ may not directly be associated with similarities and differences in $\bbeta_g(\cdot)$'s. 
We thus explore the connections between them.
In addition, we define appropriate characterization of the \textit{similar} effect in the context of scalar on vector-valued image regression in Definition S1 of the Supplementary Materials.
For vector-valued predictors, the regression coefficient $\bbeta_g(\bv)$ can be interpreted both in terms of its magnitude $\|\bbeta_g(\bv)\|_2$ and direction $\frac{\bbeta_g(\bv)}{\|\bbeta_g(\bv)\|_2}$.
For two groups $g\neq g'$, the scalar projection of $\bbeta_g(\bv)$ onto $\bbeta_g'(\bv)$ is given by $\bbeta_g(\bv)^T\bbeta_g'(\bv)/\|\bbeta_g'(\bv)\|_2$.
If $\bbeta_g(\bv)^T\bbeta_g'(\bv)>0$, the direction of this projection is along $\bbeta_g'(\bv)$. 
However, it becomes opposite to $\bbeta_g'(\bv)$, when $\bbeta_g(\bv)^T\bbeta_g'(\bv)<0$.
Following this argument, we discuss some explorations and results in Section S1 of the Supplementary Materials. 

\section{Theoretical support}
\label{sec::theo}

We establish posterior consistency results in the asymptotic regime of increasing sample size $n$ and the increasing number of spatial locations $p$.
We keep the dimension of imaging predictor $\bD(\bv)$ as $q$.
Let $\sigma_0$, $b_{00}$ and $\bbeta_0(\bv)$ be the null values of $\sigma$, $b_{0}$ and $\bbeta(\bv)$, respectively.
We assume that $b_{00}$ is known to maintain simplicity in our calculations.
We assume that $\sigma$ is bounded between $(b1,b_2)$. This assumption can be disposed of by directly working with the negative log-affinity divergence \cite{ning2020bayesian}, but for simplicity in the arguments, we present the proof assuming the additional boundedness as in \cite{roy2020high}.
Following the suggestion of a reviewer, it is possible to obtain such bounds on variance in a data-driven way. 
In our model, we may set the marginal variance of the response as the upper bound, and the lower bound can be some user-specified small number or even zero. 
Now, let us define the following distance metric between $\bbeta$ and $\bbeta_0$ as $d(\bbeta,\bbeta_0)=\int_{\bv}\|\bbeta(\bv)-\bbeta_0(\bv)\|_2\pd\bv$. 
This is an appropriate distance metric to study posterior consistency for our inference problem. 
In the next subsection, we discuss our assumptions on the design locations and predictor process in order to establish posterior convergence results with respect to $d(\bbeta,\bbeta_0)$.
Our results are first presented for a single group setting with $G=1$.
Subsequently, we discuss the possible implications that will hold for the multi-group.

\noindent{\it Notations:} The notations ``$\lesssim$'' and ``$\gtrsim$" stand for inequality up to constant multiple. Let, for a spatially varying function $a(\cdot)$, we define $\|a\|_2=\sum_{\bv\in\V} a^2(\bv)$ and $\|a\|_{\infty}=\sup_{\bv}|a(\bv)|$. The sign $\otimes$ stands for Kronecker's product. 








\subsection{Design locations}
To show posterior consistency with respect to the functional $\ell_2$ distance, we need to approximate an empirical $\ell_2$ distance to an integral.
Thus, we put the following assumption for quantifying the error.
\begin{Assmp}[Design points]
Here, the spatial locations are $d$-dimensional, and we make the following assumptions.
\begin{itemize}
    \item[(1)] Let $\B$ be a hyper cuboid as $\B=[0,1]^d$, where $d$ is the dimension of $\bv$, then we can obtain a disjoint partition $\B=\cup_{j=1}^{p_n}\V_j$ such that $\bv_j$ is the midpoint or the center of $\V_j$ and volume of $\V_j=1/p_n$ for all $j$.
    \item[(2)] $\sup_{\bv,\bv'\in\V_j} |v_{k}-v'_{k}|^d\asymp 1/p_n$ for all $k=1,\ldots,d$ and $j=1,\ldots,p_n$ where $\bv=\{v_{i}\}_{i=1}^d$ and $\bv'=\{v'_{i}\}_{i=1}^d$.
    \item[(3)] $p_n\asymp n^d$.
\end{itemize}
\label{assmp::designpoint}
\end{Assmp}
The condition on $p_n$ is similar to \cite{kang2018scalar}.
The upper bound on $p_n$ is utilized in the next subsection while discussing the appropriate assumption on the predictor process.
 It may be relaxed to $Q_3 e^{q_n}\geq p_n \geq Q_4 n^d$ for some constant $Q_3,Q_4>0$ such that $q_n=o(n)$.
 However, with this upper bound, the validity of Assumption~\ref{assmp::predictor} might be problematic.
 To avoid such issues, we consider the simpler condition as described in (3) of Assumption~\ref{assmp::designpoint}.
The lower bound is required while establishing the contraction rate with respect to $d(\bbeta,\bbeta_0)$ through the following result.

\begin{Lem}
Under Assumption~\ref{assmp::designpoint}, for a function $f$ over $\B$ with uniformly bounded first derivative i.e. $\|f'\|_{\infty}<F$, we have $|\int_{\bv\in\B} f(\bv)\pd\bv- \frac{1}{p_n}\sum_{j=1}^{p_n}f(\bv_j)|\leq C\frac{F}{p_n^{1/d}}$  for some constant $C>0$.
\end{Lem}
\begin{proof}
As volume of $\V_j$ is $\int_{\bv\in\V_j}\pd\bv=1/p_n$, we have $\int_{\bv\in\B} f(\bv)\pd\bv- \frac{1}{p_n}\sum_{j=1}^{p_n}f(\bv_j)=\sum_{j=1}^{p_n}\int_{\bv\in\V_j} f(\bv)\pd\bv- \sum_{j=1}^{p_n}f(\bv_j)\int_{\bv\in\V_j}\pd\bv$. Now $|f(\bv)-f(\bv_j)|\leq \sum_{k=1}^d |f'(cv_k+(1-c)v_{j,k})||v_k-v_{j,k}| \leq C_1d \frac{F}{p_n^{1/d}}$ due to multivariable mean value theorem for some $c\in[0,1]$ and by assumption $|v_k-v_{j,k}|\leq C_1 p_n^{-1/d}$ for some constant $C_1$.

Thus, $|\sum_{j=1}^{p_n}\int_{\bv\in\V_j} f(\bv)\pd\bv- \sum_{j=1}^{p_n}f(\bv_j)\int_{\bv\in\V_j}d\bv| \leq \sum_{j=1}^{p_n}\int_{\bv\in\V_j}|f(_j\bv)-f(\bv_j)|\pd\bv\leq \sum_{j=1}^{p_n}d \frac{F}{p_n^{1/d}}\int_{\bv\in\V_j} \pd\bv =C_1d \frac{F}{p_n^{1/d}}$. Hence, the result holds for $C=C_1d$.
\end{proof}

Now we start to show the required assumptions for the predictors. To study consistency properties for scalar on image regressions, it is common to consider the predictors to be realizations from a spatially dependent process \citep{Wang, kang2018scalar}. The DTI data is indeed spatially dependent, at least in the large white matter regions of the brain \citep{goodlett2009group, zhu2011fadtts, lan2022geostatistical}. 
 Let $\D_{i,k}=\{\bD_{i,k}(\bv_j)\}_{j=1}^p$ be the vector of data for $i$-$th$ individual for $k$-$th$ direction where $k=1,\ldots,q$. Define $n\times (qp)$ dimensional design matrix such that its $i$-$th$ row will be $\bD_i'^T=\{\D_{i,1},\ldots,\D_{i,q}\}$. Let expectation $\eE(\bD_i')=0$ and variance $V(\bD_i')=\bK$ with $0<c_{\min}\leq\eig(\bK)\leq c_{\max}<\infty$ and $K_{\ell,\ell}<\infty$. This variance assumption can be satisfied easily for $\D_{i,k}$'s being stationary ergodic processes
with spectral densities bounded between $c_{\min}$ and $c_{\max}$ or a ``spike-model" as argued in \cite{bickel2008regularized, bickel2008covariance}. Let $J_n=\frac{1}{n}\sum_i\bD_i'\bD_i'^T$ and thresholding operator $T_s(\bA) = (\!(a_{i,j}\bone\{|a_{ij}| \geq s\})\!)$. 

 \begin{Assmp}[Predictor process]
 There exists $N$ such that for $n>N$ we have $\big\|J_n-T_{t_n}\left(J_n\right)\big\|_{op}\leq c_{\min}/4$ for some $t_n\rightarrow 0$ where $0<c_{\min}\leq\eig(\bK)\leq c_{\max}<\infty$.
 \label{assmp::predictor}
\end{Assmp}
In Section S2 of the Supplementary Materials, we discuss the motivation behind this assumption and provide more details on $t_n$. Thus, we finally have, $\big\|J_n-\bK\big\|_{op}<\frac{c_{\min}}{2}\textrm{ for } n>N$. We know that $\|\bA\|_{op}\leq\max_i\sum_{j=1}^{p_n}|a_{i,j}|$. 
Thus, a sufficient condition for the above assumption to hold is that the number of non-zero entries with absolute value less than $t_n$ at each row of $J_n$ is upper bounded by $\frac{c_{\min}}{4t_n}(\rightarrow \infty$ as $t_n\rightarrow 0)$. 
However, the assumption would still hold as long as the total absolute contributions of the cross-correlation values of the thresholded coordinates at each row of $J_n$ are bounded by $c_{\min}/2$. Implicitly, the above assumption requires spatially uncorrelated voxels to have significantly small sample correlation values. 
This allowed us to ensure that the eigenvalues of the design matrix are bounded away from zero and infinity, thereby helping us to establish the posterior contraction rate both for empirical-$\ell_2$ and the usual $\ell_2$ which is the integrated square distance in function space. The justification of this additional assumption is motivated by the covariance matrix estimation results from \cite{bickel2008covariance}. However, the method is robust to this assumption as long as the boundedness assumption on the eigenvalues of $J_n$ holds.

\subsection{Large support and posterior consistency}
The large support result of \cite{kang2018scalar} relies on constructing a $\widetilde{\bbeta_0}(\bv)$ such that $\bbeta_0(\bv)=h_{\lambda}(\widetilde{\bbeta_0}(\bv))$.
However, such construction can be difficult in our setting due to the complexity of our thresholding function. 
We thus take a different approach.
Furthermore, \cite{kang2018scalar} set $\lambda$ fixed for the theoretical study.
We, however, let $\lambda=\lambda_n$ to vary with $n$ as well as the proportion of non-zero locations having in $\|\bbeta_0\|$.
We now list all the required conditions to facilitate the theoretical
results.

\begin{Assmp}[Coefficient function] For each $j$, the coefficient functions $\beta_{0,j}\in\C^{\alpha}[0,1]^d$, which stands for the space of H\"older smooth functions with regularity $\alpha$. 
Note that there is no discontinuous jump while moving from the zero region to the non-zero.
Further $\|\bbeta_0\|_{\infty}\leq M<\infty$ and let us define the set $R_0=\{\bv: \|\bbeta_0(\bv)\|_2=0\}$.
\label{assmp::smoothness}
\end{Assmp}

\begin{Assmp}[Prior on coefficient] We set $\widetilde{\bbeta}(\cdot)\sim$GP$(0,\kappa_a\otimes\bSigma)$ where GP$(0,\kappa_a\otimes\bSigma)$ is a multivariate Gaussian process with marginal covariance kernel for each component being $\kappa_a$ and cross-component covariance matrix $\bSigma$. The covariance kernel $\kappa_a(\bv,\bv')=\exp(-a^2\|\bv-\bv'\|_2^2)$ and $a^d$ follows gamma distribution with density $g(t)\propto t^{s/d}\exp(-D_2t^d)$. 
\label{assmp::kernel}
\end{Assmp}

Due to the above assumption, the covariance of the latent Gaussian process is assumed to be separable, and marginal covariance $\widetilde{\bbeta}(\bv)$ at each voxel is $\bSigma$. The space of H\"older smooth functions considered above consists of those functions on $[0,1]^d$ having continuous mixed partial derivatives up to order $\underline{\alpha}$ and such that the $\underline{\alpha}$-$th$ partial derivatives are H\"older continuous with exponent $\alpha-\underline{\alpha}$.
The gamma density on $a^d$ is a simplified version of the prior class considered in \cite{van2009adaptive} and similar to \cite{yang2016bayesian}. For simplicity, we set $\bSigma=\bI_q$ in our calculations.

\begin{Thm}[Large sup-norm support]
For any H\"older smooth function $\bbeta_0(\cdot)$ and a positive constant $\epsilon>0$, the ST2N induced
prior $\bbeta\sim$ST2N-GP$(\lambda,\bSigma,\kappa)$ with $\lambda\sim\Unif(0,R)$ satisfies $\Pi(\big\|\|\bbeta-\bbeta_0\|_2\big\|_{\infty}<\epsilon)>0$ and $\Pi(\big\|\|\bbeta\|-\|\bbeta_0\|_2\big\|_{\infty}<\epsilon)>0$. Furthermore, $\Pi(\|\bbeta_0-\bbeta\|_{\infty}<\epsilon)>0$ 
\label{thm::largesup}
\end{Thm}

The uniform prior in the above theorem is well-motivated both theoretically and computationally.
In the next Theorem, we show that optimal $\lambda_n$ decreases with $n$, and thus, for a large enough $n$, the optimal $\lambda$'s are all bounded by some constant. 
Also computationally, it helps to design a simpler posterior sampling method for $\lambda$.
For the remainder of the theoretical results, let us denote the data as $S_n=\{y_i,\bD_i\}_{i=1}^n$.

{\underline {Definition \citep{ghosal2017fundamentals}}:} The posterior contraction rate at the true parameter $\bbeta_0\in\mathcal{A}$ with respect to the semi-metric $d$ on $\mathcal{A}$ is a sequence $\epsilon_n\to 0$ such that $\Pi (\bbeta: d(\bbeta,\bbeta_0)>C \epsilon_n | S_n )\to 0$ in $P_{\bbeta_0}^{(n)}$-probability for some large constant $C$, where $\mathcal{A}$ denotes the parameter space of $\bbeta_0$.

\begin{Thm}[Consistency for the single group]
Under the Assumptions~\ref{assmp::designpoint} to \ref{assmp::kernel}, the posterior contraction rate with respect to $d(\bbeta,\bbeta_0)$ is $\epsilon_n=n^{-\alpha/(2\alpha+d)}(\log_{} n)^{\nu'}$, where $\nu'=(4\alpha+d)/(4\alpha+2d)$ with $\lambda_n = O(\epsilon_n)$.
\label{thm::consis}
\end{Thm}

\begin{Rem}
In our proof of Theorem~\ref{thm::consis}, we also establish that for all $\epsilon>0$ that $\Pi(\frac{1}{p_n}\sum_{j=1}^{p_n}\|\bbeta_0(\bv_j)-\bbeta(\bv_j)\|^2_{2}>\epsilon\mid S_n)\rightarrow 0$ and $\Pi(\frac{1}{p_n}\sum_{j=1}^{p_n}\bigm|\|\bbeta_0(\bv_j)\|_2-\|\bbeta(\bv_j)\|_2\bigm|>\epsilon\mid S_n)\rightarrow 0$ as $n\rightarrow \infty$.
\end{Rem}

\begin{Rem}
Although scalar on image regression models share commonalities with the linear regression models, the recovery of the regression coefficient in the linear regression setting is quantified in terms of the $\ell_2$ distance without the fraction $\frac{1}{p_n}$.
Our results are thus not directly comparable with the ones from the linear regression literature.
In the context of function estimation, $d(\bbeta,\bbeta_0)$ is an appropriate metric to quantify the recovery.
Due to this difference in the notion of recovery, we require conditions that are different from the linear regression situation as well.
\end{Rem}

To show selection consistency, we define $U_n(\bbeta)=\{j:\|\bbeta(\bv_j)\|_2>0,\|\bbeta_0(\bv_j)\|_2=0\}$,  $V_n(\bbeta)=\{j:\|\bbeta(\bv_j)\|_2=0,\|\bbeta_0(\bv_j)\|_2>0\}$,
$W_n(\bbeta)=\{j:\|\bbeta(\bv_j)\|_2>0,\|\bbeta_0(\bv_j)\|_2>0\}$, 
$W_n'(\bbeta)=\{j:\bbeta(\bv_j)^T\bbeta_0(\bv_j)>0,\|\bbeta_0(\bv_j)\|_2>0\}$, and further let $U(\bbeta)=\{\bv:\|\bbeta(\bv)\|_2>0,\|\bbeta_0(\bv)\|_2=0\}$, $V=\{\bv:\|\bbeta(\bv)\|_2=0,\|\bbeta_0(\bv)\|_2>0\}$,
$W(\bbeta)=\{\bv:\|\bbeta(\bv)\|_2>0,\|\bbeta_0(\bv)\|_2>0\}$, and
$W'(\bbeta)=\{\bv:\bbeta(\bv)^T\bbeta_0(\bv)>0,\|\bbeta_0(\bv)\|_2>0\}$.
The two sets $W_n'(\bbeta)$ and $W'(\bbeta)$ are especially important for vector-valued coefficients. 
They ensure consistency of the direction of effects.
Additionally, let $R_{0,n}=\{j:\|\bbeta_0(\bv_j)\|_2=0\}$.

\begin{Thm}[Empirical sparsity]
Under the conditions of the previous theorem, for any, $\epsilon>0$ we have 
$\Pi(|U_n(\bbeta)|/|R_{0,n}|<\epsilon|S_n)\rightarrow 1$ and $\Pi(|V_n(\bbeta)|/(p_n-|R_{0,n}|)<\epsilon|S_n)\rightarrow 1$ as $n\rightarrow\infty$ if $\frac{|R_{0,n}|}{|p_n|}\rightarrow p_0$ such that and $1>p_0>0$. Furthermore, $\Pi(|W_n'(\bbeta)|/(p_n-|R_{0,n}|)>1-\epsilon|S_n)\rightarrow 1$
\label{thm::empselect}
\end{Thm}

The empirical sparsity result shows that the proportion of locations in the observed data where zero locations of $\beta$ and $\beta_0$ do not match approaches to zero with probability 1. The assumption $\frac{|R_{0,n}|}{|p_n|}\rightarrow p_0<1$ and $>0$ is not unrealistic, since we are in an infill asymptotic regime. As we have more spatial observations, we should approach to the true proportion of zero locations. The above result readily implies that $\Pi(|W_n(\bbeta)|/(p_n-|R_{0,n}|)>1-\epsilon\mid S_n)\rightarrow 1$. We further establish the following result.
Let $\A(U)=\int_{\bv\in U}\pd\bv$, the area of the set $U$.
\begin{Thm}[Sparsity]
Under the conditions of the previous theorem, for any, $\epsilon>0$ we have 
$\Pi(\A(U(\bbeta))>\epsilon\mid S_n)\rightarrow 0$, $\Pi(\A(V(\bbeta))>\epsilon\mid S_n)\rightarrow 0$, $\Pi(\A(W(\bbeta))>1-\epsilon\mid S_n)\rightarrow 1$ and $\Pi(\A(W'(\bbeta))>1-\epsilon\mid S_n)\rightarrow 1$ as $n\rightarrow\infty$.
\label{thm::select}
\end{Thm}

The above result establishes selection consistency under the same set of conditions as in the previous theorem.
Specifically, it generalizes Theorem~\ref{thm::empselect} for the whole spatial domain and implies that the area where zero locations of $\beta$ and $\beta_0$ do not match approaches to zero with probability 1.



\subsection{Result for multi-group}
\label{sec::theomulti}
As the size of each group increases, all the large sample properties established under a single group setting will continue to hold for each group.
We thus focus on some additional results to aid our multi-group inference.
Apart from estimating the $\bbeta_g$'s, we are also interested in identifying the similar-effect spatial locations, defined in Definition S1 in the Supplementary Materials.
We thus focus to establish on some additional results concerning this definition.
Along with Assumptions~\ref{assmp::designpoint} to \ref{assmp::kernel}, we require the following additional assumption.
\begin{Assmp}[Norm integrablity]
$\int\|\bbeta_{0,g}(\bv)\|_2\pd\bv<M'$ for some $M'>0$ and for all $g$.
\label{assmp::multi}
\end{Assmp}
\begin{Thm}[Cross-group consistency]
Under Assumptions~\ref{assmp::designpoint} to \ref{assmp::multi}, as $n_g\rightarrow \infty$ we have $\Pi(\int|\bbeta_g(\bv)^T\bbeta_{g'}(\bv)-\bbeta_{0,g}(\bv)^T\bbeta_{0,g'}(\bv)|\pd\bv<\epsilon\mid S_n)\rightarrow 1$ for all $g,g'=1,\ldots,G$.
\label{thm::multiconsis}
\end{Thm}

The proof is in the Supplementary Materials. This result is particularly important to establish the following theorem. The following result shows selection consistency in identifying the spatial locations where the sign of inner products between cross-group effects are positive.

Let $Q_1(\bbeta_g, \bbeta_{g'})=\{\bv: \bbeta_g(\bv)^T\bbeta_{g'}(\bv) > 0, \bbeta_{0,g}(\bv)^T\bbeta_{0,g'}(\bv)>0\}$, $Q_2(\bbeta_g, \bbeta_{g'})=\{\bv: \bbeta_g(\bv)^T\bbeta_{g'}(\bv) > 0, \bbeta_{0,g}(\bv)^T\bbeta_{0,g'}(\bv)\leq 0\}$.
\begin{Thm}[Identification]
Under Assumptions~\ref{assmp::designpoint} to \ref{assmp::multi}, we have $\Pi(\A(Q_1(\bbeta_g, \bbeta_{g'})) > 1-\epsilon \mid S_n)\rightarrow 1$, $\Pi(\A(Q_2(\bbeta_g, \bbeta_{g'})) <\epsilon \mid S_n)\rightarrow 0$ as $n_g\rightarrow \infty$ for all $g,g'=1,\ldots,G$.
\label{thm::multiiden}
\end{Thm}
The proof is provided in the Supplementary Materials. The approaches are similar to the single group case.

\section{Prior specification and posterior computation}
\label{sec::comp}

We assume the error variance and the intercepts to be known in our theoretical analysis (Section~\ref{sec::theo}).
In practical applications, however, these cannot be ascertained.
Thus, we put priors on these parameters as well. We set $\sigma^{-2}\sim\Ga(c_1,c_2)$, where $\Ga$ stands for the Gamma distribution with shape and rate parameters as $c_1$ and $c_2$.
In Section~\ref{sec::theo}, we imposed a boundedness condition on $\sigma$ and proposed possible bounds. However, no appreciable changes in the results are observed on imposing those bounds in our numerical implementation.
Thus, we ignore those bounds for simplicity.
For the intercepts, we set $b_{0,g}\sim\Normal(0, \sigma_b^2)$.
We let $\lambda,\lambda_g\sim\Unif(0,R)$ for some large enough $R>0$ and $\Unif$ stands for the Uniform distribution.
While pre-specifying $R$, we take a data-driven approach. We first fit a traditional LASSO regression model for the complete data using {\tt cv.glmnet} function from R package {\tt glmnet} \citep{friedman2010regularization} and the set $R=q\lambda_{\min}$, where $\lambda_{\min}$ is the cross-validated penalty parameter from the LASSO regression model and $q$ is the dimension of the predictor at each voxel.
As discussed in Section~\ref{sec::theo}, we let $\widetilde{\bbeta}(\cdot)\sim$ GP$(0,\kappa_a\otimes\bSigma)$.
For $\kappa_a$, we consider low-rank GP, thereby reducing the computational cost.
Specifically, we let a low-rank approximation of the spatially varying coefficients as $\widetilde{\beta}_j(\bv)=\sum_{\ell=1}^LF(\bv-\bv_{\ell})\widetilde{\beta}'_j(\bv_{\ell})$ for all $j=1,\ldots,q$ and then let the coefficients $\widetilde{\bbeta}'(\cdot)=(\widetilde{\beta}_1'(\cdot),\ldots,\widetilde{\beta}_q'(\cdot))\sim$GP$(0,\kappa_a\otimes\bSigma)$, where $\kappa_a(\bv,\bv')=\exp(-a^2\|\bv-\bv'\|_2^2)$ and $a^d\sim\Ga(d_1,d_2)$.
Here $\bv_1,\ldots,\bv_L \in \mathbb{R}^d$ are a grid of spatial knots covering $\B$ and $F$ is a local kernel. 
Our predictor images do not have to be on a regular grid.
In the case of regularly gridded 3-dimensional data, we set $n_1/2 \times n_2/2 \times n_3/2$ equally spaced grid of knots while covering a $n_1 \times n_2 \times n_3$ dimensional spatial data. This strategy also aligns with previous literature \cite{higdon,nychka,kang2018scalar, roy2021spatial}. In our simulation experiments, we tried a few other choices of knots, however, the above specification worked the best. 
For our real data application, the predictors are white matter fibers. Thus, the spatial coordinates do not form a hyperrectangle. However, we may still form spatial knots based on the smallest hyperrectangle that covers the domain of the data. Consequently, some of these knots may fall outside the spatial domain of data. In that case, we consider ignoring those knot points without sacrificing any numerical accuracy.
Following \cite{kang2018scalar}, we consider tapered Gaussian kernels with bandwidth $b$ such that $F(x)=\exp\left(-\frac{x^2}{2b^2}\right)I\{x<3b\}$. Hence, $F(\|\bv-\bv_{\ell}\|)=0$ for all $\bv$ that are separated by more than $3b$ from $\bv_{\ell}$.
The approximation helps to reduce the computational burden without sacrificing any significant loss in estimation accuracy. Similar approximations are also applied to the group-specific $\widetilde{\balpha}_g$'s.
Specifically, we let $\tilde{\balpha}_g(\cdot)\sim$GP$(0,\kappa_{a_g}\otimes\bSigma_g)$
We further put a conjugate inverse Wishart prior on $\bSigma$ and $\bSigma_g$'s with parameters $\nu$ and $\bS$. For all of our numerical results, we set $\nu=4$ and $\bS=\bI_3$ as $q=3$.

The inference is based on samples drawn from the posterior using an MCMC algorithm. 
The variance $\sigma^{-2}$ is updated from its full conditional Gamma posterior distribution. 
The intercepts, $b_{0,g}$'s are also updated from their full conditional Gaussian posterior distributions.
The posterior samplings for $\bSigma$ and $\bSigma_g$'s are done from their full conditional inverse Wishart posteriors.
For all the other parameters, there is no conjugacy.
The latent coefficients $\widetilde{\bbeta}$ and $\widetilde{\balpha}_g$'s are updated using a gradient-based Hamiltonian Monte Carlo (HMC) algorithm \citep{neal2011mcmc, betancourt2015hamiltonian, betancourt2017conceptual}.
HMC has been shown to draw posterior samples much more efficiently than traditional random walk Metropolis-Hastings in complex Bayesian hierarchical models \citep{betancourt2015hamiltonian} by more efficiently exploring the target distribution under local correlations among the parameters.
We set the leapfrog step in HMC to 30, but periodically tune the step-length parameter to ensure a pre-specified level of acceptance.
The rest of the parameters, $a, a_g,\lambda$ and $\lambda_g$ are updated using Metropolis-Hastings steps. 
Additional details are in Section S6 of the Supplementary Materials.
For posterior inference, We collect 5000 MCMC samples after 5000 burn-in samples.

\section{Simulation study}
\label{sec::simu}

We carry out two simulations to evaluate the performance of the proposed method against other competing scalar on image regression models.
They are illustrated as two cases.
The main difference is in the associated predictor process.
The common data generation scheme of the two cases is described as follows.
In general, we generate 3-dimensional image predictors $\bD_i(\bv)=\{D_{i,1}(\bv),D_{i,2}(\bv),D_{i,3}(\bv)\}$ such that $\bv\in\B$ and $\B$ is $20\times 20$ equispaced grid on $[0,1]^2$ with $q=3$.
Hence, we have a total of $p=20\times 20 = 400$ spatial locations. All of our results are based on 50 replications.
However, they are designed to study all the properties essential for our DTI application.
The simulation of case 1 mimics diffusion tensor imaging where each voxel is associated with a vector of diffusion direction (a unit vector), characterizing a tissue's anatomical structure. The simulation of case 2 mimics a more generic case. 
Furthermore, we assume that there are $G=3$ many groups. We set the group sizes to 50 or 100 to evaluate the effect of varying sample sizes.
We also vary the error variance $\sigma^2$ from 1, 5 to 10 for each simulation setting.
The true regression coefficients $\bbeta_g$'s are kept the same for both of the two simulation cases.
We first write $\bbeta_g(\bv)=r_g(\bv)\etam(\bv)$, where $r_g(\bv)=\|\bbeta_g(\bv)\|_2$ and $\etam(\bv)$ is the unit vector for voxel $\bv$.
Hence, 
The unit vectors are not varied with the group, but the magnitudes are.
Figure~\ref{fig::simumag} illustrates the magnitudes for different groups.

\begin{figure}[htbp]
\centering
\subfigure{\includegraphics[width = 0.32\textwidth]{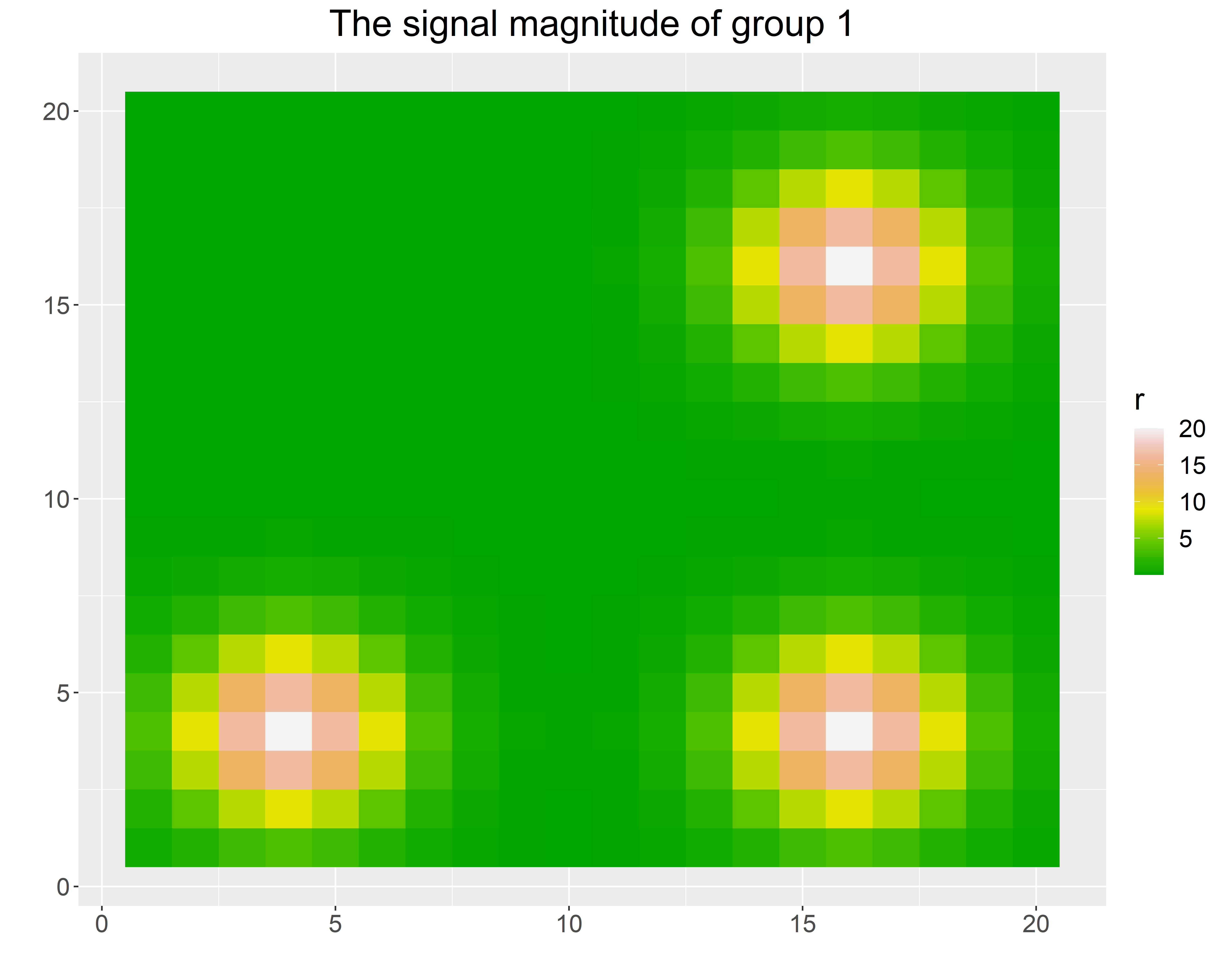}}
\subfigure{\includegraphics[width = 0.32\textwidth]{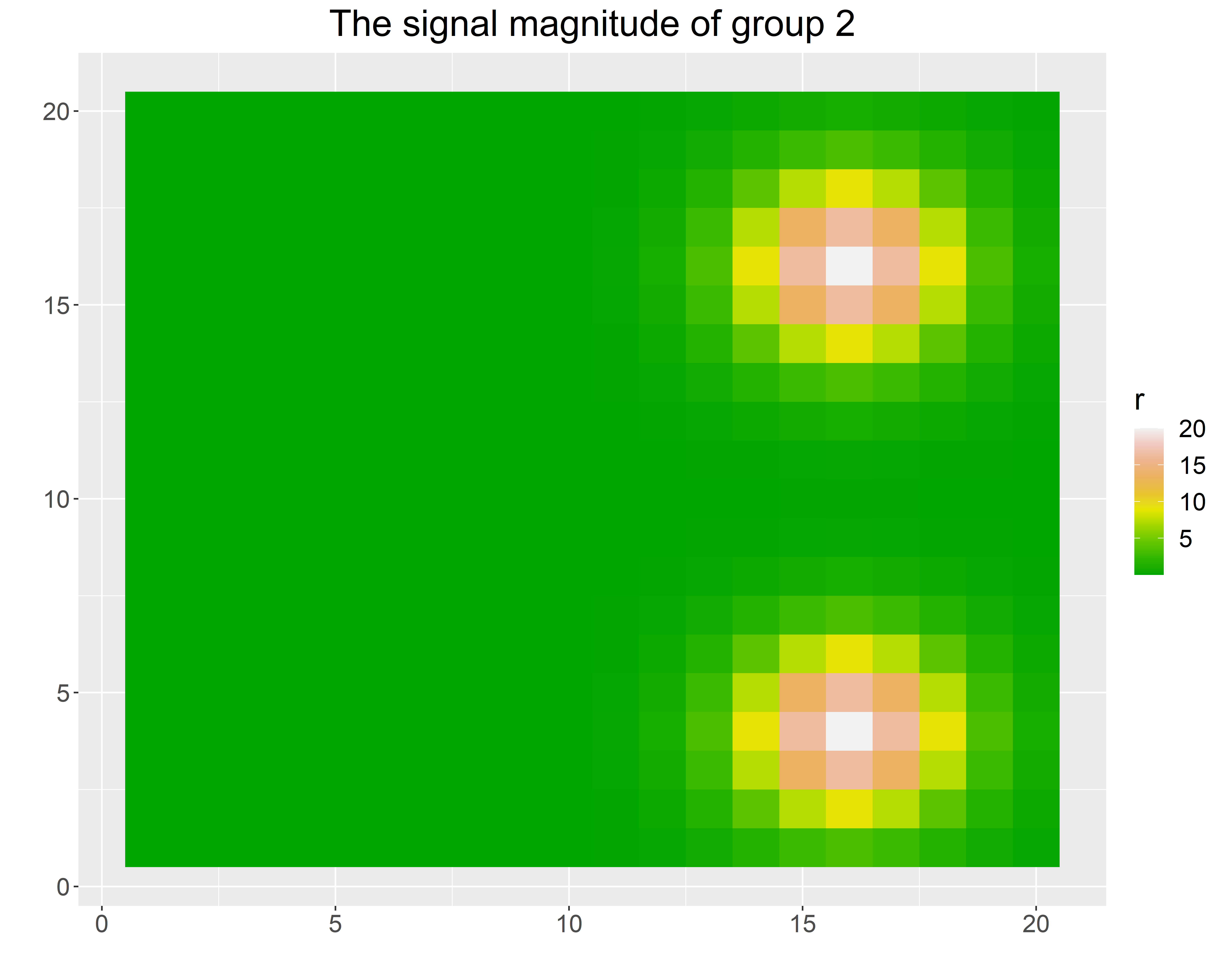}}
\subfigure{\includegraphics[width = 0.32\textwidth]{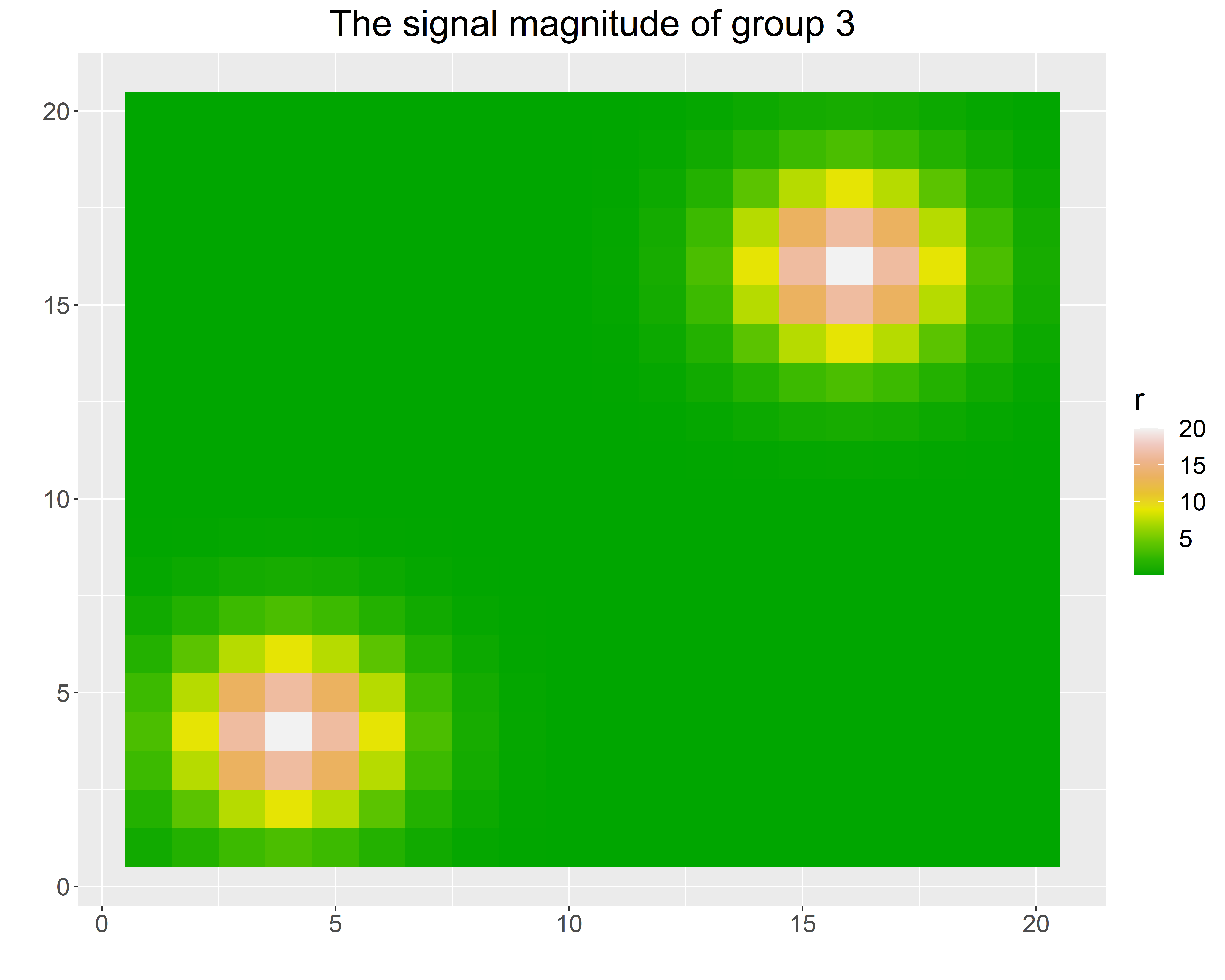}}
\caption{The magnitude $r_{g}$'s for different groups.}
\label{fig::simumag}
\end{figure}

For our proposed Bayesian method, the choices of hyperparameters are $c_1=c_2=d_1=d_2=0.1$, $\sigma_b^2=10^2, \nu=4, \bS=\bI_3$ and $R=5$. Except for $R$, all the other choices of hyperparameters lead to weakly informative priors.
The above choice of $R$ seems to work well both in simulations and real data applications. 
For all the Bayesian methods (i.e., ST2N, and STGP), we collect, 10000 MCMC samples and consider the last 5000 as post-burn-in samples for inferences. 
The other competing methods are not designed for vector-valued image predictors, nor in a multi-group setting.
Hence, we consider $D_{i,1}(\bv)$, $D_{i,2}(\bv)$ and $D_{i,3}(\bv)$ as three separate image predictors and fit the models for each group separately.
We also compare our method with STGP, LASSO, and functional principal component analysis (fPCA) \citep{jones1992displaying} based on empirical mean integrated squared error (EMISE) = $\frac{1}{3p}\sum_{j=1}^p\sum_{g=1}^3\|\hat{\bbeta}_g(\bv_j)-\bbeta_{0,g}(\bv_j)\|_2^2$.
The fPCA estimates are obtained as follows. After smoothing the images using {\tt fbps} function of {\tt refund} package \citep{refund}, the eigendecomposition of the sample covariance is computed. After that, lasso regularized principal component regression is performed. The leading eigenvectors that explain 99\% of the variation in
the sample images are used to get the final estimate.
To fit the LASSO penalized regression model, we applied the {\tt glmnet} package \citep{friedman2017package}.
We also compare with group-LASSO and PING \citep{roy2021spatial}, which are not presented as their performances are similar to LASSO and STGP respectively.

\subsection{Simulation case 1: General vector-valued predictors}
\label{sec::simucase1}
In simulation case 1, the datasets are generated such that they closely mimic our real dataset on the principal diffusion direction of diffusion tensor imaging.
The principal diffusion directions are unit vectors, and thus we rely on the von Mises-Fisher distribution to generate such data. 
We use von Mises-Fisher (vMF) distribution to generate random diffusion tensors. We use $\textrm{vMF}(\nu, \etam_{\bD})$ to denote a vMF distribution with the concentration parameter $\nu$ and the main direction $\etam_{\bD}$. Specifically, we generate $\bD_i(\bv)$ as $\bD_i(\bv)\sim \textrm{vMF}(30, \etam_{\bD}(\bv))$ using R package {\tt Rfast} \citep{Rfast} for each subject $i$ and each spatial location $\bv$ on the above-mentioned grid.
We first generate mean fiber directions $\etam_{\bD}(\bv)$'s as illustrated in panel (a) of Figure~\ref{fig::simudir}.
These mean directions maintain a spatial dependency pattern.
Table~\ref{tab::f1} compares the estimated mean square error as defined earlier across different methods and Figure~\ref{fig::esticoef} illustrates the norms of estimated regression coefficients as well as the similar-effect locations using the F-values, defined below Lemma S3 of the Supplementary Materials.
All the methods perform better when the sample size increases or the error variance decreases. 
In summary, our proposed ST2N-GP performs the best in comparison to the other methods. 
The poorest performance of LASSO is probably because it fails to incorporate spatial dependence.

\begin{figure}[htbp]
		\centering
		\subfigure[The mean direction $(\etam_{D}(\bv))$ to generate $\bD_i(\bv)$]{\includegraphics[width=60mm]{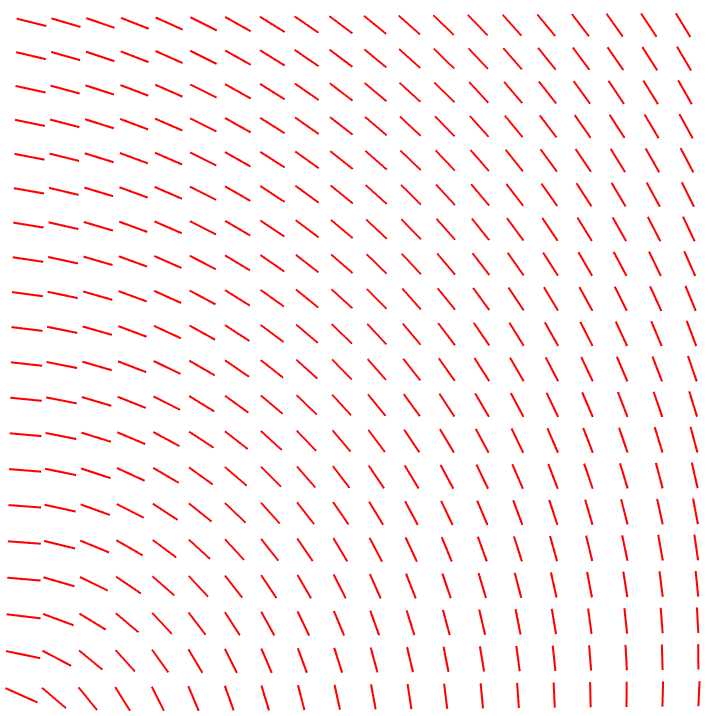}}
		\subfigure[Direction $\etam(\bv)$'s while constructing $\bbeta_{0,g}$'s]{\includegraphics[width=60mm]{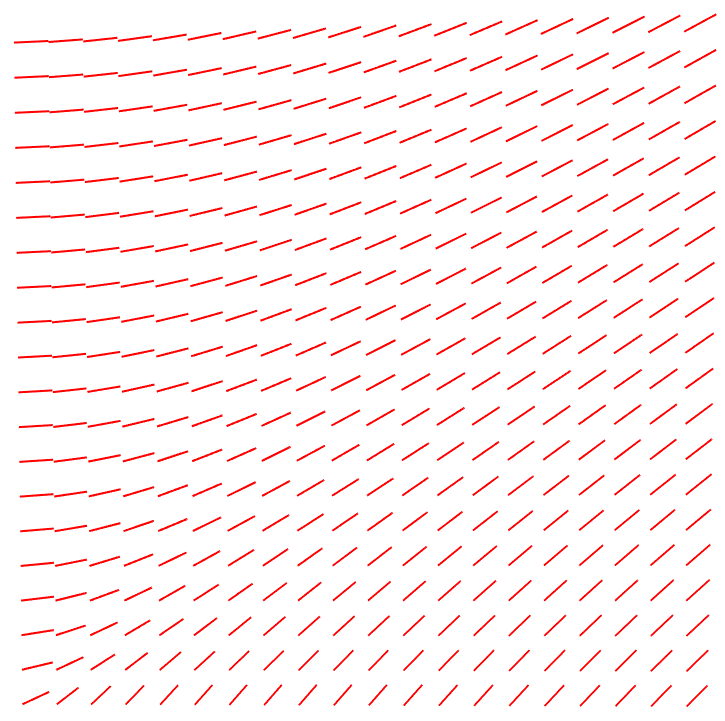}}
		\caption{The above shown directional unit vectors are generated while sampling the predictors $\bD_i(\bv)$'s and constructing the coefficients $\bbeta_{0,g}$'s.}
		\label{fig::simudir}
	\end{figure}
	
\begin{table}[htbp]
    \centering
    \caption{EMISE of the vector-valued slope $\bbeta(\bv)$ when the predictors are generated as described in simulation case 1 with varying sample size and true error variance. There are a total of three groups with sample size as mentioned in the first column. The results are based on 50 replications.} 
     \begin{tabular}{|r|r|rrrr|}
                \hline
                \multicolumn{2}{l|}{} &\multicolumn{4}{l}{Fitted Model} \\ \hline
                Group-specific &Error  & ST2N-GP  & STGP&LASSO&FPCA\\ 
                sample size &variance&&&&\\
                \hline
                
              &1 & 0.62 & 0.89 & 6.85 & 1.54 \\ 
  50&5 & 0.88 & 1.19 & 6.20 & 1.42 \\ 
  &10 & 1.35 & 1.42 & 7.13 & 1.56 \\ 
                
                \hline
                
               & 1 & 0.28 & 0.44 & 4.27 & 1.21 \\ 
  100&5 & 0.37 & 0.98 & 4.36 & 1.33 \\ 
  &10 & 0.81 & 1.03 & 6.96 & 1.37 \\ 
                \hline
            \end{tabular}
    \label{tab::f1}
\end{table}

\begin{figure}[htbp]
\centering
\subfigure{\includegraphics[width = 0.3\textwidth]{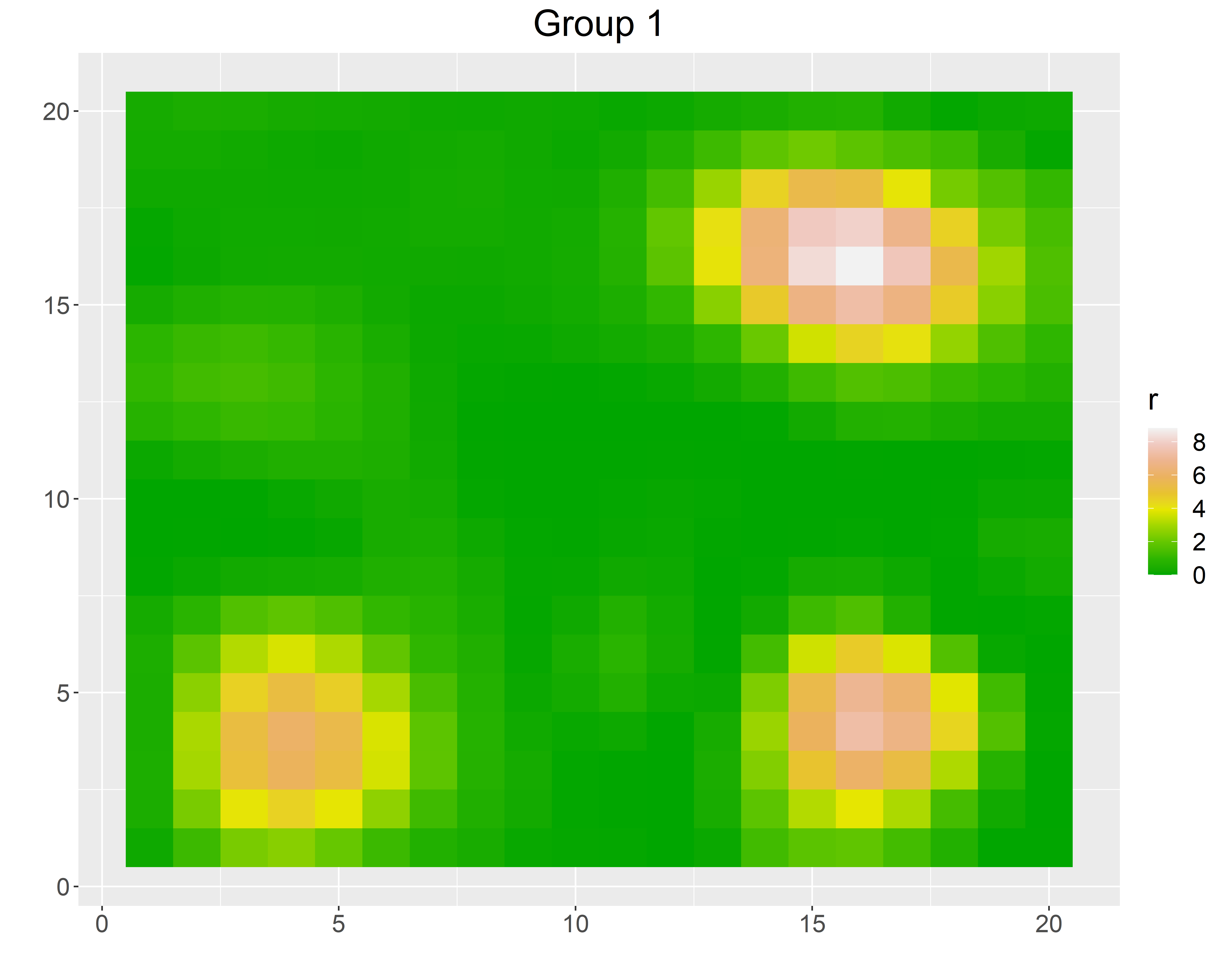}}
\subfigure{\includegraphics[width = 0.3\textwidth]{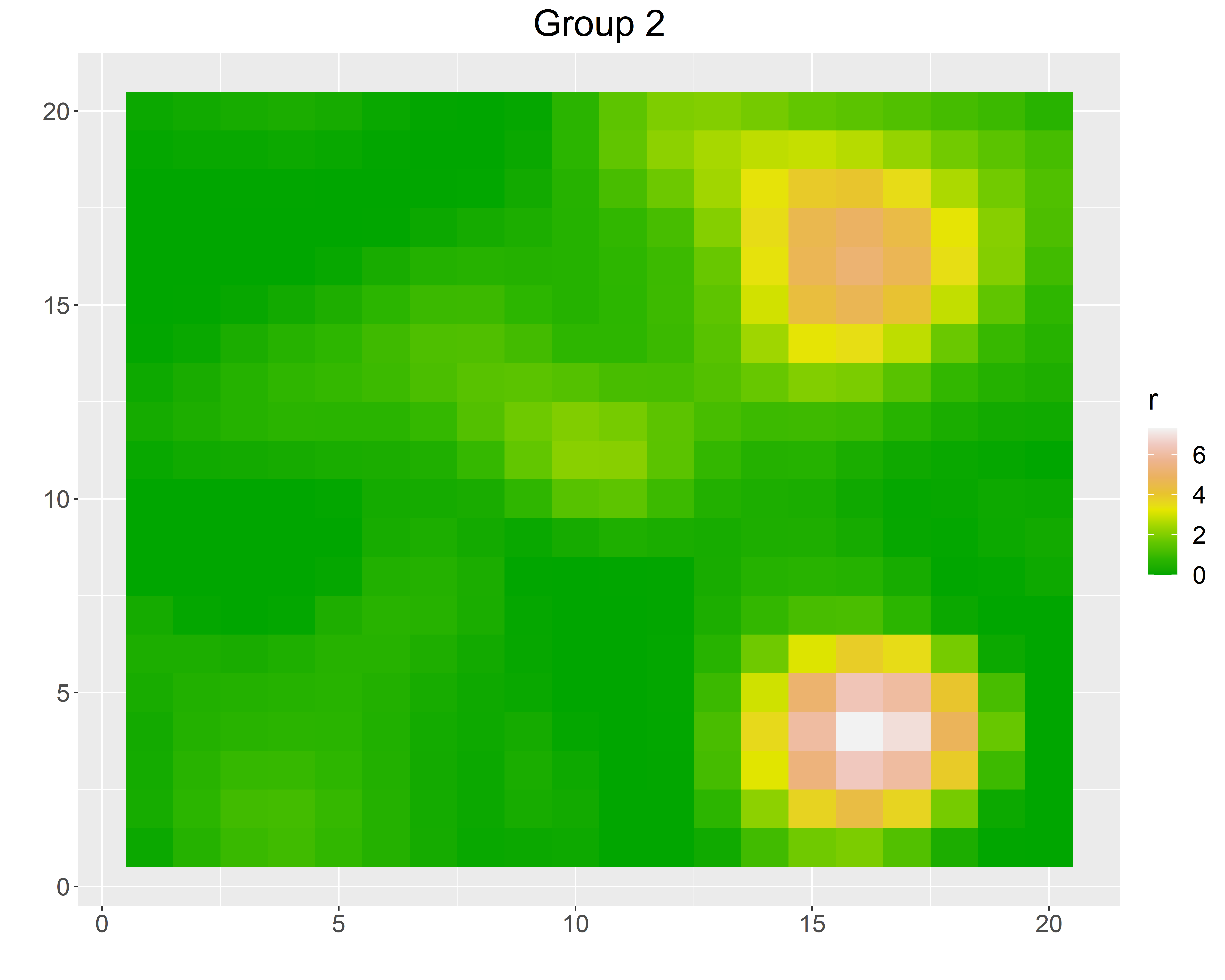}}
\subfigure{\includegraphics[width = 0.3\textwidth]{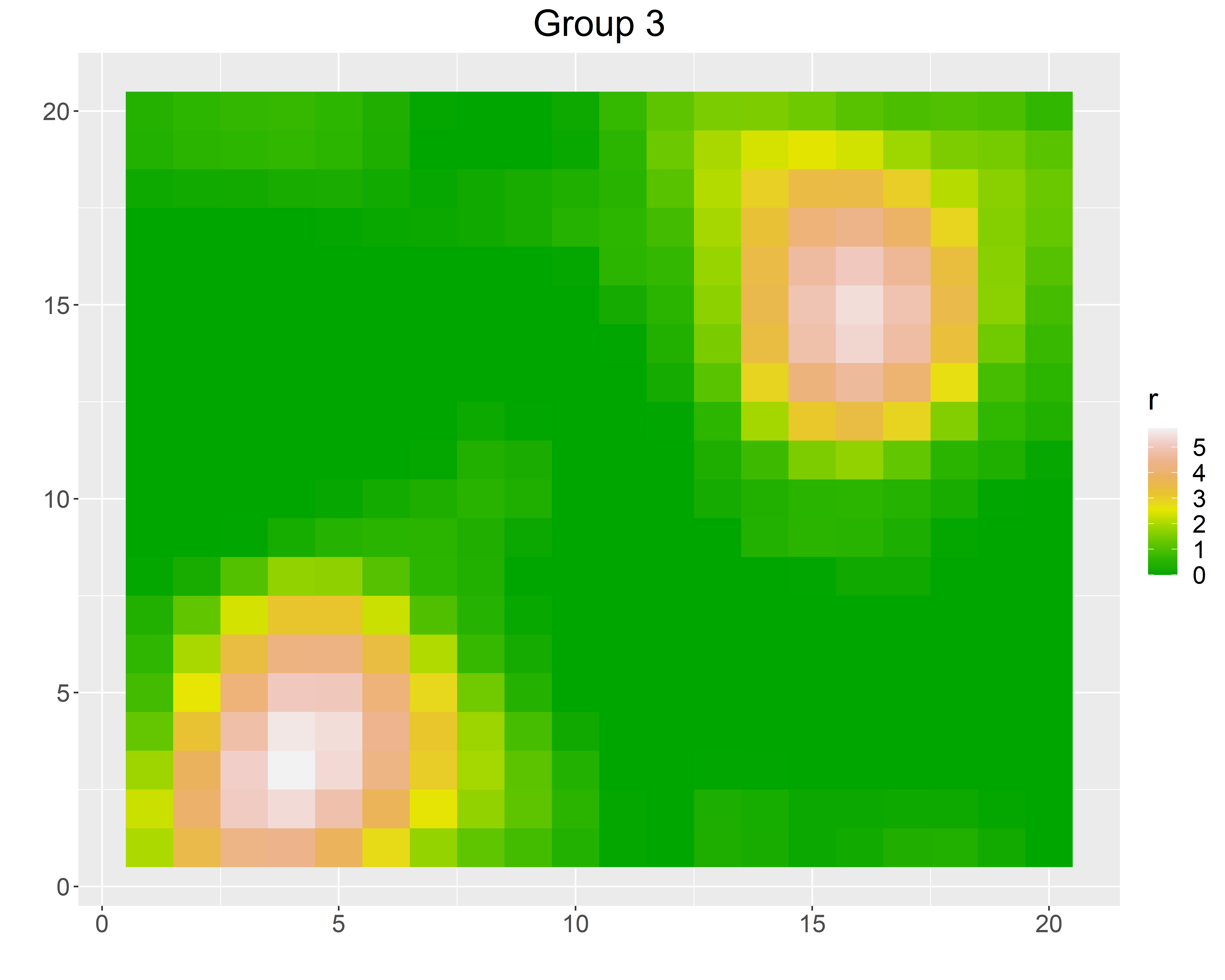}}
\subfigure{\includegraphics[width = 0.3\textwidth]{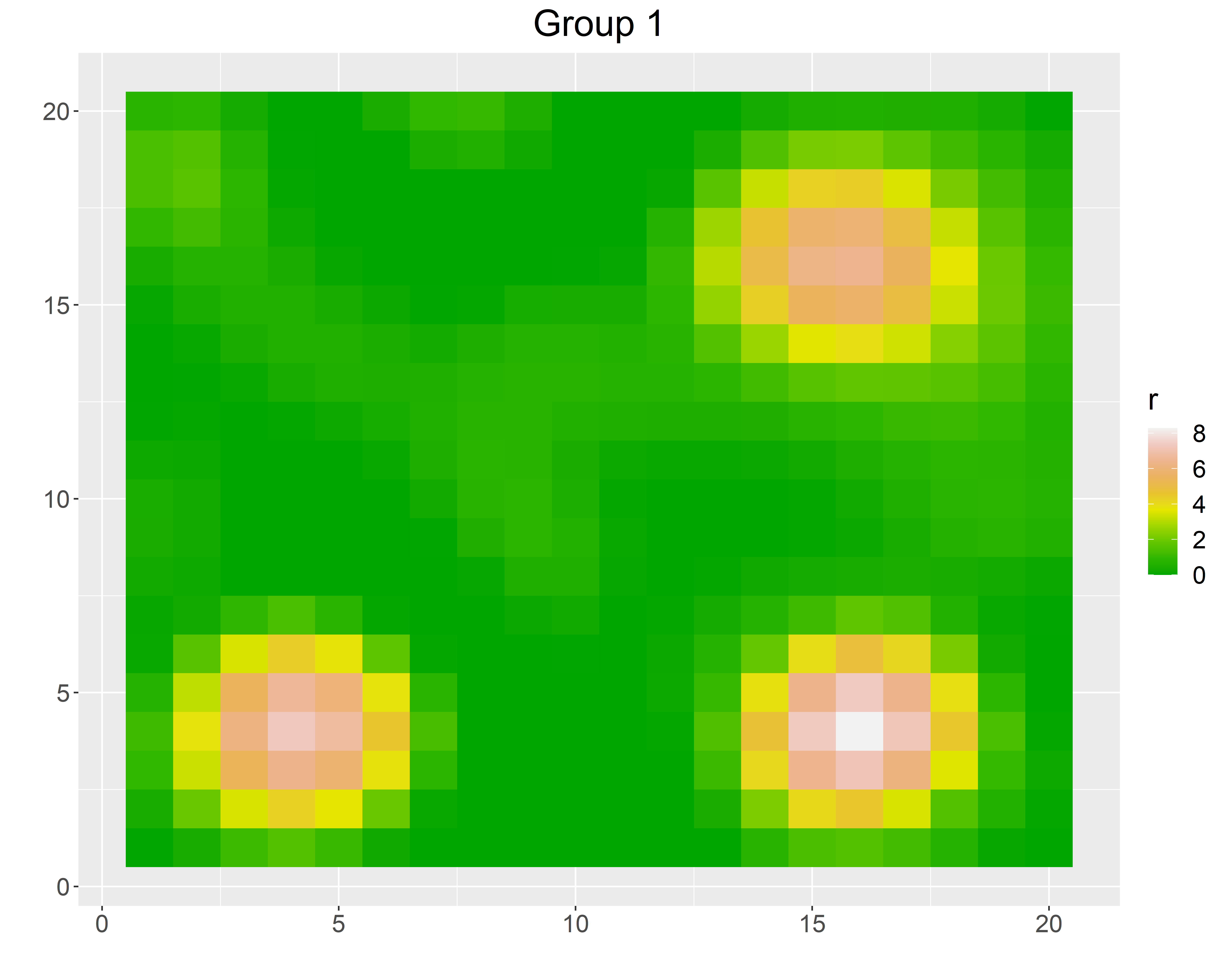}}
\subfigure{\includegraphics[width = 0.3\textwidth]{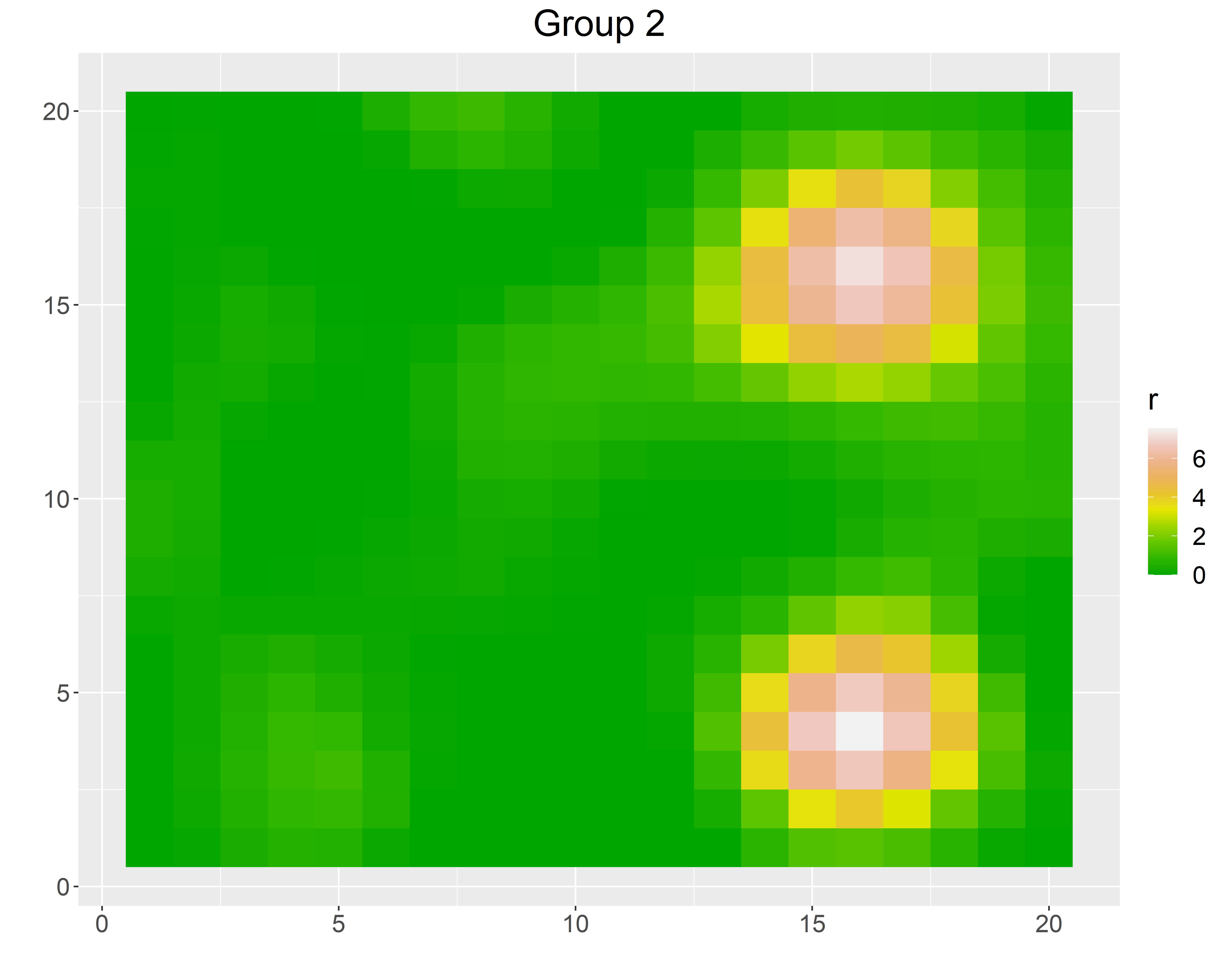}}
\subfigure{\includegraphics[width = 0.3\textwidth]{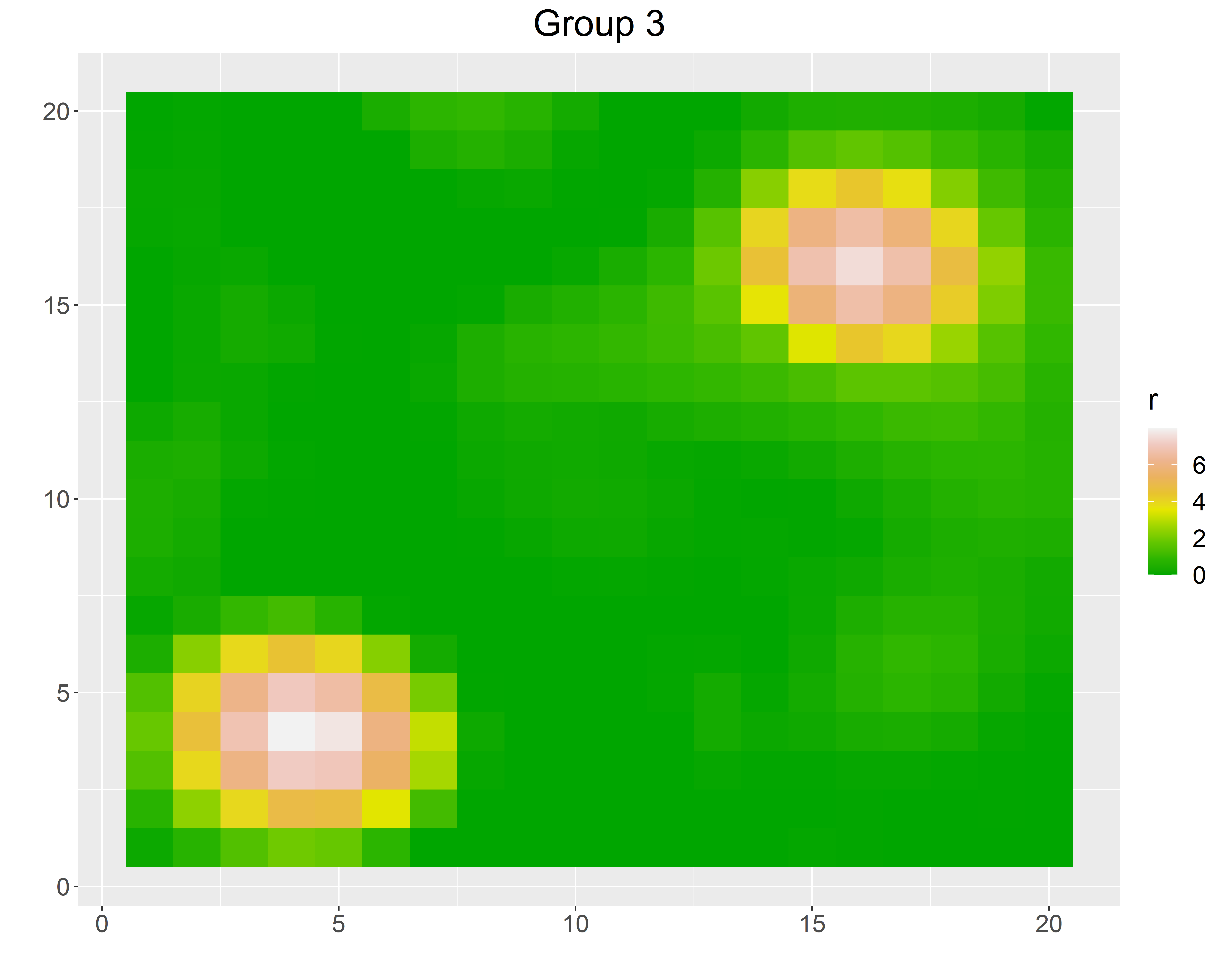}}
\caption{Estimated $\|\bbeta_g\|_2$'s when the error variance is $1$ for two sample sizes. The first row corresponds to a sample size of 50 in each group, and the second row is for 100. }
\label{fig::esticoef}
\end{figure}

\subsection{Simulation case 2: Directional-valued predictors}
\label{sec::simucase2}
The simulated datasets $\bD_i(\bv)$'s for this case are generated as follows.
We first generate $\bD'_i(\bv)=\{D_{i,1}'(\bv),D_{i,2}'(\bv),D_{i,3}'(\bv)\}$, following independent Gaussian processes with mean $\eE(D_{i,1}'(\bv))=\eE(D_{i,2}'(\bv))=\eE(D_{i,3}'(\bv))=0$ and exponential covariance kernels with different band widths, specifically, $\eE(D_{i,1}'(\bv)D_{i,1}'(\bv'))=\exp(-\|\bv-\bv'\|_2/3)$, $\eE(D_{i,2}'(\bv)D_{i,2}'(\bv'))=\exp(-\|\bv-\bv'\|_2/5)$, and $\eE(D_{i,3}'(\bv)D_{i,3}'(\bv'))=\exp(-\|\bv-\bv'\|_2/7)$.
Finally, we generate a positive definite matrix $\bPsi$ and set $\bD_i(\bv)=\bPsi\bD_i'(\bv)$. 
Thus, marginally at each $\bv$, we have $V\{\bD_i(\bv)\}=\bPsi\bPsi^T$.


\begin{table}[htbp]
    \centering
    \caption{EMISE of the vector-valued slope $\bbeta(\bv)$ when the predictors are generated as described in simulation case 2 with varying sample size and true error variance. There are a total of three groups with sample size as mentioned in the first column. The results are based on 50 replications.} 
     \begin{tabular}{|r|r|rrrr|}
                \hline
                \multicolumn{2}{l|}{} &\multicolumn{4}{l}{Fitted Model} \\ \hline
                Group-specific &Error  & ST2N-GP  & STGP&LASSO&FPCA\\ 
                sample size &variance&&&&\\
                \hline
                
                  &1 & 0.11 & 0.28 & 10.57 & 0.80 \\ 
  50&5 & 0.12 & 0.34 & 10.45 & 0.81 \\ 
  &10 & 0.18 & 0.39 & 10.46 & 0.81 \\ 
                
                \hline
                
                &1 & 0.05 & 0.12 & 2.69 & 0.77 \\ 
  100&5 & 0.05 & 0.12 & 2.70 & 0.77 \\ 
  &10 & 0.05 & 0.14 & 2.78 & 0.81 \\ 
                \hline
            \end{tabular}
    \label{tab::f2}
\end{table}

\begin{figure}[htbp]
\centering
\subfigure{\includegraphics[width = 0.3\textwidth]{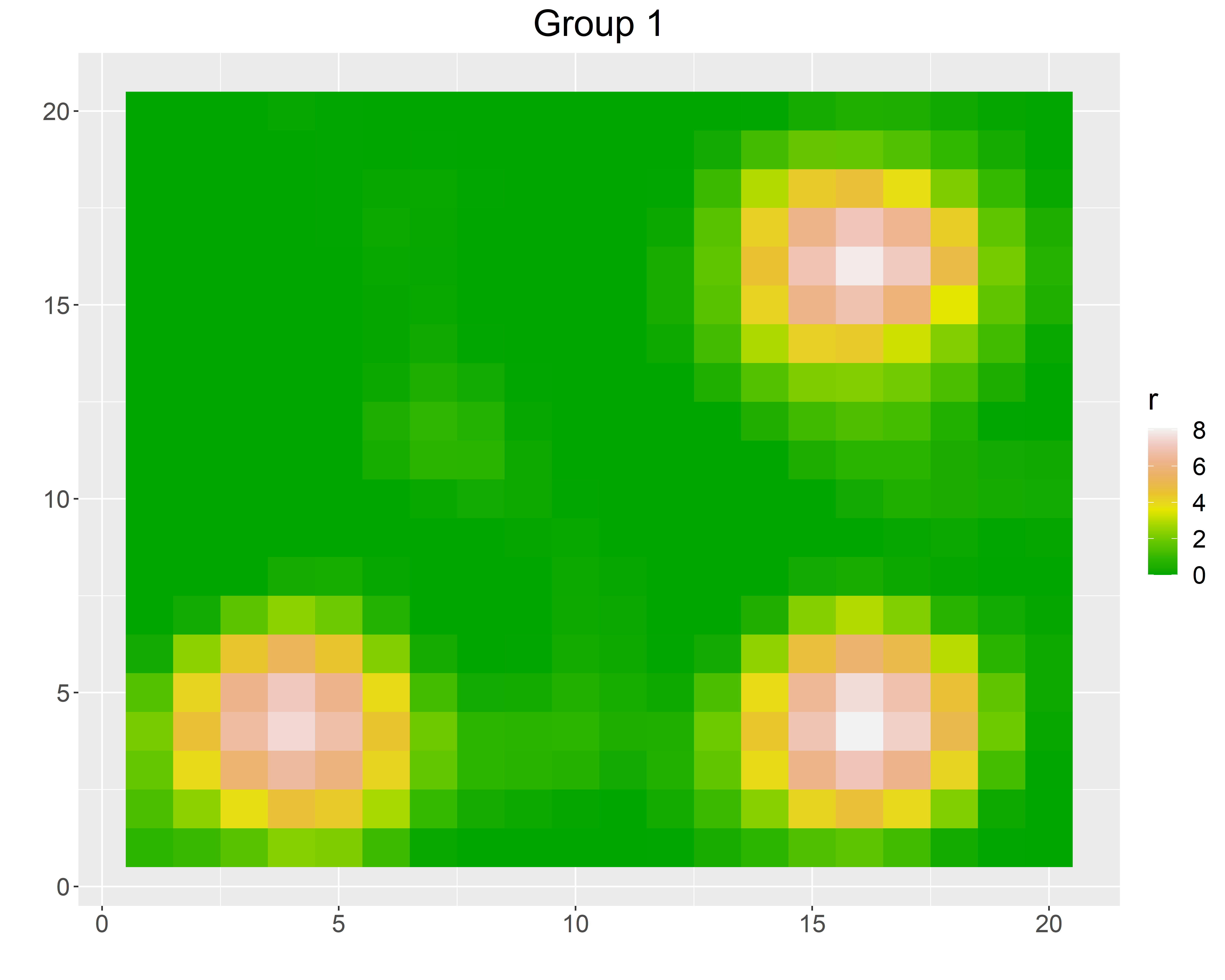}}
\subfigure{\includegraphics[width = 0.3\textwidth]{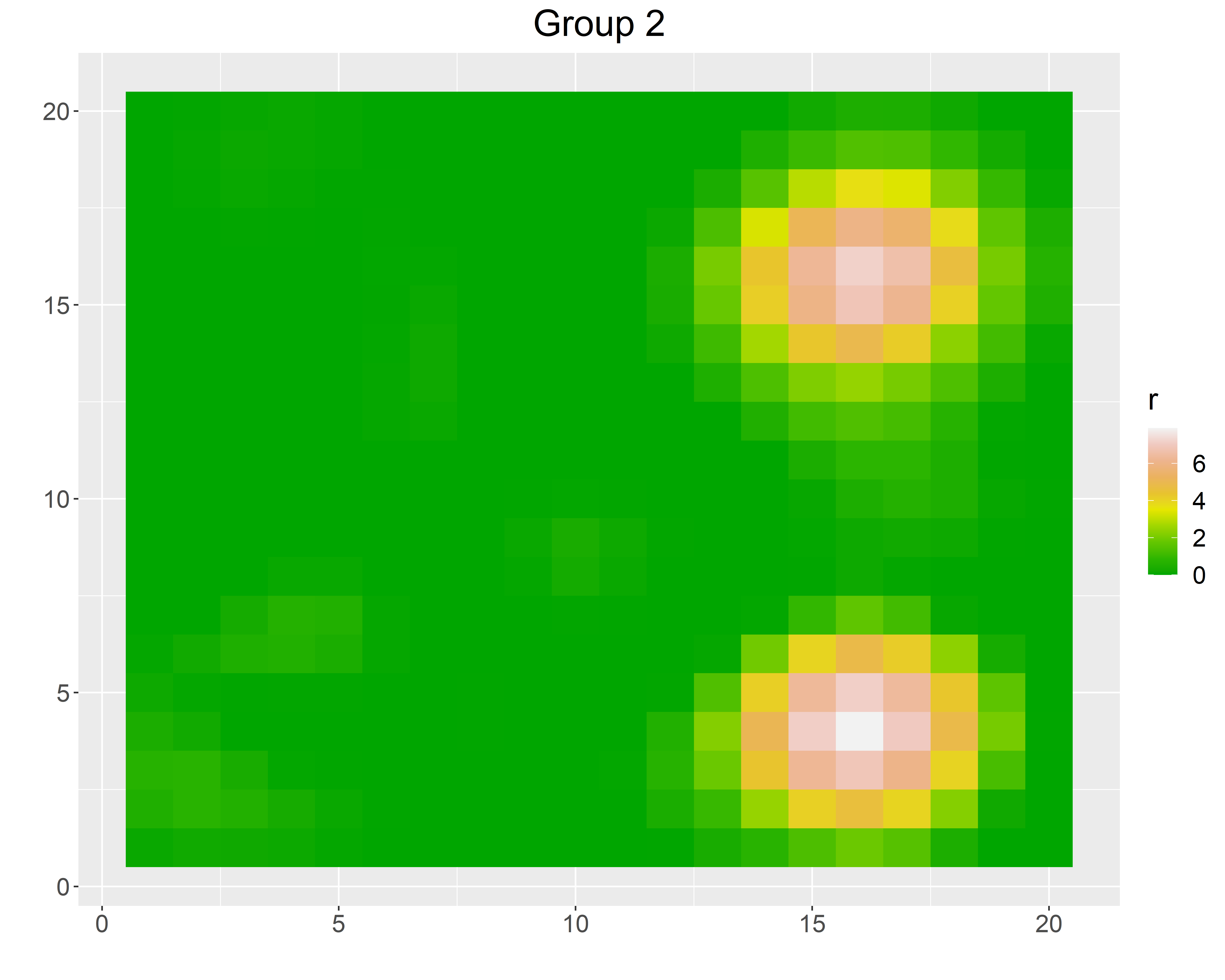}}
\subfigure{\includegraphics[width = 0.3\textwidth]{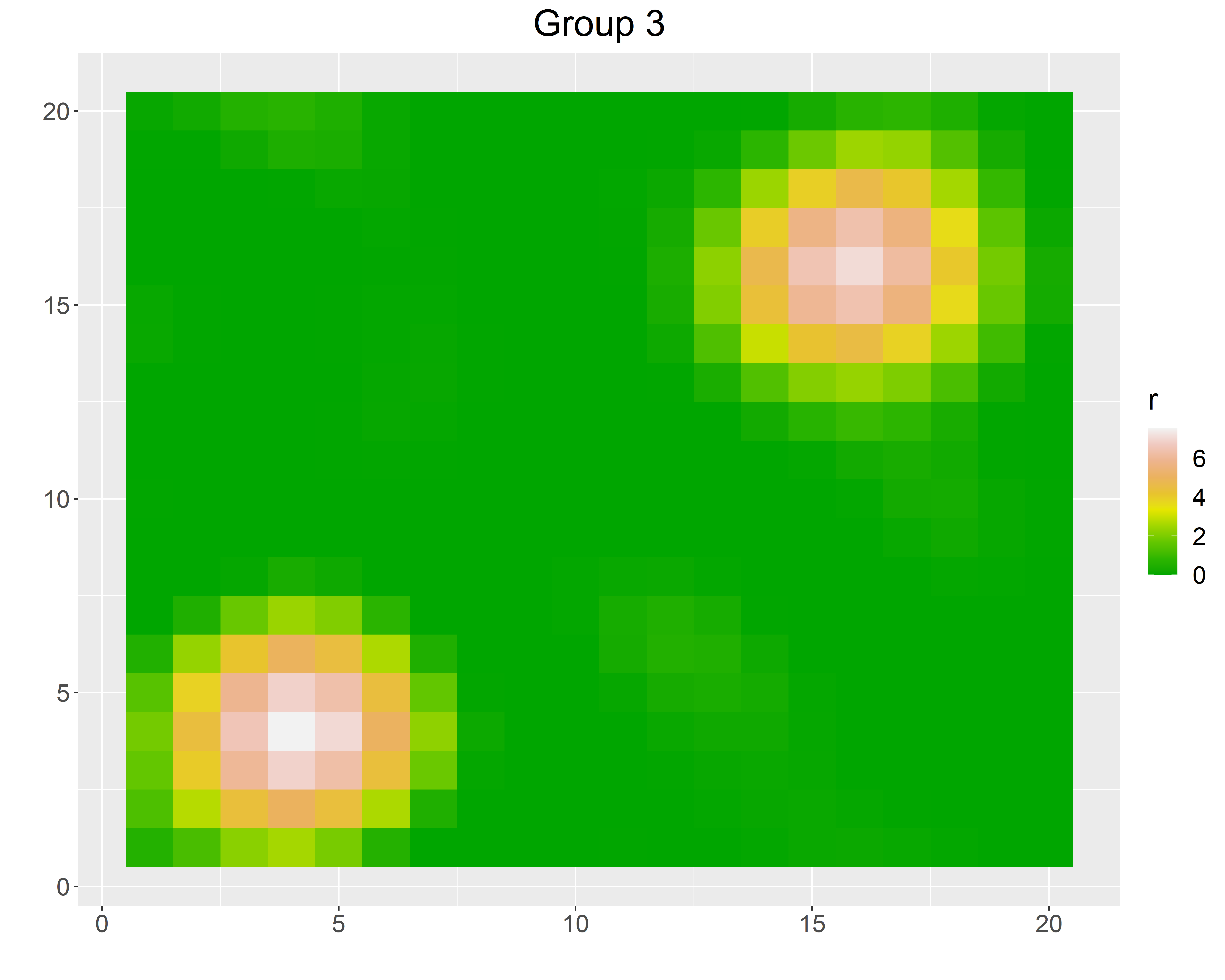}}
\subfigure{\includegraphics[width = 0.3\textwidth]{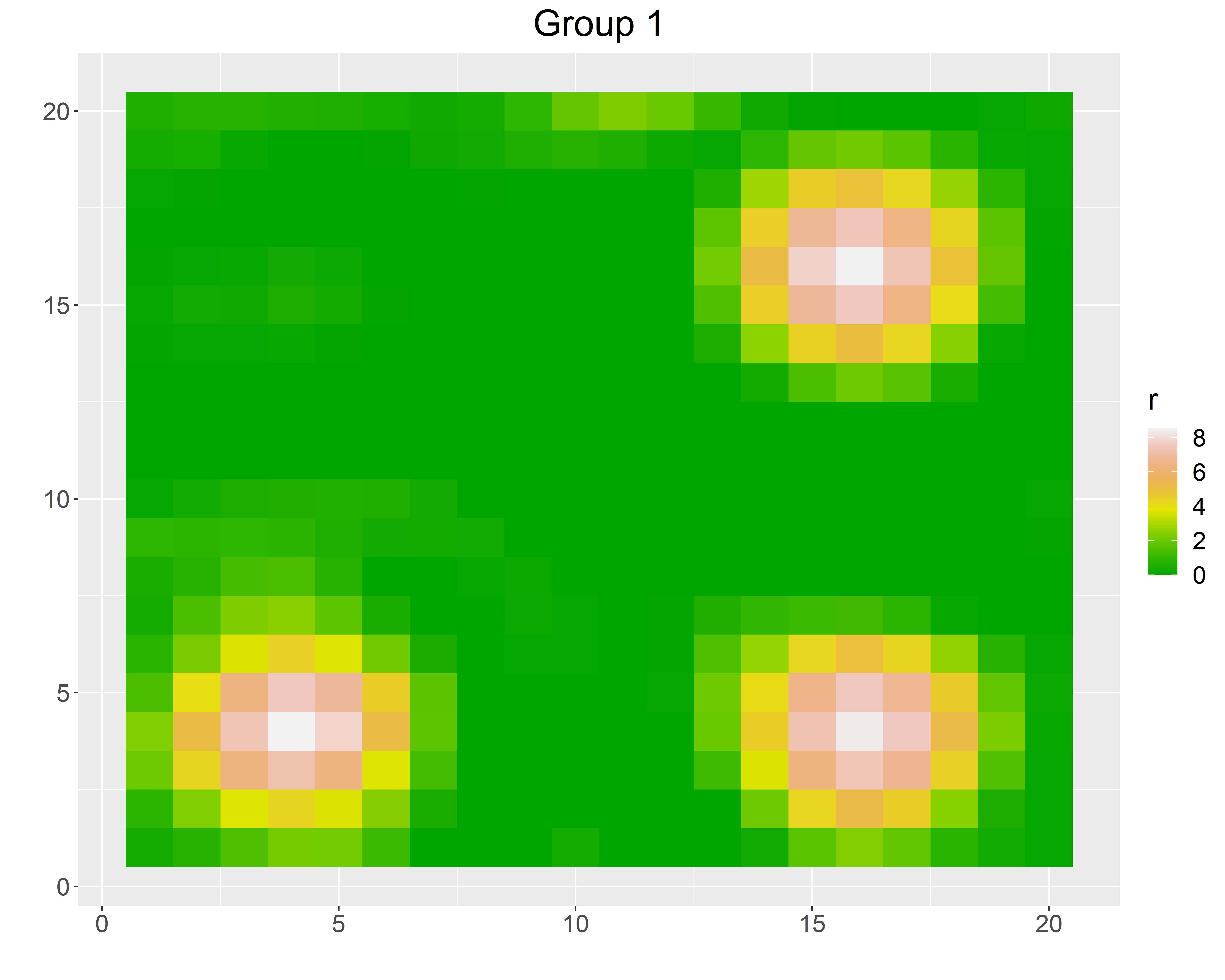}}
\subfigure{\includegraphics[width = 0.3\textwidth]{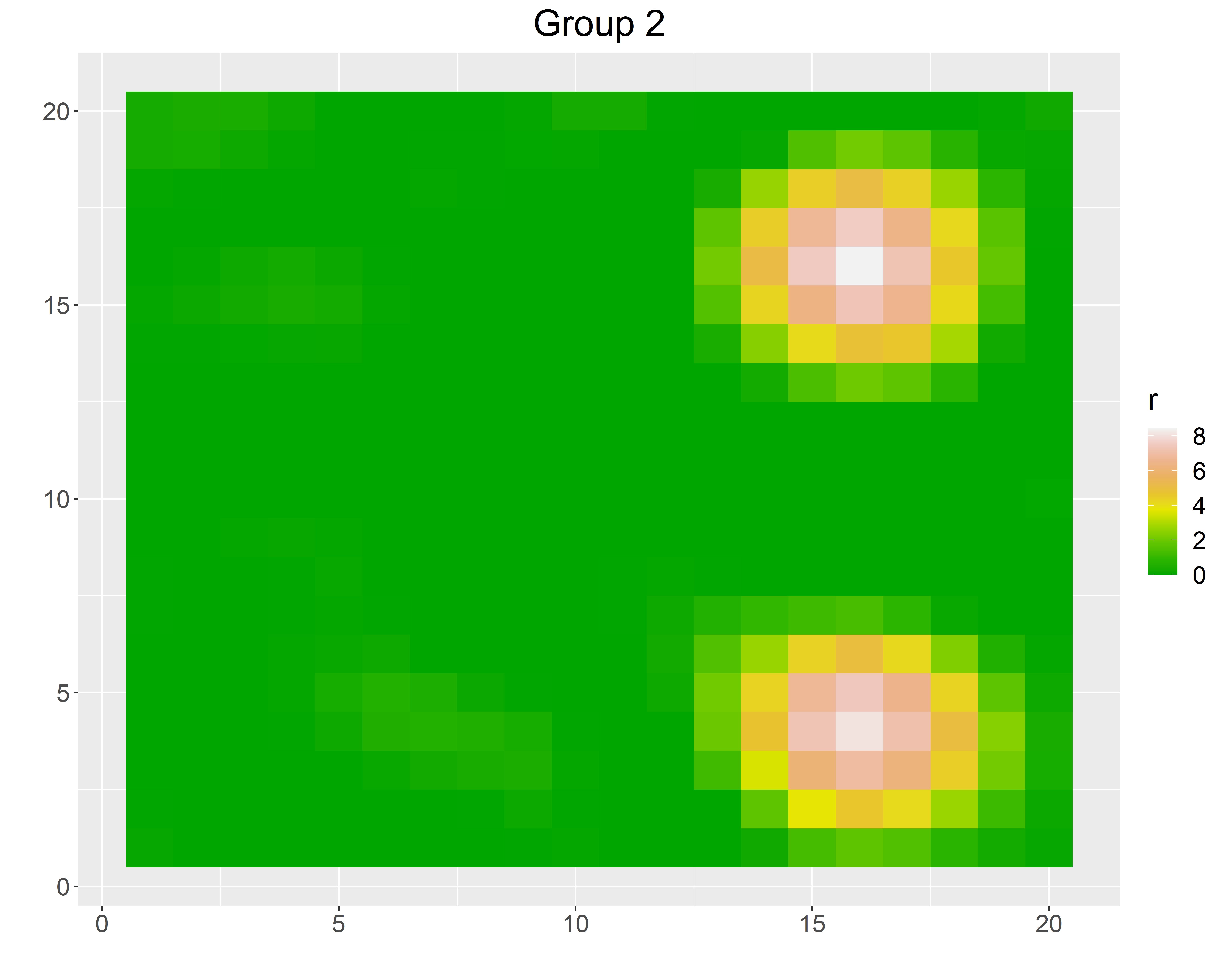}}
\subfigure{\includegraphics[width = 0.3\textwidth]{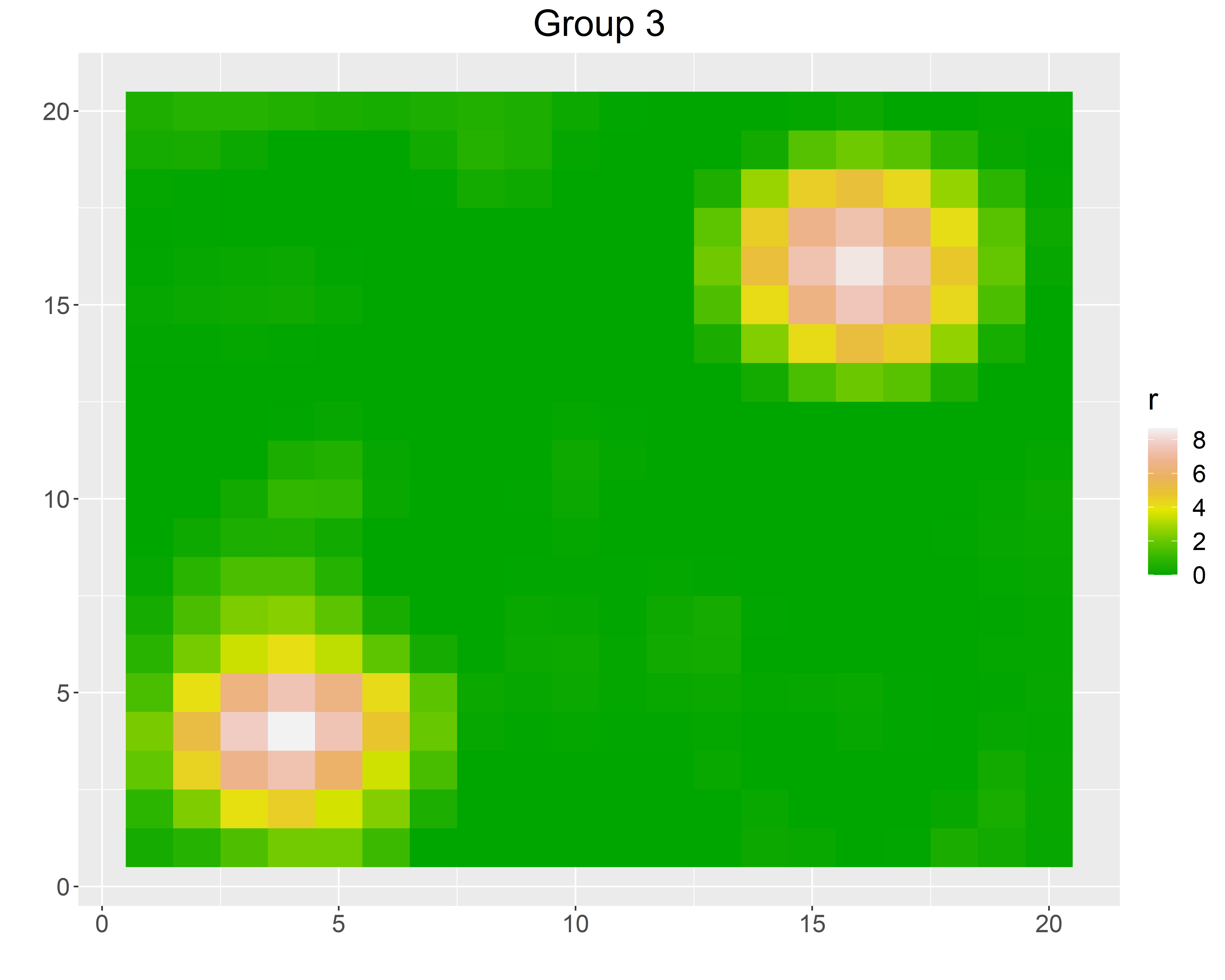}}
\caption{Estimated $\|\bbeta_g\|_2$'s when the error variance is $1$ for two sample sizes. The first row corresponds to a sample size of 50 in each group, and the second row is for 100. }
\label{fig::esticoefother} 
\end{figure}

The performances of all the methods are compared in Table~\ref{tab::f2} and the norms of the estimated coefficients are shown in Figure~\ref{fig::esticoefother}.
We again find that the proposed method is doing much better than all the other competitors. 
The relatively poor performance of other methods may be primarily attributed to the dependence among the components in $\bD_i$'s.
With an increasing sample size and smaller error variance, all the methods perform much better, as observed in case 1.
Furthermore, the performances are usually better when the error variance is lower, which corresponds to a greater signal-to-noise ratio.

\section{Application to the ADNI dataset}
\label{sec::real}

From the ADNI database, we collect data belonging to three disease groups: Alzheimer's disease, early and late mildly cognitive impaired (MCI) which are abbreviated as AD, EMCI, and LMCI. ADNI has modifed the image acquisition protocol several times.
To maintain homogeneity, in this paper, we focus on ADNI-2, which was initiated in September 2011 and continued for five years \citep{aisen2010clinical}. Based on the guidelines from the Centers for Disease Control and Prevention relying on \cite{xu2010national,heron2010national,tejada2013mortality,hurd2013monetary,matthews2019racial}, the symptoms of Alzheimer's disease often first appear after age 60.
Motivated by this, we included older than 60 subjects to create our study cohort.  This led to 30 AD, 35 EMCI, and 34 LMCI subjects.
Several studies have established associations between gender, age, and APOE status with mini-mental state examination (MMSE)  \citep{winnock2002longitudinal, qian2021association,matthews2012examining,piccinin2013coordinated,pradier2014mini}.
We thus apply our proposed model \eqref{eq: multimodel} after adjusting for these predictors.
Our final model for the real data analysis is,
\begin{align}
    {\textrm{MMSE}}_{i}=&b_{0,g}+b_{\textrm{Age}}\textrm{Age}_i+b_{\textrm{M}}Z_{i,\textrm{Male}}\nonumber\\&+b_{\textrm{Allele2}}Z_{\textrm{Allele2},i}+b_{\textrm{Allele4}}Z_{\textrm{Allele4},i}+p^{-1/2}\sum_{j=1}^p\bD_{i}(\bv_j)^T\bbeta_g(\bv_j) + e_{i}, \textrm{ for } i\in g\nonumber\\
    e_{i}\sim&\Normal(0, \sigma^2), \quad g=\textrm{AD, EMCI, LMCI} \label{eq:realmodel}
\end{align}
where the dummy variables $Z_{\text{Allele2}}$, $Z_{\textrm{Allele4}}$ standing for Alleles 2 and 4 for the two alleles APOEallele2 and APOEallele4 together setting Allele 3 as a reference group for each of the two cases.
Similarly, the dummy variable $Z_{\textrm{M}}$ indicating male gender is introduced, setting females as the reference group.
We put the ST2N-GP prior on $\bbeta_g$'s.

While applying the model to the corresponding diffusion tensor data of the ADNI, we consider 
the principal diffusion direction as our predictor. It is represented by the first eigenvector of the diffusion coefficient matrices. Hence, $\bD_i(\bv)$'s are 3-dimensional unit vectors.
In our modeling framework, we thus have $\bD_i(\bv)^T\bbeta(\bv)=\|\bbeta(\bv)\|_2\cos(\theta_{\bD_i(\bv),\bbeta(\bv)})$.
Hence, $\|\bbeta(\bv)\|_2$ now represents overall magnitude of effect from voxel $\bv$ and $\theta_{\bD_i(\bv),\bbeta(\bv)}$ stands for individual specific direction of effect. We separately apply our proposed model on five brain masks which are corpus callosum (CC), right corticospinal tract (CST), left CST, right frontal-parietal-temporal (FPT) and left FPT. Among these, CC is the largest, with around 13K voxels. The other regions on average consist of $4,000$-$7,000$ voxels.

For space constraints, Figure~\ref{fig::esticoefreal} illustrates the important regions for three slices in the CST region only. These three slices are from the frontal, middle, and posterior sections respectively. In this figure, we also demonstrate the regions that are identified as important for all the groups as `Shared'. The rest of the estimates are illustrated similarly in Figures S3 and S4 of the Supplementary Materials for CC and FPT, respectively. In our estimates, we observe the following interesting characteristics. Estimates associated with the EMCI group usually have the largest number of important regions, then LMCI, and lastly AD with the smallest number of selected regions.
The lower part of each tract has a greater number of important areas.
This may be due to the higher white matter density of that part.
Corpus callosum has the maximum percentage of important regions, which may also be attributed to the greater density of white matter. 
In Table~\ref{tab::resultrealreg}, gender turns out to be important for all the cases, and age has an important negative effect in the two cases. The uneven proportion of important regions for different disease groups justifies our multi-group model with group-specific regression effects.
The effect of APOE4 turns out to be significant only for corpus callosum.
There are thus several interesting findings in our analysis using the white matter fiber alignment directly.

\begin{table}[ht]
\centering
\caption{Estimated effects of different scalar predictors. In the first column, the region names are provided. For each predictor, we report the posterior mean along with its lower and upper credible bands.}
\resizebox{.7\textwidth}{!}{\begin{tabular}{r|r|rrrr}
  \hline
& & Gender & Age & Allele 3 & Allele 4 \\ 
  \hline
&Estimate & 1.36 & -0.40 & 0.94 & -0.66 \\ 
  Corpus Callosum & Lower CI & 0.62 & -0.76 & -0.84 & -1.37 \\ 
  & Upper CI & 2.15 & -0.06 & 2.62 & -0.01 \\ 
  \hline
  &Estimate & 0.96 & -0.10 & 0.77 & 0.64 \\ 
  Right CST & Lower CI & 0.23 & -0.49 & -0.80 & -0.39 \\ 
  &Upper CI & 1.70 & 0.25 & 2.27 & 1.53 \\
  \hline
 & Estimate & 1.06 & -0.19 & 1.11 & 0.00 \\ 
  Left CST & Lower CI & 0.36 & -0.52 & -0.27 & -0.74 \\ 
  &Upper CI & 1.72 & 0.13 & 2.47 & 0.97 \\
  \hline
  &Estimate & 1.43 & -0.51 & 1.18 & 0.12 \\ 
  Right FPT & Lower CI & 0.69 & -0.88 & -0.48 & -0.66 \\ 
  &Upper CI & 2.23 & -0.12 & 2.77 & 0.89 \\ \hline
  &Estimate & 0.96 & -0.07 & 1.23 & 0.18 \\ 
  Left FPT &Lower CI & 0.29 & -0.44 & -0.23 & -0.62 \\ 
  &Upper CI & 1.65 & 0.32 & 2.77 & 0.96 \\ 
   \hline
\end{tabular}}
\label{tab::resultrealreg}
\end{table}

\begin{figure}[htbp]
\centering
\includegraphics[width=0.8\textwidth]{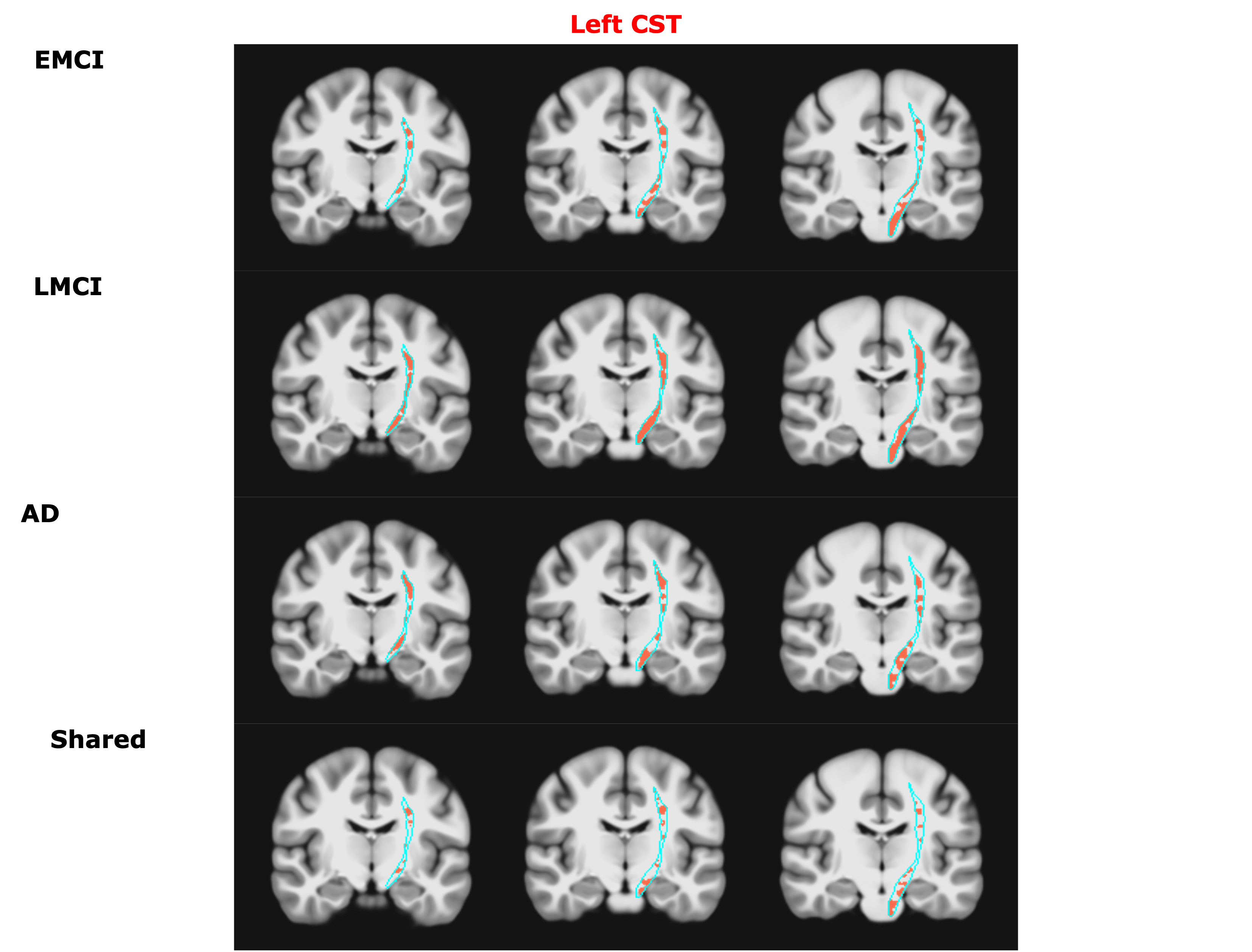}
\includegraphics[width=0.8\textwidth]{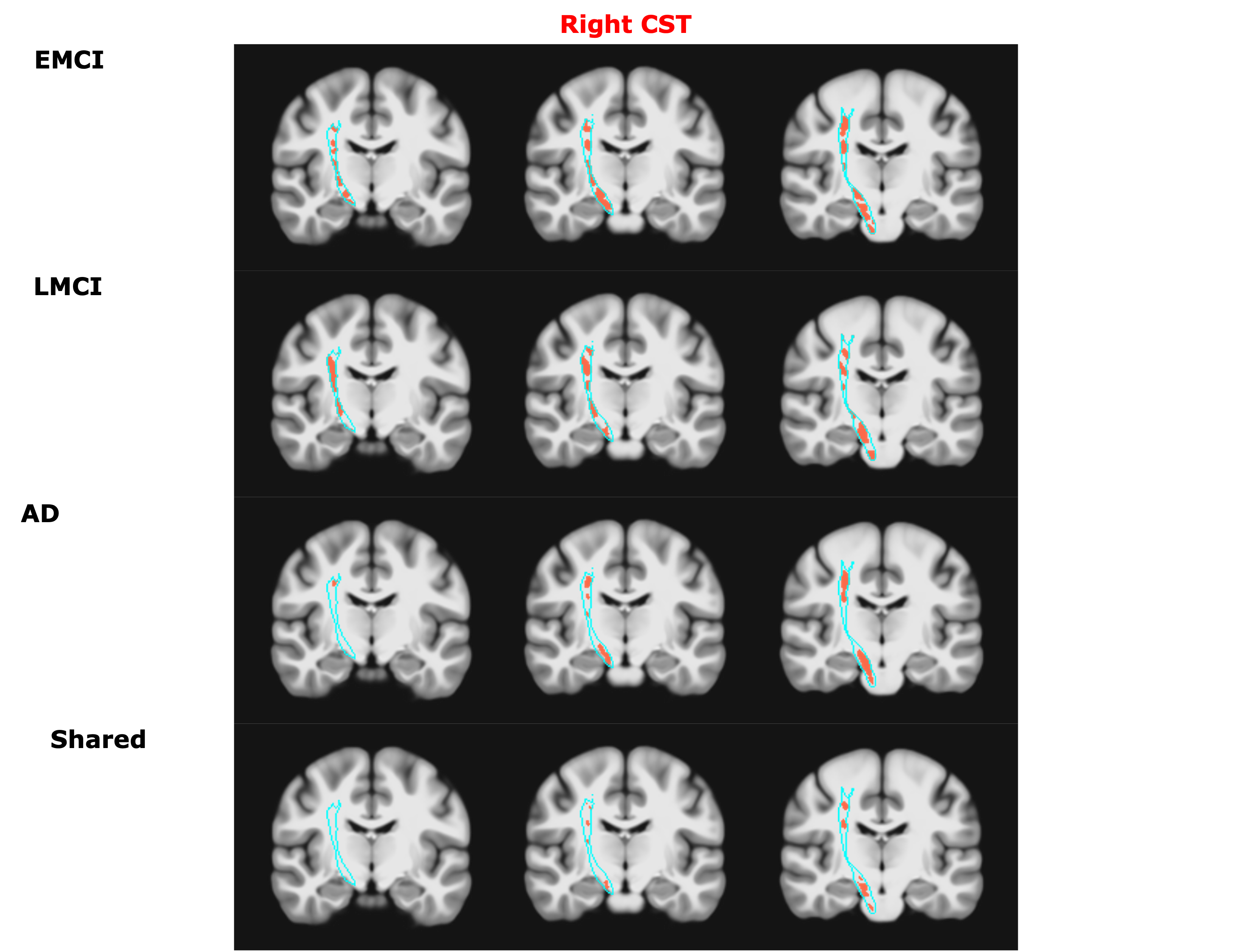}
\caption{The red voxels are the brain regions with non-zero $\|\hat{\bbeta}_g\|_2$'s. The CST regions are delineated with light blue lines in all the images on typical slices of the brain. The last row illustrates the voxels with non-zero $\|\hat{\bbeta}_g\|_2$ for all $g$. The three columns are three slices, ranging from the frontal to posterior section.}
\label{fig::esticoefreal}
\end{figure}

We also compare ST2N with STGP in terms of out-of-sample prediction in Table~\ref{tab::resultreal}.
STGP is chosen as a competitor since it was the second-best performer in the simulation after ST2N in estimating the true regression coefficient.
The available code for STGP only allows imaging predictors. For a fair comparison, we thus consider a reduced model with only vector-valued DTI image as predictor and MMSE as response. As in Section~\ref{sec::simu}, we apply STGP for each group separately and consider each coordinate of principal diffusion direction data as a separate imaging predictor. For each group, we leave out 5 subjects and run the models with the rest of the subjects. ST2N overwhelmingly outperforms STGP. The poor performance of STGP may be due to group-specific separate estimations with small sample sizes. On the other hand, our multi-group model fits the three regression coefficients simultaneously. Hence, for multi-group data, our proposed model is more suitable.
In the next subsection,
we also fit a nonparametric analog of the linear model and compare them.


\begin{table}[ht]
\centering
\caption{Out of sample prediction MSE, averaged over all the three groups for different brain masks.}
\resizebox{1\textwidth}{!}{\begin{tabular}{r|rrrrr}
  \hline
&Corpus Callosum & Right CST & Left CST & Right FPT & Left FPT \\ 
  \hline
  ST2N-GP& 53.92 & 35.28 & 117.35 & 38.43 & 143.74 \\ 
STGP  &745.73  & 746.83 & 744.43 & 722.63& 725.34\\ 
   \hline
\end{tabular}}
\label{tab::resultreal}
\end{table}

\subsection{Nonparametric analog model}
\label{sec:nonpara}

Here, we consider a non-parametric analog of the model, following the lines of \cite{ravikumar2009sparse}. Specifically, we consider the following model,
\begin{align}
    {\textrm{MMSE}}_{i}=&p^{-1/2}\sum_{j=1}^pf_{j,g}(\bD_{i}(\bv_j)) + e_{i},\nonumber\\
    f_{j,g}(\bD_{i}(\bv_j))=&\sum_{k=1}^K\theta_{g,k}(\bv_j)\exp\left(-\frac{\|\bD_{i}(\bv_j)-\bq_{k}\|_2^2}{2h_g^2}\right)I\{\|\bD_{i}(\bv_j)-\bq_k\|_2<3h_g\}, \textrm{ for } i\in g,\nonumber\\
    e_{i}\sim&\Normal(0, \sigma^2), \quad g=\textrm{AD, EMCI, LMCI}, \label{eq:nonparamodel}
\end{align}
where to model the non-parametric function, we consider Gaussian radial bases with centers $\bq_k$'s. Since $\bD_{i}(\bv)$'s are supported on a sphere, the centers also are supported in the same way. Specifically, we consider uniformly distributed Fibonacci lattice points as the centers. 
The bandwidth parameter $h_g$ is set as $\frac{\sum_{i,j:a_{i,j,k}\neq 0} a_{i,j,k}}{\sum_{i,j,k} 1\{a_{i,j,k}\neq 0\}}$, where $a_{i,j,k}=\sqrt{2}\|\bD_{i}(\bv_j)-\bq_k\|_2$ for $i\in g$.
Instead of a fully nonparametric model, these models are easier to interpret, as discussed in \cite{ravikumar2009sparse}.

We put multi-group ST2N prior on $\btheta_{g}(\bv_j)=\{\theta_{g,k}(\bv_j)\}_{k=1}^K$. Fibonacci lattice points are evaluated for a given value of $K$. For a fixed $K$, the implementation of the nonparametric model is the same as the linear model.
Here, we select $K$ via cross-validation from a pre-specified grid $K=\{8,10,15,20\}$. Due to high computational need, we run this step for 10 pairs of training and test sets.
The training set contains a randomly selected block of dimension $20\times 20\times 20$ and the test is another randomly selected block of dimension, $10\times 10\times 10$. Based on this strategy, 10 is selected as the optimal choice.

The nonparametric model puts forward an even better prediction accuracy, as shown in Table~\ref{tab::resultreal1} in comparison to Table~\ref{tab::resultreal}. In general, our proposed multi-group ST2N model has two key advantages: 1) the mechanism to have the group-specific effects borrow structural characteristics, and 2) its ability to use samples from all the groups simultaneously. This may have led to the improved performance of our method over STGP.

However, this nonparametric analog model is specifically designed for analyzing directional-valued predictor data. Hence, future work is needed to make it generally applicable to other data types. Furthermore, studying the associated theoretical results would also be of interest.

In Section S4 of the Supplementary Materials, we include limited simulation results where this nonparametric model is applied to the direction-valued predictor data from Simulation setting 2 of Section~\ref{sec::simucase2}.
Specifically, we plot $\sqrt{\sum_{k=1}^K\theta^2_{g,k}(\bv_j)}$ for each group $g$ and spatial locations $\bv_j$'s. The results indicate that the model successfully identifies the important voxels while comparing with true non-zero locations in $\bbeta_{0,g}$'s.

\begin{table}[ht]
\centering
\caption{Out of sample prediction MSE, averaged over all the three groups for different brain masks under the same setting as in Table~\ref{tab::resultreal}.}
\resizebox{.9\textwidth}{!}{\begin{tabular}{r|rrrrr}
  \hline
&Corpus Callosum & Right CST & Left CST & Right FPT & Left FPT \\ 
  \hline
Nonpara ST2N& 11.92 & 5.23 & 50.22 & 12.21 & 70.12 \\
   \hline
\end{tabular}}
\label{tab::resultreal1}
\end{table}

\section{Discussion}
\label{sec::discussion}
In this paper, we developed a novel spatial variable selection prior for the vector-valued image on scalar regression.
We further extended the model for the multi-group setting and implemented an efficient MCMC algorithm.
The utility of the proposed method was demonstrated in a DTI data application.
Methodologically, our proposed prior can be easily applied to a sparse nonparametric additive scalar on the image regression model generalizing \eqref{eq:nonparamodel}, which may be formulated as, $y_i\sim\Normal\left(\sum_{\bv\in\mathcal{V}}\eta(\bv)f_{\bv}(x_{i}(\bv)),\sigma^2\right)$ which is a generalization of the nonparametric analog model from Section~\ref{sec:nonpara}.
It is possible to rewrite the mean as $\sum_{\bv\in\mathcal{V}}\eta(\bv)f_{\bv}(x_{i}(\bv))=\sum_{\bv\in\mathcal{V}}\bbeta(\bv)^T \bD_i(\bv)$, where $\bD_i(\bv)=\{B_k(x_i(\bv))\}_{k=1}^K$ and $\bbeta(\bv)=\{\eta(\bv)\theta_k(\bv)\}_{k=1}^K$.
Since $\eta(\bv)$ and $\theta_k(\bv)$ are not separately identifiable, we may directly put our ST2N-GP prior on $\bbeta(\bv)$ for spatial variable selection. 
Extension to such a modeling framework will be very useful for building a nonparametric regression model with an imaging predictor.  
The DTI application of our paper followed some of the previous works focusing on identifying the ``regions"  within a fiber tract that drive the variation of the outcome \cite{zekelman2022white,yeatman2012tract}. Likewise, our work is to identify the voxels that cover the ``regions"  within a fiber tract that drive the variation. The available {population-averaged} fiber tractography information helps us localize the fiber tract of interest, and our method further helps identify the voxels of ``regions" that impact the outcome through a process of spatial variable selection. 

Another important issue raised by one of the reviewers is the appropriateness of geospatial-type scalar on image regression for DTI data. Specific to our DTI application, principal diffusion directions exhibit the white matter fiber orientations/tracts in the brain. Alzheimer's disease is known to affect these tracts, and thus it is indeed reasonable to quantify the effect of different tracts on the outcome and also possibly identify the differentially impacting tracts on MMSE. We can now use {population-averaged} fiber tractography information to build a more comprehensive model for such analysis. Although our geospatial-covariance-based model will not be suitable here, it may be possible to modify the kernel and then impose group-penalty-induced shrinkage to address tract-based analysis. We will consider this as one of our future research directions.
Future work will also explore the application of the proposed method beyond large white matter regions.
Our choice of the large white matter regions is based on the literature and our limited preliminary studies that within white matter regions, the diffusion characteristics are indeed spatially dependent \cite{zhu2010multivariate}.
Theoretically, it would be interesting to study the properties under the number of groups increasing with $n$ regime. This would require the groups to share structure more systematically.

To examine the algorithmic complexity of our proposed method, we estimate and plot computational times with an increasing number of spatial locations and sample sizes, while keeping the number of groups fixed. The figures are added in Section S5 of the Supplementary Materials, with some additional follow-up analysis on the order of computation. Specifically, we fit separate regression models in the log scale as $\log(t)$ on $\log(n)$ and $\log(t)$ on $\log(p)$ to evaluate the simulation-based order of computation. In both cases, results suggest that the computation time is linear.
Although the algorithmic complexities are reasonably good, future efforts will consider developing Rcpp or Rcpparamdillo-based implementations for faster computation to reduce CPU time.

In the case of multimodal imaging, the predictor vector $\bD_i(\bv)$ may consist of both the main effects and the interaction effect terms, and we can analyze with the ST2N-GP prior.
In our ADNI application, it will be interesting to add other imaging markers, derived from fMRI or PET, and run a combined multimodal analysis.
The scaled squared exponential kernel can be replaced by a more flexible Mat\'ern kernel too. 
Furthermore, the latent process $\widetilde{\bbeta}$ does not necessarily need to be a GP.
It might follow newly proposed processes, such as Nearest Neighbor GP which is computationally more affordable than traditional GP.
The latent process may even be modeled using the deep neural net, RBF-net, or any other flexible classes.
In the future, we explore these possibilities too.

The ST2N transformation-based priors can be applied easily to the other types of vector-valued image regression models, where we may even have the image data as a response instead of being a predictor.
Our future works will focus on such applications.
We have also established the theoretical results under a completely new set of assumptions on the predictor process.
Such characterization may help establish theoretical properties for other statistical models with spatially varying predictors.
However, additional control on the predictor process might help us to establish even stronger posterior contraction and variable selection results.

\section{Acknowledgments}
We thank the Editor, Associate Editor, and referee for the insightful comments that led to improvements in the content and presentation
of the paper.

Data collection and sharing for this project was funded by the Alzheimer's Disease Neuroimaging Initiative (ADNI) (National Institutes of Health Grant U01 AG024904) and DOD ADNI (Department of Defense award number W81XWH-12-2-0012). ADNI is funded by the National Institute on Aging, the National Institute of Biomedical Imaging and Bioengineering, and through generous contributions from the following: AbbVie, Alzheimer's Association; Alzheimer's Drug Discovery Foundation; Araclon Biotech; BioClinica, Inc.; Biogen; Bristol-Myers Squibb Company; CereSpir, Inc.; Cogstate; Eisai Inc.; Elan Pharmaceuticals, Inc.; Eli Lilly and Company; EuroImmun; F. Hoffmann-La Roche Ltd and its affiliated company Genentech, Inc.; Fujirebio; GE
Healthcare; IXICO Ltd.; Janssen Alzheimer's Immunotherapy Research \& Development, LLC.; Johnson \& Johnson Pharmaceutical Research \& Development LLC.; Lumosity; Lundbeck; Merck \& Co., Inc.; Meso
Scale Diagnostics, LLC.; NeuroRx Research; Neurotrack Technologies; Novartis Pharmaceuticals Corporation; Pfizer Inc.; Piramal Imaging; Servier; Takeda Pharmaceutical Company; and Transition
Therapeutics. The Canadian Institutes of Health Research is providing funds to support ADNI clinical sites in Canada. Private sector contributions are facilitated by the Foundation for the National Institutes of Health ({\tt www.fnih.org}). The grantee organization is the Northern California Institute for Research and Education,
and the study is coordinated by the Alzheimer's Therapeutic Research Institute at the University of Southern California. ADNI data are disseminated by the Laboratory for Neuro Imaging at the University of Southern California.




\section*{Supplementary materials}

\setcounter{section}{0}

\makeatletter
\renewcommand \thesection{S\@arabic\c@section}
\renewcommand\thetable{S\@arabic\c@table}
\renewcommand \thefigure{S\@arabic\c@figure}
\makeatother

\section{Additional results and figures}

We introduce the following definition of similar effect. 

\begin{SDef}[Similar effect]
The effect at $\bv$-$th$ spatial location is considered similar across all the groups if and only if $\bbeta_g(\bv)^T\bbeta_{g'}(\bv)>0$ for all $g\neq g'$.
\label{def::signsimilar} 
\end{SDef}
In the above definition of similar effect, only the directions of effects play a role.
We only require that the maximum possible angle between $\bbeta_g(\bv)$ and $\bbeta_g'(\bv)$ at a similar-effect location is within $[0,\pi/2)$ for all $g\neq g'$. 

We further quantify the similarity by $\psi(\bv)=\min_{g,g':g\neq g'} \left(\frac{\bbeta_g(\bv)}{\|\bbeta_g(\bv)\|_2}\right)^T\frac{\bbeta_{g'}(\bv)}{\|\bbeta_{g'}(\bv)\|_2}$ 
at a similar-effect location $\bv$.
Thus $\cos^{-1}\psi(\bv)$ is the maximum angle between $\bbeta_g(\bv)$ and $\bbeta_g'(\bv)$ at $\bv$.
Hence, having smaller $\cos^{-1}\psi(\bv)$ implies that the directions of $\bbeta_g(\bv)$'s are getting closer to each other.
The composite effect of $\bv$-$th$ location on the response for $g$-$th$ group with predictor $\bD(\bv)$ is $\bbeta_{g}^T(\bv)\bD(\bv)$.
Hence, the sign of $\bbeta_{g}^T(\bv)\bD(\bv)$ implies whether $\bv$-$th$ location affects the response positively or negatively.
We thus present our following Lemma~S\ref{lem::signpred} which implies that for a predictor vector $\bD(\bv)$, the sign of $\bbeta_{g}^T(\bv)\bD(\bv)$ is more likely to be the same for all $g$ for larger values of $\psi(\bv)$.

\begin{SLem}
If $\psi(v)=c$, we have $P\left(\Big[\frac{\bbeta_{1}(\bv)}{\|\bbeta_{1}(\bv)\|_2}\Big]^T\bD(\bv)>0\bigm| \Big[\frac{\bbeta_{2}(\bv)}{\|\bbeta_{2}(\bv)\|_2}\Big]^T\bD(\bv)<0\right)$ decreases as $c$ increases for $G=2$.
\label{lem::signpred}
\end{SLem}
In the above Lemma, the probability is calculated under the marginal distribution of $\bD(\bv)$. Specifically, we calculate the probability of the event $\{\bD(\bv): \Big[\frac{\bbeta_{1}(\bv)}{\|\bbeta_{1}(\bv)\|_2}\Big]^T\bD(\bv) > 0\}$ conditioning on $\{\bD(\bv):\Big[\frac{\bbeta_{2}(\bv)}{\|\bbeta_{2}(\bv)\|_2}\Big]^T\bD(\bv)<0\}$
Hence, for simplicity, we use $\psi(\bv)$ to quantitatively summarize the across-group similarity in regression effects. The proof is straightforward and provided in the supplementary materials.
The result is presented for $G=2$ to maintain simplicity in the calculation.
However, it can be generalized easily by defining, $\psi_{g,g'}(\bv)=\left(\frac{\bbeta_g(\bv)}{\|\bbeta_g(\bv)\|_2}\right)^T\frac{\bbeta_{g'}(\bv)}{\|\bbeta_{g'}(\bv)\|_2}$ which is always more than or equal to $\psi(\bv)$.
Hence, for a general $G$, we can show the above result for each pair of $(g,g')$ taking $\psi_{g,g'}(\bv)=c$. Some additional results on prior properties are in the supplementary materials for space.

Our decomposition of $\widetilde{\bbeta}_g(\cdot)$ is motivated to center the group-specific latent processes and to improve computational performances.
However, this decomposition has some additional advantages.
Specifically, we expect that the regression effects at most of the spatial locations will be similar across the groups, except for some sparsely distributed regions.
Nevertheless, we study some additional properties of this model for better understanding.
Although $\bbeta_{g}(\cdot)$'s are identifiable, the components of $\widetilde{\bbeta}_g(\cdot)$ namely $\widetilde{\bbeta}(\cdot)$ and $\balpha_g(\cdot)$'s are not separately identifiable, as they produce equivalent likelihood as long as $h_{\lambda}(\widetilde{\bbeta}(\cdot) + \balpha_g(\cdot))$ does not change.
In the Bayesian paradigm, however, the posterior mode of $\widetilde{\bbeta}(\cdot)$ is controlled by the prior probabilities only. 
We study the prior probabilities and other related properties in connection to our definition of similar-effect locations (Definition 1).
We put Gaussian process priors on the spatially varying latent functions as discussed in the next section, and thus we have the following result.

\begin{SLem}
The angular separation between the prior mode of $\widetilde{\bbeta}(\bv)$ and $\bbeta_g(\bv)$'s is within $[0, \pi/2)$ at a similar effect location $\bv$.
\label{lem::priorshare}
\end{SLem}

Recall that, from Definition 1, a similar-effect location $\bv$ satisfies $\bbeta_g^T(\bv)\bbeta_{g'}(\bv)>0$ for all $g\neq g'$. The proof can be found in the Appendix. The above result is applicable only for the similar-effect spatial locations.
In order to appropriately identify the similar-effect locations using $\widetilde{\bbeta}(\bv)$ alone, we need to study its properties at the `non-similar-effect' locations as well.
Specifically, we propose a spatially adaptive thresholding transformation for $\widetilde{\bbeta}(\bv)$ that works remarkably well in identifying similar-effect spatial locations. It is based on the following result.

\begin{SLem}
At a given spatial location $\bv$, if $\|\widetilde{\bbeta}(\bv)\|_2>\lambda$ with $\|\bbeta_{g}(\bv)\|_2=0$, then $\balpha_g(\bv)^T\{\balpha_g(\bv)+2\widetilde{\bbeta}(\bv)\}<0$ and $\|h_{\lambda}(\widetilde{\bbeta}(\bv))\|_2^2<-\balpha_g(\bv)^T\{\balpha_g(\bv)+2\widetilde{\bbeta}(\bv)\}$.
\label{lem::adpthresh}
\end{SLem}
\begin{proof}
The proof follows by noting that $\|\bbeta_g(\bv)\|_2=0$
implies that, $\|\widetilde{\bbeta}(\bv)+\balpha_g(\bv)\|_2<\lambda$. Then breaking the square, we get $\widetilde{\bbeta}(\bv)^T\widetilde{\bbeta}(\bv) + \balpha_g(\bv)^T\{\balpha_g(\bv)+2\widetilde{\bbeta}(\bv)\} < \lambda^2$. Thus $\widetilde{\bbeta}(\bv)^T\widetilde{\bbeta}(\bv)>\lambda^2$ implies $\balpha_g(\bv)^T\{\balpha_g(\bv)+2\widetilde{\bbeta}(\bv)\}<0$ and $\|h_{\lambda}(\widetilde{\bbeta}(\bv))\|_2^2<-\balpha_g(\bv)^T\{\balpha_g(\bv)+2\widetilde{\bbeta}(\bv)\}$ 
\end{proof}

Using the above lemma, we propose a spatially varying thresholding function for $\|h_{\lambda}(\widetilde{\bbeta}(\bv))\|_2$ as $\lambda_S(\bv)=\sqrt{-\max_{g:\balpha_g(\bv)^T\{\balpha_g(\bv)+2\widetilde{\bbeta}(\bv)\}<0} \balpha_g(\bv)^T\{\balpha_g(\bv)+2\widetilde{\bbeta}(\bv)\}}$.
The similar-effect is thus can be identified using $F(h_{\lambda}(\widetilde{\bbeta}(\bv)))=\left(1-\frac{\lambda_S(\bv)}{\|h_{\lambda}(\widetilde{\bbeta}(\bv))\|_2}\right)_{+}h_{\lambda}(\widetilde{\bbeta}(\bv))$.
We show its excellent performance in identifying similar-effect locations in Figures~\ref{simufig1} and \ref{simufig2} for simulations 1 and 2, respectively.

\begin{figure}[htbp]
\centering
\subfigure{\includegraphics[width = 0.4\textwidth]{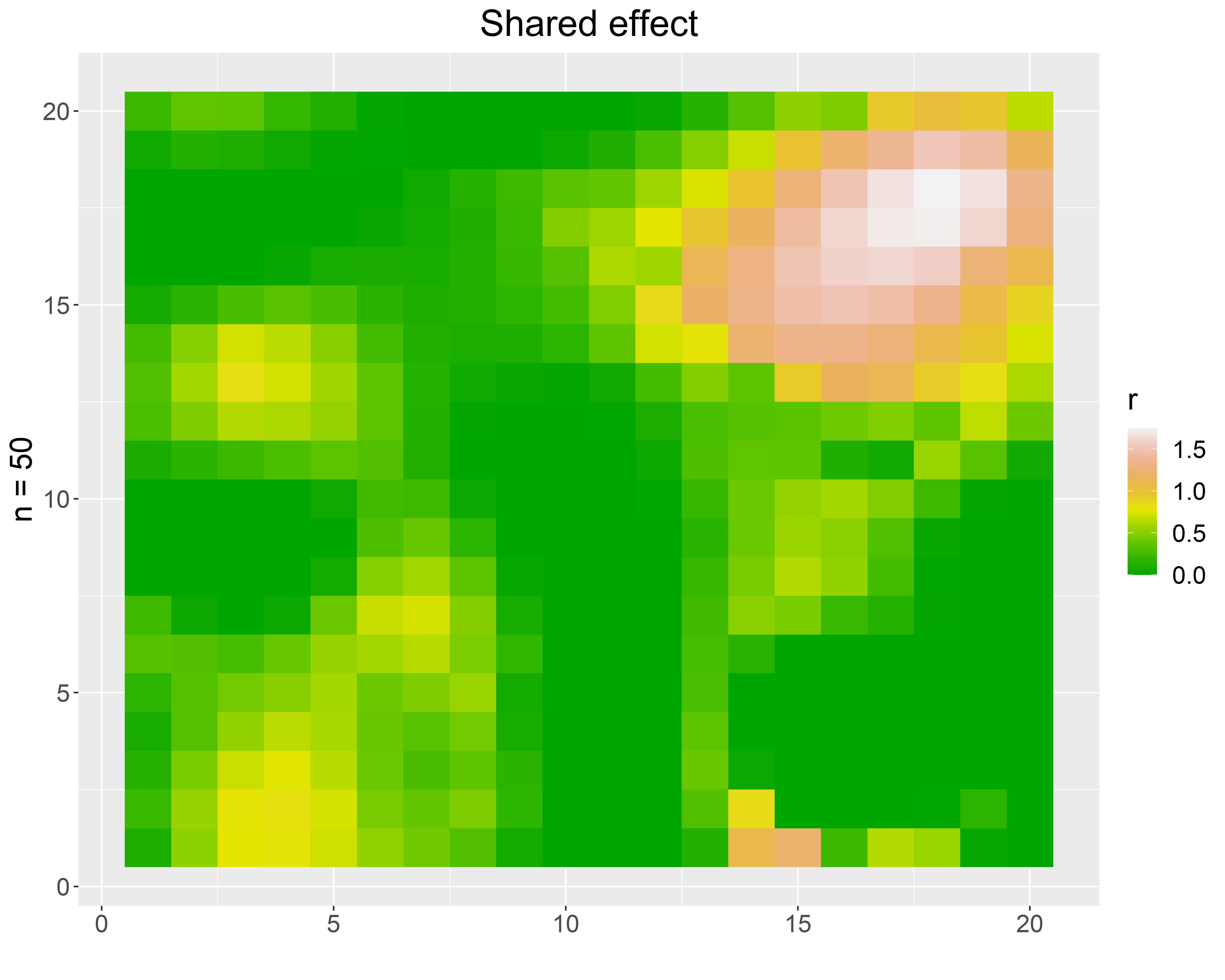}}
\subfigure{\includegraphics[width = 0.4\textwidth]{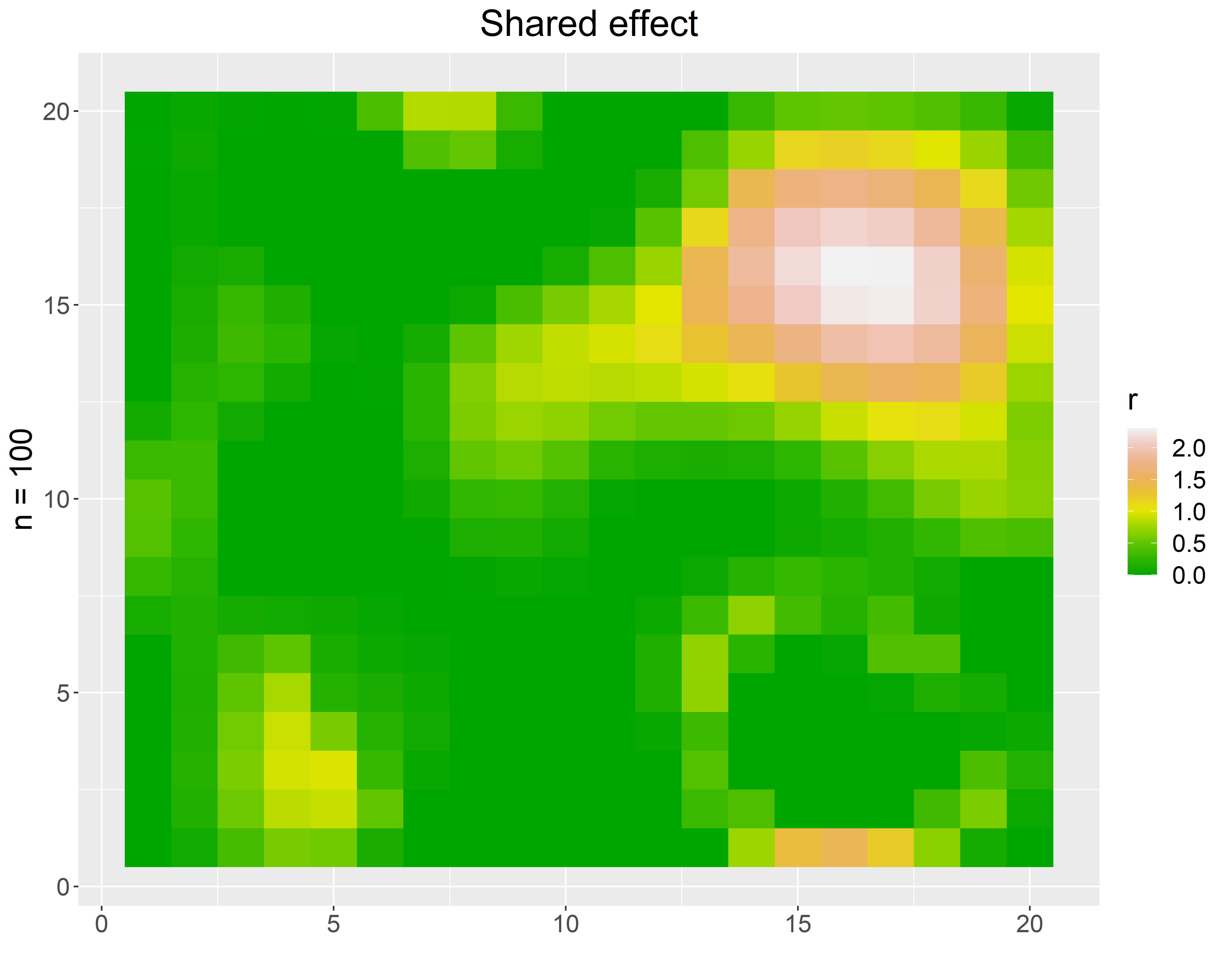}}
\caption{Estimated $\|F(h_{\lambda}(\widetilde{\bbeta}))\|_2$'s when the error variance is $1$ for two sample sizes. The two images correspond to sample sizes of 50 and 100 in each group, respectively, in Simulation case 1.}
\label{simufig1} 
\end{figure}

\begin{figure}[htbp]
\centering
\subfigure{\includegraphics[width = 0.4\textwidth]{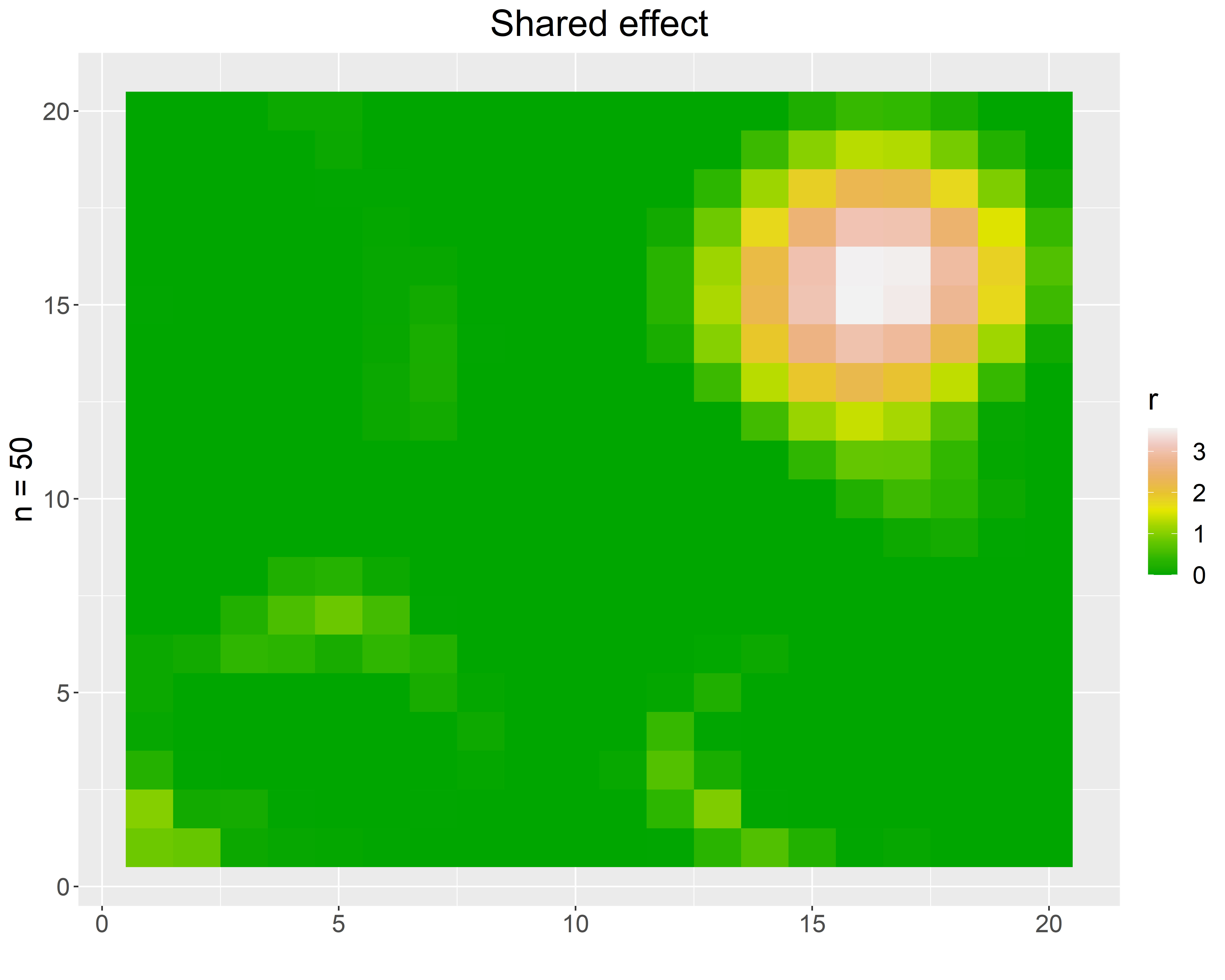}}
\subfigure{\includegraphics[width = 0.4\textwidth]{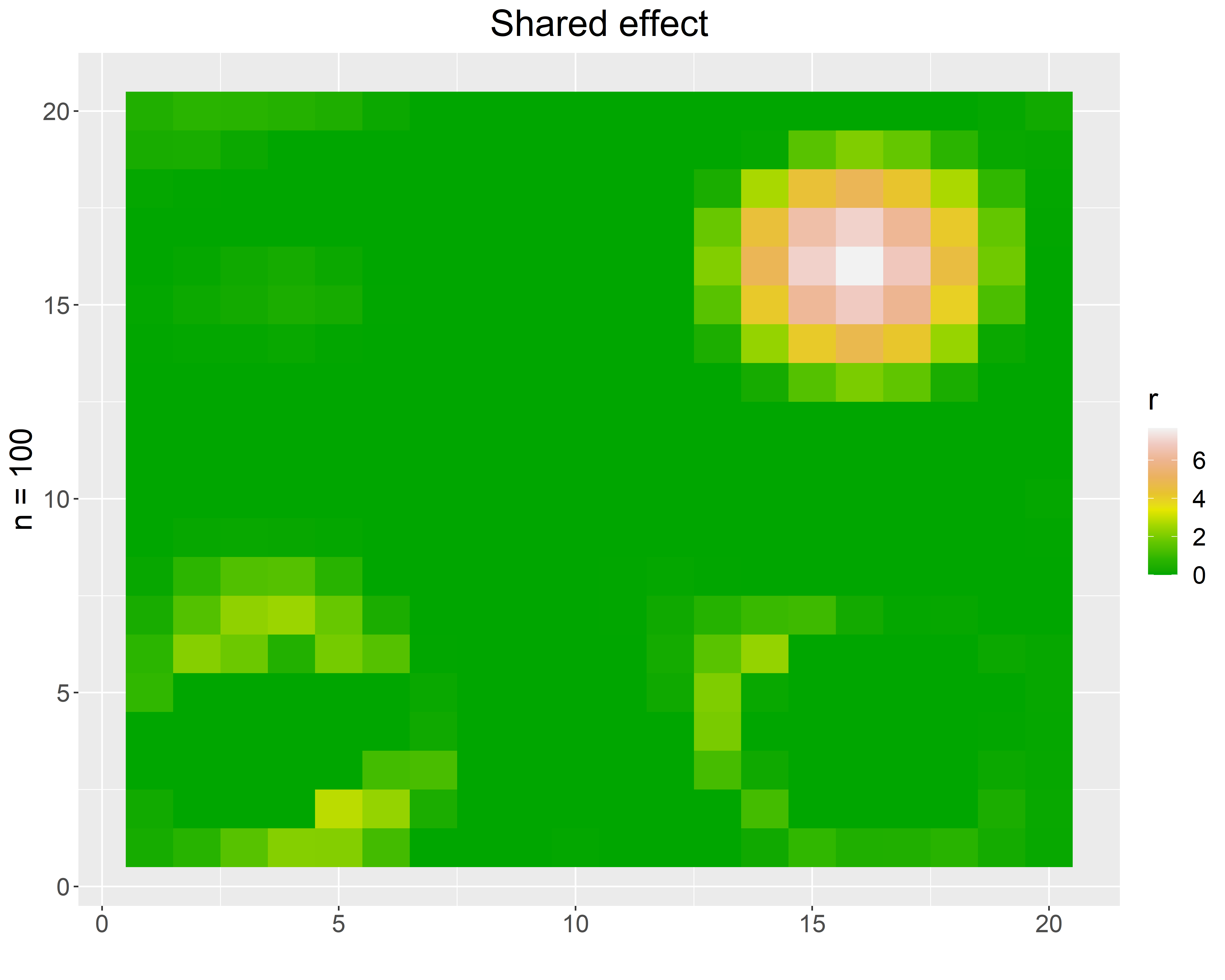}}
\caption{Estimated $\|F(h_{\lambda}(\widetilde{\bbeta}))\|_2$'s when the error variance is $1$ for two sample sizes. The two images correspond to sample sizes of 50 and 100 in each group, respectively, in Simulation case 2.}
\label{simufig2} 
\end{figure}

Additionally, due to the double soft-thresholding construction, we have the Lemma \ref{lem::priorbd} concerning the group-specific differences. When $\|\bbeta_g(\bv)\|>0$, it is easy to see that $\frac{\bbeta_g(\bv)}{\|\bbeta_g(\bv)\|_2}=\frac{\widetilde{\bbeta}_g(\bv)}{\|\widetilde{\bbeta}_g(\bv)\|_2}$. 
\begin{SLem}
For the spatial locations where $\|\bbeta_{g}(\bv)\|_2,\|\bbeta_{g'}(\bv)\|_2>0$, we have 
$$\bbeta_{g}(\bv)-\bbeta_{g'}(\bv)=\widetilde{\bbeta}_{g}(\bv)-\widetilde{\bbeta}_{g'}(\bv)-\lambda\left(\frac{\bbeta_{g}(\bv)}{\|\bbeta_{g}(\bv)\|_2}-\frac{\bbeta_{g'}(\bv)}{\|\bbeta_{g'}(\bv)\|_2}\right),$$
and
$\left(\frac{\bbeta_g(\bv)}{\|\bbeta_g(\bv)\|_2}\right)^T\frac{\bbeta_{g'}(\bv)}{\|\bbeta_{g'}(\bv)\|_2}=\left(\frac{\widetilde{\bbeta}_g(\bv)}{\|\widetilde{\bbeta}_g(\bv)\|_2}\right)^T\frac{\widetilde{\bbeta}_{g'}(\bv)}{\|\widetilde{\bbeta}_{g'}(\bv)\|_2}$.
\label{lem::priorbd}
\end{SLem}
Applying above Lemma, we further have,
$$\left|\|\bbeta_{g}(\bv)-\bbeta_{g'}(\bv)\|_2-\|\widetilde{\bbeta}_{g}(\bv)-\widetilde{\bbeta}_{g'}(\bv)\|_2\right|\leq\lambda\left\|\frac{\widetilde{\bbeta}_{g}(\bv)}{\|\widetilde{\bbeta}_{g}(\bv)\|_2}-\frac{\widetilde{\bbeta}_{g'}(\bv)}{\|\widetilde{\bbeta}_{g'}(\bv)\|_2}\right\|_2\leq 2\lambda,$$
and for the spatial locations where $\|\bbeta_1(\bv)\|_2,\|\bbeta_2(\bv)\|_2>0$, we have 
$$\left|\|\bbeta_{g}(\bv)-\bbeta_{g'}(\bv)\|_2-\|\widetilde{\bbeta}_{g}(\bv)-\widetilde{\bbeta}_{g'}(\bv)\|_2\right|\leq\lambda\left\|\frac{\bbeta_{g}(\bv)}{\|\bbeta_{g}(\bv)\|_2}-\frac{\bbeta_{g'}(\bv)}{\|\bbeta_{g'}(\bv)\|_2}\right\|_2\leq 2\lambda.$$

At a similar-effect location, we have $\left|\|\bbeta_{g}(\bv)-\bbeta_{g'}(\bv)\|_2-\|\widetilde{\bbeta}_{g}(\bv)-\widetilde{\bbeta}_{g'}(\bv)\|_2\right|\leq\lambda\sqrt{2-2\psi(\bv)}\leq \lambda\sqrt{2}$ following Definition 1.
Hence, when the angle between $\bbeta_{g}(\bv)$ and $\bbeta_{g'}(\bv)$ is small, we have $\|\bbeta_{g}(\bv)-\bbeta_{g'}(\bv)\|_2\approx \|\balpha_{g}(\bv)-\balpha_{g'}(\bv)\|_2$.
Furthermore, if $\|\bbeta_{g}(\bv)\|_2\approx\|\bbeta_{g'}(\bv)\|_2$, we have $\bbeta_{g}(\bv)-\bbeta_{g'}(\bv)\approx(\widetilde{\bbeta}_{g}(\bv)-\widetilde{\bbeta}_{g'}(\bv))\left(1-\frac{\lambda}{\|\widetilde{\bbeta}_{g}(\bv)\|_2}\right)$.
Since $\frac{\bbeta_g(\bv)}{\|\bbeta_g(\bv)\|_2}=\frac{\widetilde{\bbeta}_g(\bv)}{\|\widetilde{\bbeta}_g(\bv)\|_2}$, $\bbeta_g(\bv)^T\bbeta_{g'}(\bv)>0$ or $<0$ if and only if $\widetilde{\bbeta}_g(\bv)^T\widetilde{\bbeta}_{g'}(\bv)>0$ or $<0$ respectively. We also have $\bigm\|\frac{\bbeta_g(\bv)}{\|\bbeta_g(\bv)\|_2}-\frac{\bbeta_{g'}(\bv)}{\|\bbeta_{g'}(\bv)\|_2}\bigm\|_2=\bigm\|\frac{\widetilde{\bbeta}_g(\bv)}{\|\widetilde{\bbeta}_g(\bv)\|_2}-\frac{\widetilde{\bbeta}_{g'}(\bv)}{\|\widetilde{\bbeta}_{g'}(\bv)\|_2}\bigm\|_2$ when $\|\bbeta_1(\bv)\|_2,\|\bbeta_2(\bv)\|_2>0$.

\begin{figure}[htbp]
\centering
\subfigure[CC]{\includegraphics[width = 1\textwidth]{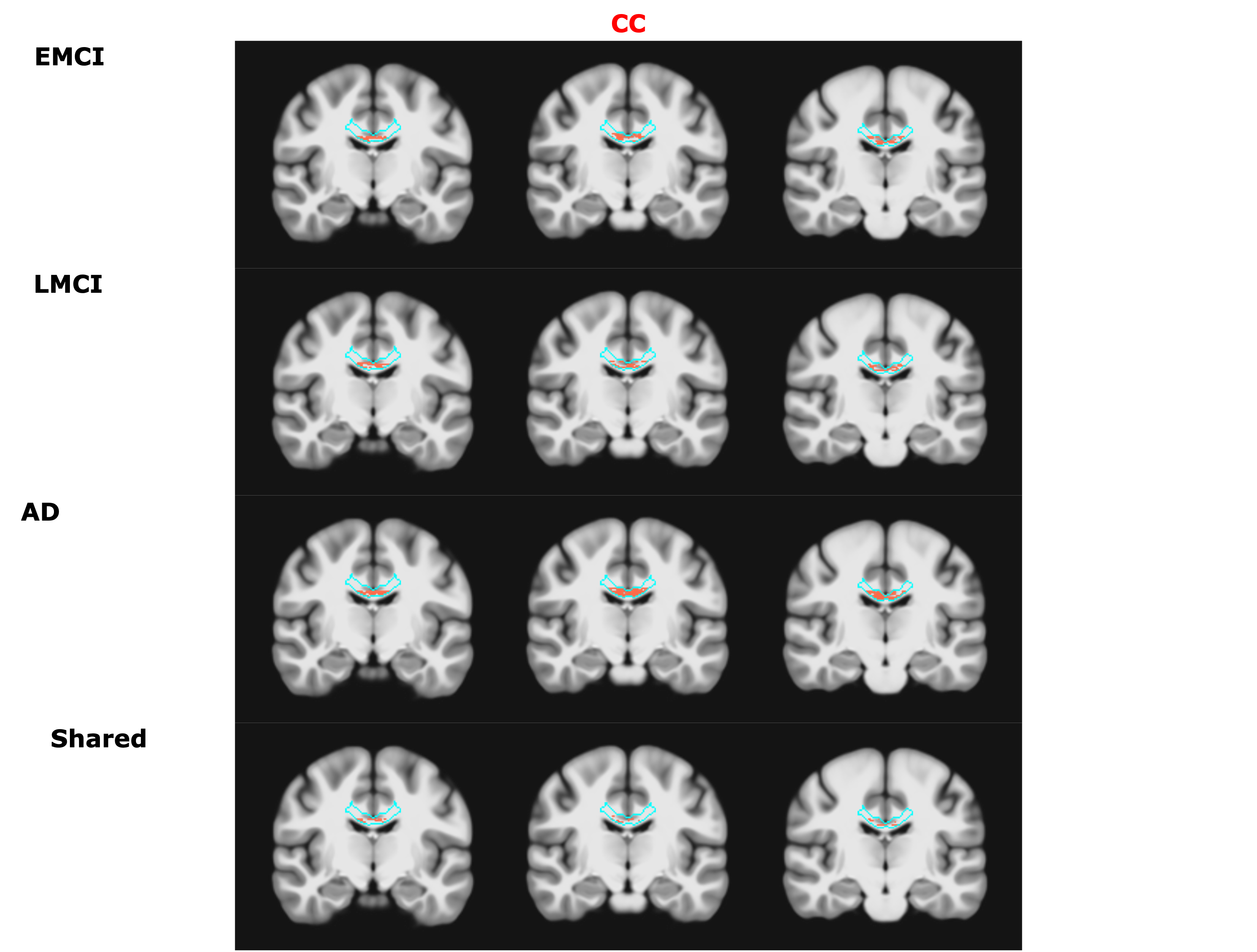}}
\caption{The red voxels are the brain regions with non-zero $\|\hat{\bbeta}_g\|_2$'s. The CC regions are delineated with light blue lines in all the images on typical slices of the brain. The last row illustrates the voxels with non-zero $\|\hat{\bbeta}_g\|_2$ for all $g$. The three columns are three slices, one from each of the frontal, middle, and posterior sections, respectively.}
\label{fig::esticoefreal1}
\end{figure}

\begin{figure}[htbp]
\centering
\subfigure[Left FPT]{\includegraphics[width = 1\textwidth]{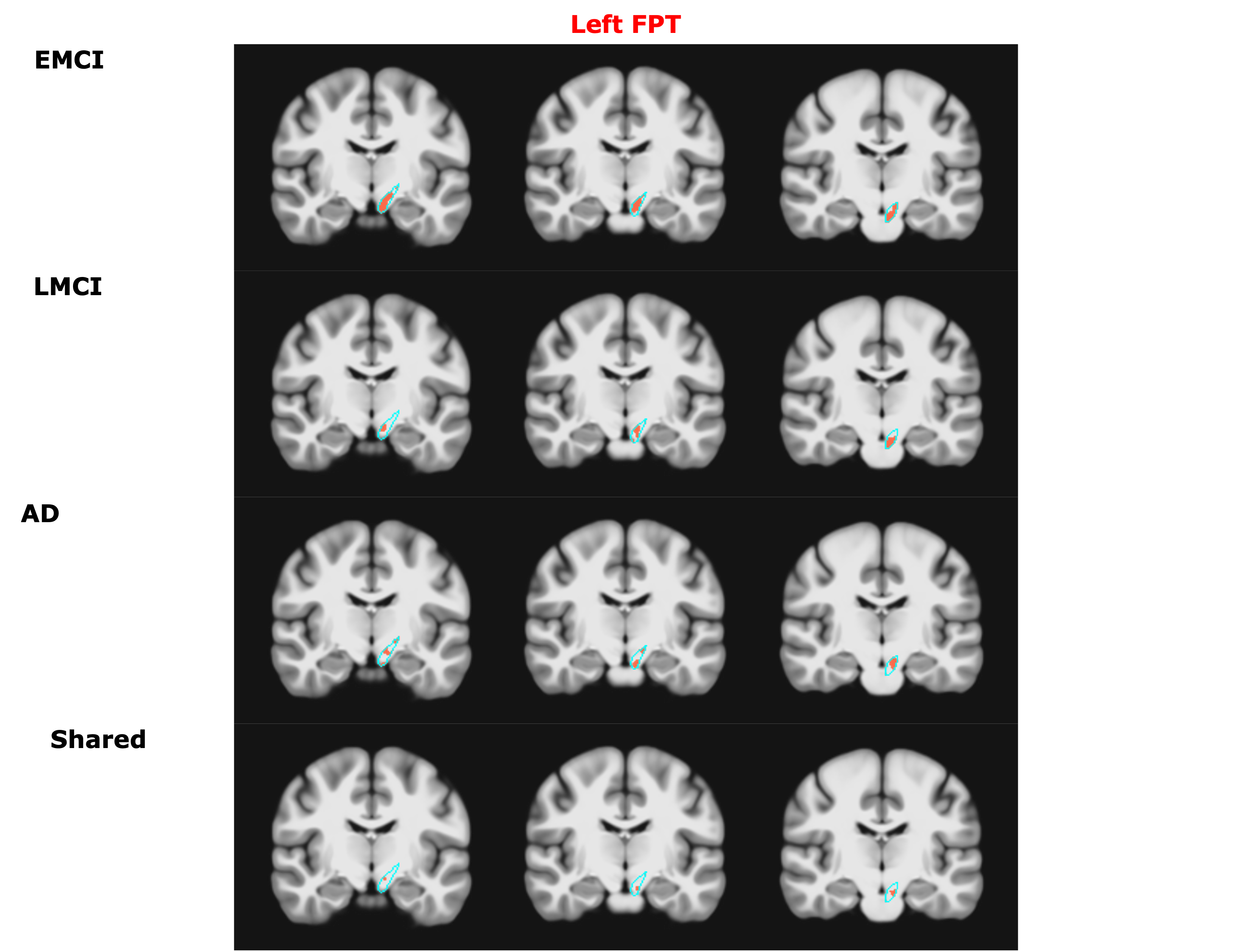}}
\subfigure[Left FPT]{\includegraphics[width = 1\textwidth]{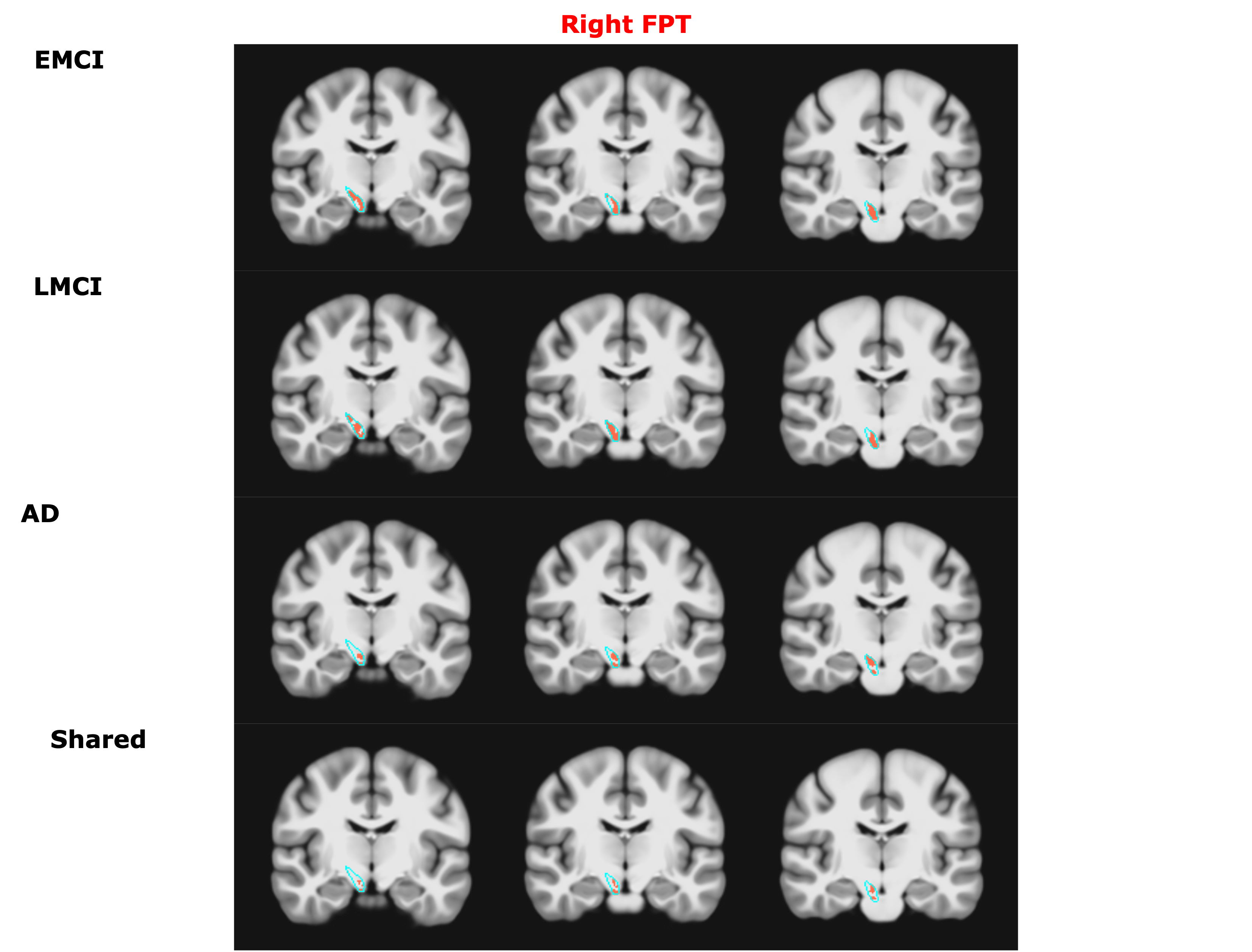}}
\caption{The red voxels are the brain regions with non-zero $\|\hat{\bbeta}_g\|_2$'s. The FPT regions are delineated with light blue lines in all the images on typical slices of the brain. The last row illustrates the voxels with non-zero $\|\hat{\bbeta}_g\|_2$ for all $g$. The three columns are three slices, one from each of the frontal, middle, and posterior sections, respectively.}
\label{fig::esticoefreal2}
\end{figure}

\section{Predictor process}
Let $\D_{i,k}=\{\bD_{i,k}(\bv_j)\}_{j=1}^p$ be the vector of data for $i$-$th$ individual for $k$-$th$ direction where $k=1,\ldots,q$. Define $n\times (qp)$ dimensional design matrix such that its $i$-$th$ row will be $\bD_i'^T=\{\D_{i,1},\ldots,\D_{i,q}\}$.
Spatial dependence of the data is the key to ensuring consistency of our Bayesian method.
Although we can easily vectorize an image data and write our image regression model as a linear regression model with $\bD_i'$ being the predictor vector of the $i$-$th$ individual, then
it is easy to show posterior convergence in terms of empirical $\ell_2$ distance.
However, such distance is not strong enough to help us establish spatial variable selection.
Instead of a compatibility condition as in \cite{castillo2015bayesian}, our proposed setup utilizes the spatial dependence of the imaging predictor.
This helps to establish better convergence results for the posterior of $\bbeta$.

Let expectation $\eE(\bD_i')=0$ and variance $V(\bD_i')=\bK$ with $0<c_{\min}\leq\eig(\bK)\leq c_{\max}<\infty$ and $K_{\ell,\ell}<\infty$. This variance assumption can be satisfied easily for $\D_{i,k}$'s being stationary ergodic processes
with spectral densities bounded between $c_{\min}$ and $c_{\max}$ or a ``spike-model" as argued in \cite{bickel2008regularized, bickel2008covariance}.
We assume that our covariance kernel follows the assumptions of Theorem 1 in \cite{bickel2008covariance}.

Let $\|\bA\|_{op}$ stands for operator norm of a matrix $\bA=(\!(a_{i,j})\!)$. 
When $\log_{}(p_n)=o(n)$, \cite{bickel2008regularized} and \cite{bickel2008covariance} showed consistency of thresholded sample covariance for a wide class of population covariance matrices. A stationary covariance kernel easily satisfy those conditions. Specifically, they showed $\big\|T_{t_n}\left(\frac{1}{n}\sum_i\bD_i'\bD_i'^T\right)-\bK\big\|_{op}\rightarrow 0$ with probability tending to 1, where $T_s()$ is a thresholding function with parameter $s$.
Hence, for large enough $n$, we have $\big\|T_{t_n}\left(\frac{1}{n}\sum_i\bD_i'\bD_i'^T\right)-\bK\big\|_{op} < c_{\min}/4$ with probability 1. 
To make our assumptions transparent, we consider the thresholding scheme from \cite{bickel2008covariance} where $T_s(\bA) = (\!(a_{i,j}\bone\{|a_{ij}| \geq s\})\!)$ and let $J_n=\frac{1}{n}\sum_i\bD_i'\bD_i'^T$. Following \cite{bickel2008covariance}, for some constant $T>0$, we set $t_n=T\sqrt{\frac{\log_{}( p_n)}{n}}\rightarrow 0$ as $\log( p_n) = o(n)$.
To control the separation between the thresholded estimate and sample covariance, we make the following assumption,
\begin{Assmp}[Predictor process]
 There exists $N$ such that for $n>N$ we have $\big\|J_n-T_{t_n}\left(J_n\right)\big\|_{op}\leq c_{\min}/4$ where $0<c_{\min}\leq\eig(\bK)\leq c_{\max}<\infty$.
 \label{assmp::predictor}
\end{Assmp}
Thus, we finally have, $\big\|J_n-\bK\big\|_{op}<\frac{c_{\min}}{2}\textrm{ for } n>N$. We know that $\|\bA\|_{op}\leq\max_i\sum_{j=1}^{p_n}|a_{i,j}|$. 
Thus, a sufficient condition for the above assumption to hold is that the number of non-zero entries with absolute value less than $t_n$ at each row of $J_n$ is upper bounded by $\frac{c_{\min}}{4t_n}(\rightarrow \infty$ as $t_n\rightarrow 0)$. 
However, the assumption would still hold as long as the total absolute contributions of the cross-correlation values of the thresholded coordinates at each row of $J_n$ are bounded by $c_{\min}/2$. Implicitly, the above assumption requires spatially uncorrelated voxels to have significantly small sample correlation values.

\section{Proofs}

\begin{proof}[Proof of Lemma 1]
We show the above result under the $\ell_2$ and $\ell_{\infty}$ norms as that will immediately help in our other theoretical results.
Specifically, we show $\|h_{\lambda}(\bx_1)-h_{\lambda}(\bx_2)\|_2\leq \|\bx_1-\bx_2\|_2$ and $\|h_{\lambda}(\bx_1)-h_{\lambda}(\bx_2)\|_\infty\leq \|\bx_1-\bx_2\|_\infty$.
We distribute the proof into three disjoint cases. Let $\S^{q-1}$ stands for the sphere with radius $\lambda$ around the origin.

Case 1: $\|\bx_1\|_2, \|\bx_2\|_2\leq \lambda$, we have $\|h_{\lambda}(\bx_1)-h_{\lambda}(\bx_2)\|_2=0\leq \|\bx_1-\bx_2\|_2$.

Case 2: $\|\bx_1\|_2\geq \lambda, \|\bx_2\|_2\leq \lambda$, We have $\|h_{\lambda}(\bx_1)-h_{\lambda}(\bx_2)\|_2=\|h_{\lambda}(\bx_1)\|_2=\|\bx_1\|_2-\lambda$.

We have $ \|\bx_1-\bx_2\|_2\geq \big\|\bx_1-\|\bx_2\|_2\frac{\bx_1}{\|\bx_1\|_2}\big\|_2$ by Cauchy–Schwarz inequality which gives us $(\sum_i{x_{1,i}^2})(\sum_i{x_{1,i}^2})\geq(\sum_ix_{1,i}x_{2,i})^2$.
Since $\|\bx_2\|_2<\lambda$, we have $\big\|\bx_1-\|\bx_2\|_2\frac{\bx_1}{\|\bx_1\|_2}\big\|_2>\big\|\bx_1\|_2-\lambda$.


Case 3: $\|\bx_1\|_2\geq \lambda, \|\bx_2\|_2\geq \lambda$,
\begin{align*}
    \|\bx_1-\bx_2\|^2_2-\|h_{\lambda}(\bx_1)-h_{\lambda}(\bx_2)\|^2_2\geq 2\lambda+2<h_{\lambda}(\bx_1),h_{\lambda}(\bx_2)>-2<\bx_1,\bx_2>,
\end{align*}
where $<,>$ stands for the inner-product between two vectors. 
The angle between $h_{\lambda}(\bx_1)$ and $h_{\lambda}(\bx_2)$ is the same as the one between $\bx_1$ and $\bx_2$.
Let that angle be $\vartheta$. 
Next we have $\|\bx_1\|_2\|\bx_2\|_2-\|h_{\lambda}(\bx_1)\|_2\|h_{\lambda}(\bx_2)\|_2>\lambda$.
Thus, $2\lambda+2<h_{\lambda}(\bx_1),h_{\lambda}(\bx_2)>-2<\bx_1,\bx_2>\geq2\lambda(1-\cos\vartheta)\geq0$.
Thus, $\|h_{\lambda}(\bx_1)-h_{\lambda}(\bx_2)\|_2\leq \|\bx_1-\bx_2\|_2$.

We have, $\|\bx_1-\bx_2\|_\infty\leq\|\bx_1-\bx_2\|_2\leq q\|\bx_1-\bx_2\|_\infty$ which automatically establishes Lipchitz property with respect to $\ell_{\infty}$ norm.

\end{proof}

\begin{proof}[Proof of Lemma S1]

Since, $\frac{\bbeta_{2}(\bv)}{\|\bbeta_{2}(\bv)\|_2}^T\bD(\bv)=\frac{\bbeta_{1}(\bv)}{\|\bbeta_{1}(\bv)\|_2}^T\bD(\bv)-\left(\frac{\bbeta_{1}(\bv)}{\|\bbeta_{1}(\bv)\|_2}-\frac{\bbeta_{2}(\bv)}{\|\bbeta_{2}(\bv)\|_2}\right)^T\bD(\bv)$,
we have, $P(\frac{\bbeta_{1}(\bv)}{\|\bbeta_{1}(\bv)\|_2}^T\bD(\bv)>0\mid \frac{\bbeta_{2}(\bv)}{\|\bbeta_{2}(\bv)\|_2}^T\bD(\bv)<0)\leq P(\|\bD(\bv)\|_2\sqrt{2-2c}>\frac{\bbeta_{1}(\bv)}{\|\bbeta_{1}(\bv)\|_2}^T\bD(\bv)>0\mid \frac{\bbeta_{2}(\bv)}{\|\bbeta_{2}(\bv)\|_2}^T\bD(\bv)<0)=P(\sqrt{2-2c}>\frac{\bbeta_{1}(\bv)}{\|\bbeta_{1}(\bv)\|_2}^T\frac{\bD(\bv)}{\|\bD(\bv)\|_2}>0\mid \frac{\bbeta_{2}(\bv)}{\|\bbeta_{2}(\bv)\|_2}^T\bD(\bv)<0)$. The interval $(0, \sqrt{2-2c})$ reduces to a null set as $c$ increases to 1, and thus the probability decreases to zero.
Here we applied Cauchy-Squartz inequality to get $$|\bD(\bv)^T\left(\frac{\bbeta_{2}(\bv)}{\|\bbeta_{2}(\bv)\|_2}-\frac{\bbeta_{1}(\bv)}{\|\bbeta_{1}(\bv)\|_2}\right)|\leq \|\bD(\bv)\|_2\sqrt{(2-2\psi(\bv))}.$$
\end{proof}

\begin{proof}[Proof of Lemma ~S\ref{lem::priorshare}]
Since we put Gaussian process priors on $\widetilde{\bbeta}(\cdot)$ and $\widetilde{\balpha}_g(\cdot)$, for a given spatial location $\bv$, $\|\widetilde{\bbeta}(\bv)\|_2^2$ and $\|\widetilde{\balpha}_g(\bv)\|_2^2$ marginally follow $\chi^2$-distributions.

For simplicity, let $G=2$.
Then for $\lambda_g=0$ the likelihood 
$P(\widetilde{\bbeta}(\bv), \widetilde{\balpha}_1(\bv)=\widetilde{\bbeta}_1(\bv)-\widetilde{\bbeta}(\bv), \widetilde{\balpha}_2(\bv)=\widetilde{\bbeta}_2(\bv)-\widetilde{\bbeta}(\bv))$ is maximized when $\widetilde{\bbeta}(\bv)=\frac{\widetilde{\bbeta}_1(\bv)/\sigma_1^2+\widetilde{\bbeta}_2(\bv)/\sigma_2^2}{1/\sigma^2+1/\sigma_1^2+1/\sigma_2^2}$ assuming that $\|\widetilde{\bbeta}(\bv)\|_2^2/\sigma_S^2, \|\widetilde{\balpha}_1(\bv)\|_2^2/\sigma_1, \|\widetilde{\balpha}_2(\bv)\|_2^2/\sigma_2\sim\chi^2(q)$.
Here $\sigma_S$, $\sigma_1$ and $\sigma_2$ are the marginal prior variances of $\widetilde{\bbeta}(\bv)$, $\widetilde{\bbeta}_1(\bv)$ and $\widetilde{\bbeta}_2(\bv)$ respectively.

The above solution of $\widetilde{\bbeta}(\bv)$ will have the same direction of effect as $\bbeta_g$'s. Hence, the posterior mode of $\widetilde{\bbeta}(\bv)$ at a similar-effect location $\bv$ has the same direction as the direction of similar-effect.
\end{proof}

\begin{proof}[Proof of Lemma ~S\ref{lem::priorbd}]
Let $\widetilde{\bbeta}_g(\bv)=\widetilde{\bbeta}(\bv) +\balpha_g(\bv).$

Then, we have that $\bbeta_g(\bv)=\widetilde{\bbeta}_g(\bv)\left(1-\frac{\lambda}{\|\widetilde{\bbeta}_g(\bv)\|_2}\right)_{+}$. 
For the spatial locations where $\|\bbeta_1(\bv)\|_2,\|\bbeta_2(\bv)\|_2>0$, we have $\widetilde{\bbeta}_g(\bv)=\bbeta_g(\bv)\left(1+\frac{\lambda}{\|\bbeta_g(\bv)\|_2}\right)$ and thus

\begin{align*}
    &\|\bbeta_1(\bv)-\bbeta_2(\bv)\|_2-\lambda\left\|\frac{\bbeta_1(\bv)}{\|\bbeta_1(\bv)\|_2}-\frac{\bbeta_2(\bv)}{\|\bbeta_2(\bv)\|_2}\right\|_2\\&\leq\|\widetilde{\bbeta}_1(\bv)-\widetilde{\bbeta}_2(\bv)\|_2\\&\quad\leq\|\bbeta_1(\bv)-\bbeta_2(\bv)\|_2+\lambda\left\|\frac{\bbeta_1(\bv)}{\|\bbeta_1(\bv)\|_2}-\frac{\bbeta_2(\bv)}{\|\bbeta_2(\bv)\|_2}\right\|_2,
    \end{align*}
    which implies
    \begin{align*}
    &\|\bbeta_1(\bv)-\bbeta_2(\bv)\|_2-\lambda\left\|\frac{\bbeta_1(\bv)}{\|\bbeta_1(\bv)\|_2}-\frac{\bbeta_2(\bv)}{\|\bbeta_2(\bv)\|_2}\right\|_2\\&\leq\|\balpha_1(\bv)-\balpha_2(\bv)\|_2\\&\quad\leq\|\bbeta_1(\bv)-\bbeta_2(\bv)\|_2+\lambda\left\|\frac{\bbeta_1(\bv)}{\|\bbeta_1(\bv)\|_2}-\frac{\bbeta_2(\bv)}{\|\bbeta_2(\bv)\|_2}\right\|_2
\end{align*}
\end{proof}

\begin{proof}[Proof of Theorem 1]
By linear algebra, it is easy to show that $\bigm|\|\bbeta(\bv)\|_2-\|\bbeta_0(\bv)\|_2\bigm|\leq \|\bbeta(\bv)-\bbeta_0(\bv)\|_2$.
Hence, $\Pi(\big\|\|\bbeta\|_2-\|\bbeta_0\|_2\big\|_{\infty}<\epsilon)>\Pi(\big\|\|\bbeta-\bbeta_0\|_2\big\|_{\infty}<\epsilon)$.

\noindent
We have $\|\bbeta_0(\bv)-\bbeta(\bv)\|_{2}\leq \|\bbeta_0(\bv)-h_{\lambda}(\bbeta_0(\bv))\|_{2}+\|h_\lambda(\widetilde{\bbeta}(\bv))-h_{\lambda}(\bbeta_0(\bv))\|_{2}$.

For the first term, $\|\bbeta_0(\bv)-h_{\lambda}(\bbeta_0(\bv))\|_{2}\leq \|\bbeta_0(\bv)\|_2\bone\{\|\bbeta_0(\bv)\|_2\leq\lambda\}+\lambda\bone\{\|\bbeta_0(\bv)\|_2>\lambda\}\leq\lambda$.

Now the second term,
$\|h_\lambda(\widetilde{\bbeta}(\bv))-h_{\lambda}(\bbeta_0(\bv))\|_{2}\leq\|\widetilde{\bbeta}(\bv)-\bbeta_0(\bv)\|_{2}$, using Lemma 1.

Hence, $$\Pi(\sup_{\bv}\|\bbeta_0(\bv)-\bbeta(\bv)\|_{2}<\epsilon)>\Pi(\lambda<\epsilon/2)\Pi(\sup_{\bv}\|\widetilde{\bbeta}(\bv)-\bbeta_0(\bv)\|_{2}<\epsilon/2).$$

Due to the uniform prior on $\lambda$, we have $\Pi(\lambda<\epsilon/2)>0$.
Furthermore,
$\Pi(\sup_{\bv}\|\widetilde{\bbeta}(\bv)-\bbeta_0(\bv)\|_{2}<\epsilon/2)>\prod_{j=1}^q\Pi(\|\widetilde{\beta}_{j}-\beta_{0,j}\|_{\infty}<\frac{\epsilon}{2q})$.

To bound $\Pi(\|\widetilde{\beta}_{j}-\beta_{0,j}\|_{\infty}<\frac{\epsilon}{2q}\mid a)$, we use its connection with the concentration function $\phi^a_{\beta_{0,j}}(\frac{\epsilon}{2q})=-\log_{} \Pi (\|\widetilde{\bbeta}_j\|_{\infty} \le \frac{\epsilon}{2q}\mid a) + \frac{1}{2}\inf\{\|h\|_{\mathbb{H}^a}^2 : h\in \mathbb{H}^a,\|h-\beta_{0,j}\|_{\infty} \leq \frac{\epsilon}{2q}\}$.
Here $\mathbb{H}^a$ is the RKHS attached to the Kernel $\kappa_{a}$
Specifically, we apply  Lemma I.28 of \cite{ghosal2017fundamentals} which implies that $\phi^a_{\beta_{0,j}}(\frac{\epsilon}{2q})\leq -\log_{}\Pi(\|\widetilde{\beta}_{j}-\beta_{0,j}\|_{\infty}<\frac{\epsilon}{2q}\mid a)\leq \phi^a_{\beta_{0,j}}(\frac{\epsilon}{4q})$.

Applying the Lemmas 4.3 and 4.6 from \cite{van2009adaptive}, it is shown that for $\alpha$-smooth H\"older function, for any $a_0>0$ there exist positive constants
$\epsilon_0 < 1/2, C, D$ and $K$ that depend on the true function $\beta_{0,j}$ such that, for $a > a_0$, $\epsilon' < \epsilon_0$ and $\epsilon' > Ca^{-\alpha}$, we have $\phi^a_{\beta_{0,j}}(\epsilon')\lesssim a^d\left(\log_{}\frac{a}{\epsilon'}\right)^{1+d}$.

Then we have $\Pi(\|\widetilde{\beta}_{j}-\beta_{0,j}\|_{\infty}<\frac{\epsilon'}{2q}\mid a)\geq \int_{(\frac{C}{\epsilon'})^{1/\alpha}}^{2(\frac{C}{\epsilon'})^{1/\alpha}}\exp\left(-\phi^a_{\beta_{0,j}}(\frac{\epsilon'}{4q})\right)g(a)da$.

Since $\phi^a_{\beta_{0,j}}(\epsilon')$ is an increasing function in $a$, we can further lower bound the expression in the above display as $$\int_{(\frac{C}{\epsilon'})^{1/\alpha}}^{2(\frac{C}{\epsilon'})^{1/\alpha}}\exp\left(-\phi^a_{\beta_{0,j}}(\epsilon'/4q)\right)g(a)da\geq C_1\exp\left(-C_2(C/\epsilon')^{d/\alpha}\left(\log_{}\frac{1}{\epsilon'}\right)^{1+d}\right)(C/\epsilon')^{s/\alpha+1/\alpha},$$
for some constants $C_1$ and $C_2$. Choosing $\epsilon'<\min\{\epsilon_0,\epsilon\}$, we have $\prod_{j=1}^q\Pi(\|\widetilde{\beta}_{j}-\beta_{0,j}\|_{\infty}<\frac{\epsilon}{2q})>0$

We have $\sqrt{q}\|\bbeta_0-\bbeta\|_{\infty}\geq \sup_{\bv}\|\bbeta_0(\bv)-\bbeta(\bv)\|_{2}\geq \|\bbeta_0-\bbeta\|_{\infty}$. Hence, the two distances are equivalent and thus, the last assertion of the theorem also follows trivially.
\end{proof}

\begin{proof}[Proof of Theorem 2]
To derive the posterior contraction rate $\epsilon_n$ asserted in Theorem~2, we apply the general theory of posterior contraction for independent, not-identically distributed observations as in Section~8.3 of \cite{ghosal2017fundamentals}. This requires verifying certain that prior concentration in a ball of size $\epsilon_n^2$ in the sense of Kullback-Leibler divergence, existence of tests with error probabilities at most $e^{-c_1 n\epsilon_n^2}$ for testing the true distribution against a ball of size of the order $\epsilon_n$ separated by more than $\epsilon_n$ from the truth in terms of the metric $d$, a ``sieve'' in the parameter space which contains at least $1-e^{-c_2 n\epsilon_n^2}$ prior probability and which can be covered by at most $e^{c_3 n\epsilon_n^2}$ for some constants $c_1,c_3>0$ and $c_2>c_1+2$. As the fixed dimensional parameters $\sigma$ contribute $e^{-c_1\log n}$.

Recall that for $q, q^* \in \mathcal{P}$, let the Kullback-Leibler divergences be given by 
\begin{gather*}
K(q^*, q) = \int q^*\log_{}{\frac{q^*}{q}} \qquad  V(q^*, q) = \int q^*\log_{}^2{\frac{q^*}{q}}.
\end{gather*}

Let $P_{0i}$ denote the true distribution with density $p_{0i}$ for $i$th individual. For each individual $i$, we consider the observation $y_{i} \sim \Normal(\mu_i, \sigma^2)$, where $\mu_i = \frac{1}{p_n}\sum_{j=1}^{p_n}\bD_i(\bv_j)^T\bbeta(\bv_j)$. 
Let $q_{0,i}=\Normal(\mu_{0,i}, \sigma^{2})$ and $q_i=\Normal(\mu_i, \sigma^2)$.
{ \begin{align*}
&K(q_i^*,q_i) =\frac{{1}}{2}\log\bigg(\frac{\sigma}{\sigma^*}\bigg) - \frac{1}{2}\bigg[1-\frac{(\mu_i^*-\mu_i)'(\mu_i^*-\mu_i)}{\sigma^2} - \frac{\sigma^{*2}}{\sigma^2}\bigg],\\
&V(q_i^*,q_i) = \frac{{1}}{2}\bigg(\frac{\sigma^{*2}}{\sigma^2}-1\bigg)^2+\frac{(\mu_i^*-\mu_i)'(\mu_i^*-\mu_i)}{\sigma^4}\sigma^{*2}.
\end{align*}
Let $\mu = (\mu_1,\ldots,\mu_n)$ with all the means stacked together. Let $\mu_0$ be the true value and $\sigma_0$ be that of $\sigma$. Applying Lemma L.4 of \cite{ghosal2017fundamentals}, we have $K(q_0^n,q^n)\lesssim \frac12 \|\bmu_0-\bmu\|_2^2+n|\sigma_0^2-\sigma^2|$ and $V(q_0^n,q^n)\lesssim \|\bmu_0-\bmu\|_2^2+n|\sigma_0^2-\sigma^2|$, where $q^n=\prod_{i=1}^nq_{i},q_0^n=\prod_{i=1}^nq_{0,i},\bmu_0=\{\mu_{0,1},\ldots,\mu_{0,n}\}$ and $\bmu=\{\mu_{1},\ldots,\mu_{n}\}$.




Since, $\max\{K(q_0^n,q^n),V(q_0^n,q^n)\lesssim \|\bmu_0-\bmu\|_2^2\}$, we have 
\begin{align*}
    &-\log_{}\Pi(K(q_0^n,q^n)\leq n\epsilon_n^2,V(q_0^n,q^n)\leq n\epsilon_n^2\mid \B_n) \\&\quad\leq -\log_{}\Pi(\|\bmu_0-\bmu\|_2^2+n|\sigma^2-(\sigma^*)^2| \lesssim n\epsilon_n^2\mid \B_n)\\&\quad\leq -\log_{}\Pi(d(\bbeta,\bbeta_0) \lesssim \epsilon_n^2\mid \B_n)-\log_{}\Pi(|\sigma^2-(\sigma^*)^2| \lesssim \epsilon_n^2\mid \B_n).
\end{align*}


The fixed dimensional parameters $\sigma$
contribute $e^{-c'\log n}$. Hence,
the contraction rate will be governed by $\bbeta$.
We thus need:
\begin{itemize}
    \item[(i)] (Prior mass condition) $-\log_{}\Pi(\frac{1}{p_n}\sum_{j=1}^{p_n}\|\bbeta_0(\bv_j)-\bbeta(\bv_j)\|^2_{2} \lesssim \epsilon_n^2)\leq c_1n\epsilon_n^2$
    \item[(ii)](Sieve) construct the sieve $\H_n$ such that $\Pi(\H_n^c)\leq \exp(-C_2n\epsilon_n^2)$ and
    \item[(iii)](Test construction) exponentially consistent tests $\chi_n$.
\end{itemize}

In regard to the test construction, let $\bmu^*$ be a value such that $\frac{1}{n}\|\bmu^*-\bmu_0\|_2^2>\epsilon^2_n$ and $\sigma^*$ be such that $|\sigma^*-\sigma_0|>\epsilon_n$. If $\sigma$ can be assumed to take values in a fixed bounded set, then by Lemma 8.27 of \cite{ghosal2017fundamentals}, it follows that there exists a test for $(\bmu_0,\sigma_0)$ against $\{(\bmu,\sigma): \|\bmu-\bmu^*\|_2\le \|\bmu_0-\bmu^*\|_2/2, |\sigma-\sigma^*|\le |\sigma_0-\sigma^*|/2\}$ with type I and type II error probabilities bounded by $\exp(-K[\|\bmu_0-\bmu^*\|_2^2+n(\sigma_0-\sigma^*)^2]/(\max\{\sigma^2_0,(\sigma^*)^2\}))$ with a universal constant $K$ and it leads to $e^{-c_1 n\epsilon_n^2}$ for a constant $c_1>0$ due to the boundedness condition on $\sigma$. The boundedness condition on $\sigma$ can be disposed of by directly working with the negative log-affinity divergence, but for simplicity of the arguments, we present the proof assuming the additional boundedness. }
Hence, we have an exponentially consistent test for $\frac{1}{n}\|\bmu-\bmu_0\|^2_2+(\sigma-\sigma_0)^2$ and $\frac{1}{n}\|\bmu-\bmu_0\|^2_2\leq \frac{1}{n}\|\bmu-\bmu_0\|^2_2+(\sigma-\sigma_0)^2$.

We then have,
\begin{align*}
   \frac{1}{p_n}\sum_{j=1}^{p_n}\|\bbeta_0(\bv_j)-\bbeta(\bv_j)\|^2_{2}\lesssim \frac{1}{n}\|\bmu-\bmu_0\|^2_2\lesssim \frac{1}{p_n}\sum_{j=1}^{p_n}\|\bbeta_0(\bv_j)-\bbeta(\bv_j)\|^2_{2},
\end{align*}
due to Assumption 2.
Hence, the two distances are equivalent.
We thus apply the general theory of posterior consistency to show $\eE_0\Pi(\kappa: d(\bbeta,\bbeta_0)>M_n \epsilon_n | S_n)\rightarrow 0$

We have $\frac{1}{p_n}\sum_{j=1}^{p_n}\|\bbeta_0(\bv_j)-\bbeta(\bv_j)\|^2_{2}\leq\lambda^2\frac{p_n'}{p_n}+\frac{1}{p_n}\sum_{j=1}^{p_n}\|\bbeta_0(\bv_j)-\widetilde{\bbeta}(\bv_j)\|^2_{2}\lesssim\lambda^2+\big\|\|\widetilde{\bbeta}-\bbeta_0\|^2_2\big\|^2_{\infty}$.

We have $\Pi(d(\bbeta,\bbeta_0)\lesssim \epsilon_n^2)>\Pi\left(\lambda^2<\epsilon^2_n/2\right)\Pi(\sum_{j=1}^q\|\beta_{0,j}-\widetilde{\beta}_j\|^2_{\infty}<\epsilon^2_n/2)$. 

We thus need $\lambda_n\leq\frac{1}{\sqrt{2}} \epsilon_n$.
Following the proof steps from Theorem 1, to ensure $-\log_{}\Pi(\|\beta_{0,j}-\widetilde{\beta}_j\|^2_{\infty}<\frac{\epsilon_n^2}{2q})\lesssim n\epsilon_n^2$, we need $\epsilon_n$ to be some large multiple of $n^{-\alpha/(2\alpha+d)}(\log_{} n)^{\nu}$, where $\nu=(1+d)/(2+d/\alpha)$.



We can further show that $\big\|\|\bbeta-\bbeta_0\|^2_2\big\|^2_{\infty}\leq\sum_{j=1}^q\|\beta_{0,j}-\beta_j\|^2_{\infty}$.
Hence, the sieve construction and prior concentration results are established with respect to $\|\cdot\|_{\infty}$ norm on each component function.


No sieve is needed for $\sigma$ as its support is bounded, and we choose the following sieve for the rest, $\H_n=\{\widetilde{\bbeta}:\widetilde{\beta}_{j}\in \left(M_n\sqrt{\frac{\xi}{\delta}}\mathbb{H}^{\xi}_1 + \epsilon\mathbb{B}_1\right)\cup\left(\cup_{a<\delta}(M_n\mathbb{H}^{a}_1) + \epsilon\mathbb{B}_1\right), |\widetilde{\bbeta}_j'|_{\infty}<M_{1,n}, A_{1,n}\le a \leq A_{2,n}\}$, where $\mathbb{B}_1$ is the unit ball of $C[0, 1]^d$, the space of all continuous functions $f:[0,1]^d\rightarrow\mathbb{R}$ and $\widetilde{\bbeta}'$ is the first derivative process.
We set $A_{1,n}=\exp(-C_1n\epsilon_n^2)$ and $A_{2,n}=n^C$ with $C<\min\{1/2,\frac{\alpha}{2d}\}$. Then $M_nA_{2,n}\leq M_{1,n} \leq (n\epsilon_n^2)^{\alpha/d}\epsilon_n$.
Then $M_n^2=Q_1n\epsilon_n^2(\log(\xi_n/\epsilon_n))^{1+d}$ for a constant $Q_1>0$ where $\xi_n^d=Q_2n\epsilon_n^2$ for some constant $Q_2>0$ for an $\epsilon_n$ such that $n\epsilon_n^2$ diverges to infinity. The role of $\xi_n$ is discussed in the later part of the proof. The above choices allow us to bound the probability of sieve-complement with exponentially decaying probability and have it covered by at most $e^{c_3n\epsilon_n^2}$ balls of size $\epsilon_n^2$.

Our scaled covariance kernel belongs to the class considered in \cite{ghosal2006posterior}.
Hence, borrowing the results from \cite{van2009adaptive} and Theorem 5 of \cite{ghosal2006posterior}, we have, $\Pi(\widetilde{\bbeta}\notin\H_n)\leq \Pi(a>\xi)+\int_0^{\xi}\Pi(\widetilde{\bbeta}\notin\H_n\mid a)g(a)\pd a + Ke^{-d_1M^2_{1,n}/\sigma_1^2}$, with a constant $d_1$ and $\sigma^2_1=\sup_{\bv}$Var$(\tilde{\beta}_j'(\bv))\leq D_1 a^2\le D_1A^2_{2,n}$, where $D_1=\sup_{\bv}$Var$(\xi(\bv))<\infty$ \citep{ghosal2006posterior}. Here $\xi$ follows a GP with a squared exponential covariance kernel and unit marginal variance. Then $Ke^{-d_1M^2_{1,n}/\sigma_1^2}\leq Ke^{-d_1'M^2_{n}}$ setting $M_{1,n}=M_nA_{2,n}$ and $d_1'=d_1/D_1$. 

For each $j$, we have $\Pi(\widetilde{\beta}_{j}\notin\H_n\mid a)\leq 1 - \Phi(\Phi^{-1}(e^{-\phi_0^{\xi}(\epsilon)}) + M_n)$. If 
$M_n^2\gtrsim \xi^d(\log_{}(\xi/\epsilon))^{1+d}$, then $\Pi(\widetilde{\beta}_{j}\notin\H_n\mid a)\leq e^{-d_2M_n^2}$ where $d_2=\min\{d'_1,1/8\}$. 

Due to the above choices on $A_{1,n}$ and $A_{2,n}$, the contribution of $a$ in sieve complement probability decays exponentially.  


Then, by varying above $\xi$ with $n$, we obtain $\Pi(\widetilde{\beta}_{j}\notin\H_n)\leq \Pi(a>\xi_n) + 2e^{-d_2M_n^2}$.

Hence, a sufficient condition for $\Pi(\widetilde{\beta}_{j}\notin\H_n)\leq e^{-n\epsilon_n^2}$ is that 1) $M_n^2\gtrsim n\epsilon_n^2$ due to the second part and for the first part we need 2) $\xi^d\gtrsim n\epsilon_n^2$ and 3) $\xi^d-\log(\xi^{s-d+1})\gtrsim n\epsilon_n^2$ on applying Lemma 4.9 of \cite{van2009adaptive}.
For large $n$, the above mentioned choice of $(\xi_n,M_n)$ satifies 1), 2), 3) and $M_n^2\gtrsim \xi^d(\log_{}(\xi/\epsilon))^{1+d}$.

Following the proof steps of the results from \cite{van2009adaptive} and Lemma 2 of \cite{ghosal2006posterior}, we can obtain the entropy bounds as $\log_{} N(3\epsilon,\H_n,\|\cdot\|_{\infty})\lesssim q\left\{\xi^d\left(\log_{}\left(\frac{M_n^{3/2}\sqrt{2\tau \xi}d^{1/4}}{\epsilon^{3/2}}\right)\right)^{1+d}+\log_{}\frac{2M_n}{\epsilon}+\left(\frac{M_{1,n}}{\epsilon}\right)^{d/\alpha}\right\}$, where $\tau^2 =\int\|\iota\|^2d\gamma(\iota)$ as long as $M_n^{3/2}\sqrt{2\tau \xi}d^{1/4}>2\epsilon^{3/2}, \xi>a_0$ for some $a_0>0$ and $M_n\geq\epsilon$.
These conditions are also satisfied by the above solutions of $M_n$ and $\xi_n$.

Thus, under the conditions 1), 2) and 3) specified above, we have $\log_{} N(3\bar{\epsilon}_n,\H_n,\|\cdot\|_{\infty})\lesssim n\bar{\epsilon}_n^2$, where $\bar{\epsilon}_n$ is a large multiple of $\epsilon_n(\log_{} n)^{(1+d)/2}$. Then the final rate will be $\max\{\epsilon_n,\bar{\epsilon}_n\}=\bar{\epsilon}_n=n^{-\alpha/(2\alpha+d)}(\log_{} n)^{\nu'}$, where $\nu'=(4\alpha+d)/(4\alpha+2d)$

{\underline{Connecting to the other distance metric:}} We have $\frac{1}{p_n}\sum_{j=1}^{p_n}\|\bbeta_0(\bv_j)-\bbeta(\bv_j)\|^2_{2}\leq\epsilon_n^2$ implies $\frac{1}{p_n}\sum_{j=1}^{p_n}\bigm|\|\bbeta_0(\bv_j)\|_2-\|\bbeta(\bv_j)\|_2\bigm|\leq\epsilon_n$.

 For $\bbeta,\bbeta_1\in\H_n$, we have $\bigm|d_{n}(\bbeta,\bbeta_1)- \frac{1}{p_n}\sum_{j=1}^{p_n}\|\bbeta_1(\bv_j)-\bbeta(\bv_j)\|_2 \bigm| \leq C\frac{M_{1,n}+M_{1,n}}{p_n^{1/d}}$ based on Lemma 2 of the main paper. For the above choice of $M_{1,n}$, we have $\frac{M_{1,n}+M_{1,n}}{p_n^{1/d}}\lesssim \epsilon_n/2$ as $p_n^{1/d}\geq Q_4 n$ for some constant $Q_4>0$ and $M_{1,n} \lesssim n\epsilon_n$ by construction.
Constants can be adjusted so that, $C\frac{M_{1,n}+M_{1,n}}{p_n^{1/d}}\leq \epsilon_n/2$ and thus a sufficient condition for, $d_{n}(\bbeta,\bbeta_1)\leq\epsilon_n$ is to have $\frac{1}{p_n}\sum_{j=1}^{p_n}\|\bbeta_1(\bv_j)-\bbeta(\bv_j)\|_2\leq \epsilon_n/2$ when $\bbeta,\bbeta_1\in\H_n$. 
Thus, the contraction rate derived above holds for, $d_{n}(\bbeta,\bbeta_0)$ too.
Similarly, it also holds for $\frac{1}{p_n}\sum_{j=1}^{p_n}\bigm|\|\bbeta_0(\bv_j)\|_2-\|\bbeta(\bv_j)\|_2\bigm|$ and  $\int\bigm|\|\bbeta_0(\bv)\|_2-\|\bbeta(\bv)\|_2\bigm|\pd\bv$ again applying Lemma 2 of the main paper.


\end{proof}

\begin{proof}[Proof of Theorem 3]

Let $U_{n,m}=\{j:\|\bbeta(\bv_j)\|>1/m,\|\bbeta_0(\bv_j)\|=0\}$ for $m\in\mathbb{N}$.
Then $\frac{1}{|U_{n,m}|}\sum_{j\in U_{n,m}}\|\bbeta(\bv_j)\|_2>1/m$.

Now, we can write $\frac{1}{p_n}\sum_{j\in U_{n,m}}\|\bbeta(\bv_j)\|_2 = \frac{|R_{0,n}|}{p_n}\frac{|U_{n,m}|}{|R_{0,n}|}\frac{1}{|U_{n,m}|}\sum_{j\in U_{n,m}}\|\bbeta(\bv_j)\|_2>\frac{|R_{0,n}|}{p_n}\frac{|U_{n,m}|}{|R_{0,n}|}\frac{1}{m}>(p_0-\epsilon_1)\epsilon/m$ if $\frac{|U_{n,m}|}{|R_{0,n}|}>\epsilon$ and there exists $N>0$ such that for all $n>N$, we have $\frac{|R_{0,n}|}{p_n}>p_0-\epsilon_1$ for some small $\epsilon_1>0$.
By Theorem 2, we have $\Pi(\frac{1}{p_n}\sum_{j=1}^{p_n}\bigm|\|\bbeta(\bv_j)\|_2-\|\bbeta_0(\bv_j)\|_2\bigm|>\epsilon(p_0-\epsilon_1)/m|S_n)\rightarrow 0$. 
This implies $\Pi(\frac{1}{p_n}\sum_{j\in U_{n,m}}\|\bbeta(\bv_j)\|_2>\epsilon(p_0-\epsilon_1)/m|S_n)\rightarrow 0$.
Thus, we have $\Pi(|U_{n,m}|/|R_{0,n}|<\epsilon|S_n)\rightarrow 1$.

We have $U_n=\cup_{m=1}^{\infty}E_{n,m}$ and $U_{n,m}\subset E_{n,m+1}$. Thus $\Pi(|U_{n,m}|/p_n<\epsilon|S_n)>\Pi(|E_{n,m+1}|/p_n<\epsilon|S_n)$.
By the monotone continuity of probability measure, we have $\Pi(|U_{n}|/p_n<\epsilon|S_n)=\lim_{m\rightarrow\infty}\Pi(|U_{n,m}|/p_n<\epsilon|S_n)\rightarrow 1$ as $n\rightarrow \infty$.

For the second assertion, note that, $\sum_{j\in R_{0,n}^c}\bigm|\|\bbeta(\bv_j)\|_2-\|\bbeta_0(\bv_j)\|_2\bigm| = \sum_{j\in R_{1,n}\cap R_{0,n}^c}\bigm|\|\bbeta_0(\bv_j)\|_2\bigm|+\sum_{j\in R_{0,n}^c\cap R_{1,n}^c}\bigm|\|\bbeta(\bv_j)\|_2-\|\bbeta_0(\bv_j)\|_2\bigm|$, where $R_{1,n}=\{j:\|\bbeta(\bv_j)\|_2=0\}$. Following the steps above, and $\frac{|R_{0,n}^c|}{p_n}\rightarrow 1-p_0$, we must have $R_{1,n}\cap R_{0,n}^c=\emptyset$ to ensure that Theorem 2 holds.

We have $\bbeta(\bv_j)^T\bbeta_0(\bv_j)>-\|\bbeta(\bv_j)-\bbeta_0(\bv_j)\|_2^2$.
Thus, $\|\bbeta(\bv_j)-\bbeta_0(\bv_j)\|_2^2<1/m$ implies that $\bbeta(\bv_j)^T\bbeta_0(\bv_j)>-1/m$.
Let $W_{n,m}'=\{j:\bbeta(\bv_j)^T\bbeta_0(\bv_j)>-1/m, \|\bbeta_0(\bv)\|_2>0\}$ and
$W_{n,m,c}'=\{j:\|\bbeta(\bv_j)-\bbeta_0(\bv_j)\|_2^2<1/m, \|\bbeta_0(\bv)\|_2>0\}$.
Then $\frac{|W_{n,m}'|}{p_n-R_{0,n}}\geq \frac{|W_{n,m,c}'|}{p_n-R_{0,n}}$.
From Theorem 2 and above results, we have already established that $\Pi(1-\frac{|W_{n,m,c}'|}{p_n-R_{0,n}}<\epsilon\mid S_n)\rightarrow 1$ as $n\rightarrow\infty$.
Thus $\Pi(1-\frac{|W_{n}'|}{p_n-R_{0,n}}<\epsilon\mid\S_n)=\lim_{m\rightarrow\infty}\Pi(1-\frac{|W_{n,m}'|}{p_n-R_{0,n}}<\epsilon\mid\S_n)\rightarrow 1$ as $n\rightarrow\infty$.
\end{proof}

\begin{proof}[Proof of Theorem 4]
As defined earlier, $R_0$ denotes the set of voxels with $\bbeta_0(\bv)=0$.
Define $\F_m(R_0)=\{\bbeta:\int_{R_0}\|\bbeta(\bv)\|_2\pd\bv<1/m\}$.
We have $\int_{\B}\bigm|\|\bbeta(\bv)\|_2-\|\bbeta_0(\bv)\|_2\bigm|\pd\bv>\int_{R_0}\|\bbeta(\bv)\|_2\pd\bv$.
Under Theorem 2, we immediately have $\Pi(\F_m(R_0)\mid S_n)\rightarrow 1$ as $n\rightarrow\infty$ since $\Pi(\int_{\B}\bigm|\|\bbeta(\bv)\|_2-\|\bbeta_0(\bv)\|_2\bigm|\pd\bv<1/m\mid S_n)\rightarrow 1$ as $n\rightarrow 1$.
We have,
$\{\int_{R_0}\|\bbeta(\bv)\|_2=0\}=\cap_{m=1}^{\infty} \F_m(R_0)$.
Again, applying monotone continuity, $\Pi(\{\int_{R_0}\|\bbeta(\bv)\|_2=0\}\mid S_n)=\lim_{m\rightarrow \infty}\Pi(\F_m(R_0)\mid S_n)=1$ as $n\rightarrow\infty$.
Thus $\Pi(\A(U(\bbeta))>\epsilon\mid S_n)\leq 1-\Pi(\{\int_{R_0}\|\bbeta(\bv)\|_2=0\}\mid S_n) \rightarrow 0$.

For the second assertion, note that $\int_{R_0^c}\bigm|\|\bbeta(\bv)\|_2-\|\bbeta_0(\bv)\|_2\bigm|\pd\bv = \int_{R_{10}\cap R_0^c}\bigm|\|\bbeta_0(\bv)\|_2\bigm|\pd\bv+\int_{R_0^c\cap R_{10}^c}\bigm|\|\bbeta(\bv)\|_2-\|\bbeta_0(\bv)\|_2\bigm|\pd\bv$, where $R_{10}=\{\bv:\|\bbeta(\bv)\|_2=0\}$.
By construction $\int_{R_{10}}\bigm|\|\bbeta_0(\bv)\|_2\bigm|\pd\bv>0$. Let this be $\epsilon_1$. Hence, $\Pi(\int_{R_0^c}\bigm|\|\bbeta(\bv)\|_2-\|\bbeta_0(\bv)\|_2\bigm|\pd\bv > \epsilon_1\mid S_n)$ will not converge to zero, contradicting Theorem 2. 
The third assertion holds immediately, based on the above observation. 

We have $\bbeta(\bv)^T\bbeta_0(\bv)>-\|\bbeta(\bv)-\bbeta_0(\bv)\|_2^2$.
Thus $\Pi(\A(W'(\bbeta))>1-\epsilon\mid S_n)\geq \Pi(\int_{R_0}\|\bbeta(\bv)-\bbeta_0(\bv)\|_2^2\pd\bv<\epsilon\mid S_n)\rightarrow 1$ as $n\rightarrow 1$ applying Theorem 2.
\end{proof}

\begin{proof}[Proof of Theorem 5]
Applying Theorem 2 for all $g$, we have,
$\Pi(\int \|\bbeta_g(\bv)-\bbeta_{0,g}(\bv)\|_2\pd\bv > \epsilon\mid S_n)\rightarrow 0$ as $n_g\rightarrow\infty$.
We have $\int \|\bbeta_g(\bv)-\bbeta_{0,g}(\bv)\|_2\pd\bv>\int \bigm|\|\bbeta_g(\bv)\|_2-\|\bbeta_{0,g}(\bv)\|_2\bigm|\pd\bv$.
Let $\int \|\bbeta_g(\bv)-\bbeta_{0,g}(\bv)\|_2\pd\bv<\epsilon$ and $\int \|\bbeta_{g'}(\bv)-\bbeta_{0,g'}(\bv)\|_2\pd\bv<\epsilon$.
Thus, $\int \|\bbeta_g(\bv)\|_2\pd\bv<M'+\epsilon$ and $\int \|\bbeta_g'(\bv)\|_2\pd\bv<M'+\epsilon$.

We have,
\begin{align*}
    &\int|\bbeta_g(\bv)^T\bbeta_g'(\bv)-\bbeta_{0,g}(\bv)^T\bbeta_{0,g'}(\bv)|\pd\bv\\
    &\quad\leq \int|\bbeta_g(\bv)^T(\bbeta_g'(\bv)-\bbeta_{0,g'}(\bv))|\pd\bv + \int|\bbeta_{0,g'}(\bv)^T(\bbeta_g(\bv)-\bbeta_{0,g}(\bv))|\pd\bv\\
    &\quad\leq  \int|\bbeta_g(\bv)\|_2\|\bbeta_g'(\bv)-\bbeta_{0,g'}(\bv)\|_2\pd\bv + \int\|\bbeta_{0,g'}(\bv)\|_2\|\bbeta_g(\bv)-\bbeta_{0,g}(\bv)\|_2\pd\bv\\
    &\quad\leq (2M'+\epsilon)\epsilon \lesssim \epsilon.
\end{align*}
This completes the proof of posterior consistency.
\end{proof}

\begin{proof}[Proof of Theorem 6]

Define $U_{m}(\bbeta_g,\bbeta_{g'})=\{\bv: \bbeta_g(\bv)^T\bbeta_{g'}(\bv)\leq -1/m, \bbeta_{0,g}(\bv)^T\bbeta_{0,g'}(\bv)>0\}$.

$\Pi(\A(U_{m}(\bbeta_g,\bbeta_{g'})) < \epsilon \mid S_n)\leq \Pi(\int \|\bbeta_g-\bbeta_{g'}\|_2^2\pd\bv > \epsilon/m \mid S_n)$


$\{\bbeta_g,\bbeta_{g'}: \int |\bbeta_g(\bv)^T\bbeta_{g'}(\bv)-\bbeta_{0,g}(\bv)^T\bbeta_{0,g'}(\bv)|\pd\bv>0\}$

$\int |\bbeta_g(\bv)^T\bbeta_{g'}(\bv)-\bbeta_{0,g}(\bv)^T\bbeta_{0,g'}(\bv)|\pd\bv\geq\int_{U_{m}} |\bbeta_g(\bv)^T\bbeta_{g'}(\bv)-\bbeta_{0,g}(\bv)^T\bbeta_{0,g'}(\bv)|\pd\bv>\A(U_{m})/m$, where $\A(U_{m})$ stands for the area of $U_{m}$.
Then, $\Pi(\int_{U_{m}} |\bbeta_g(\bv)^T\bbeta_g'(\bv)-\bbeta_{0,g}(\bv)^T\bbeta_{0,g'}(\bv)|\pd\bv>\A(U_{m})/m\mid S_n)\leq \Pi(\int |\bbeta_g(\bv)^T\bbeta_{g'}(\bv)-\bbeta_{0,g}(\bv)^T\bbeta_{0,g'}(\bv)|\pd\bv>\A(U_{m})/m\mid S_n)\rightarrow 0$.
Hence, we must have $\A(U_{m})=0$ for all $m$.
\end{proof}

 \section{Simulation for the nonparametric model}
Here, we present a small simulation study for the nonparametric, where the data is generated following the linear model from Simulation Setting 2. Then we fit the following model,
\begin{align}
    y_{i}=&p^{-1/2}\sum_{j=1}^pf_{j,g}(\bD_{i}(\bv_j)) + e_{i},\nonumber\\
    f_{j,g}(\bD_{i}(\bv_j))=&\sum_{k=1}^K\theta_{g,k}(\bv_j)\exp\left(-\frac{\|\bD_{i}(\bv_j)-\bq_{k}\|_2^2}{2h_g^2}\right)I\{\|\bD_{i}(\bv_j)-\bq_k\|_2<3h_g\}, \textrm{ for } i\in g,\nonumber\\
    e_{i}\sim&\Normal(0, \sigma^2), \label{eq:nonparamodel}
\end{align}
The results are calculated for $K=20$ with a sample size of 50 in each group and illustrated in Figure~\ref{fig::esticoefnon}.
We find that the non-zero locations in true $\bbeta_{0,g}$ are almost accurately identified by the estimated $\|\btheta_{g}\|_2$'s of the nonparametric model.
 \begin{figure}[htbp]
\centering
\subfigure{\includegraphics[width = 0.30\textwidth]{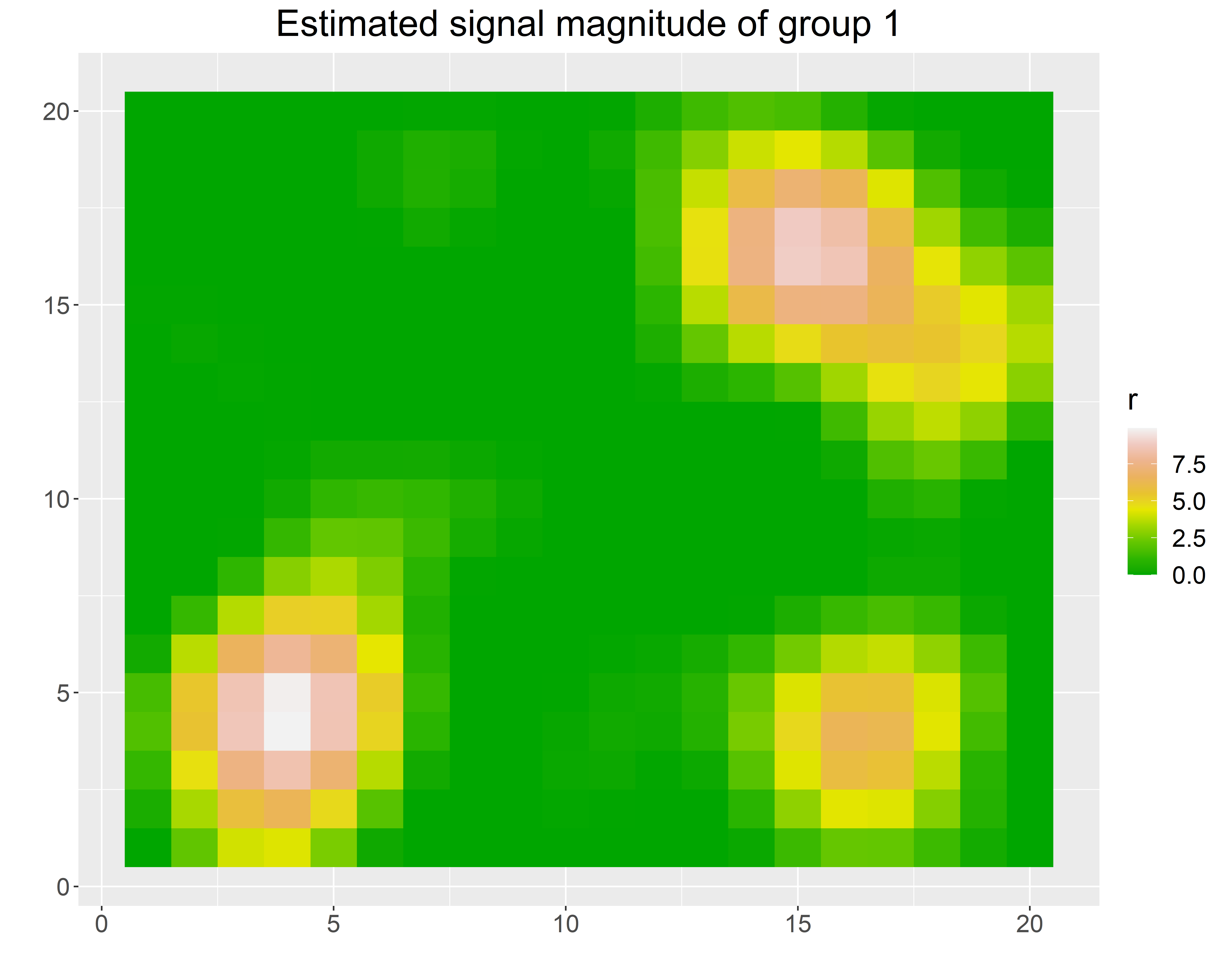}}
\subfigure{\includegraphics[width = 0.30\textwidth]{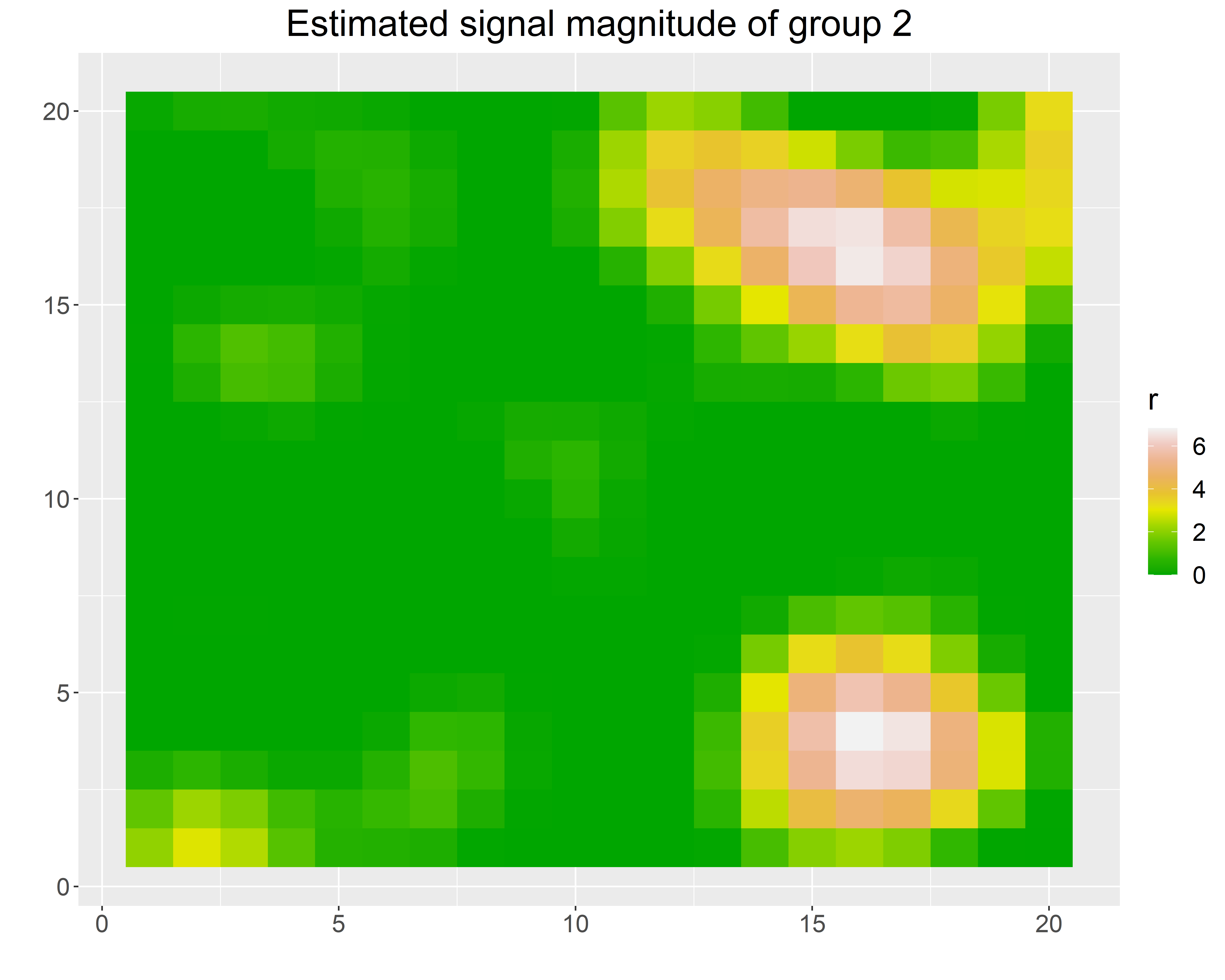}}
\subfigure{\includegraphics[width = 0.30\textwidth]{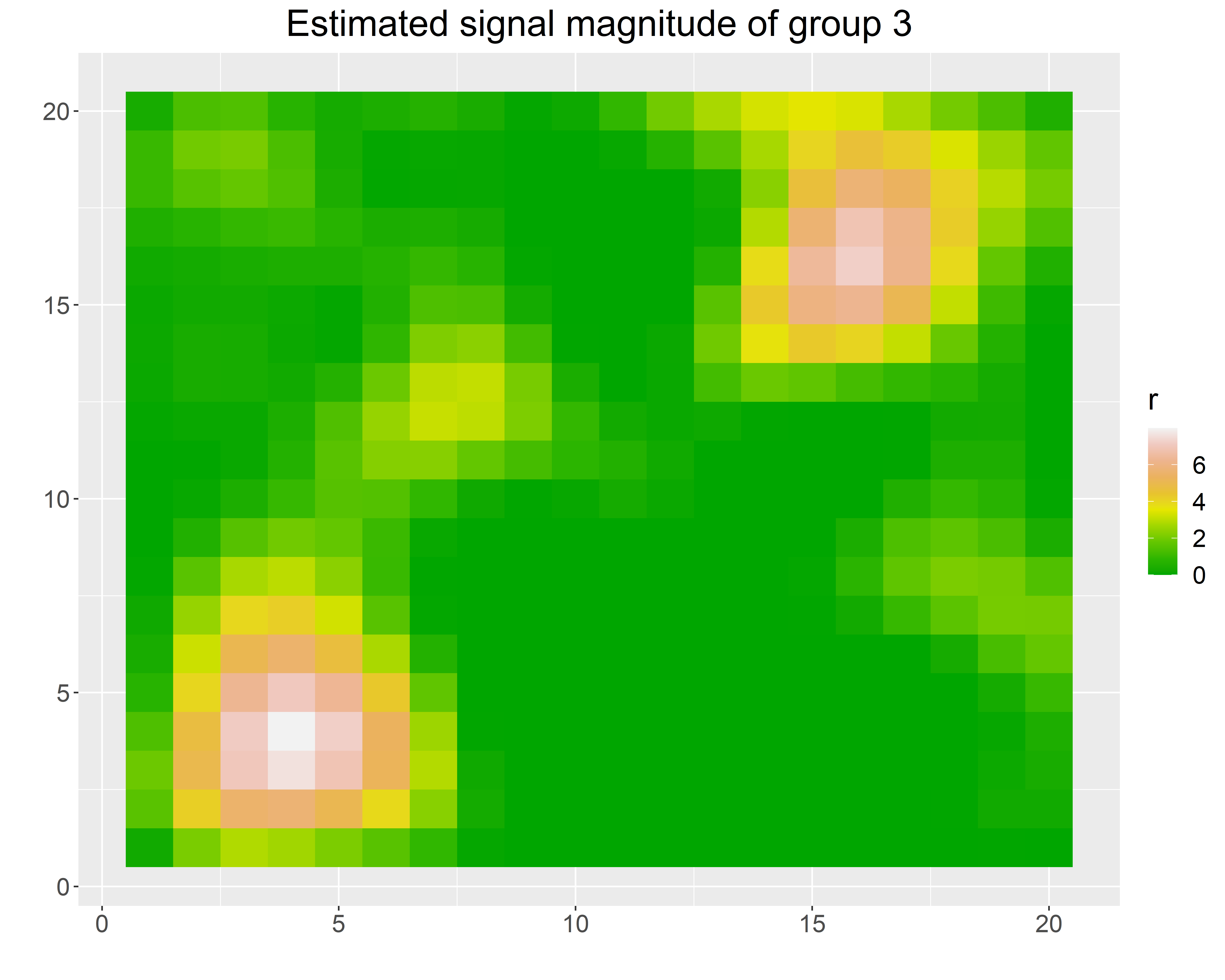}}
\caption{Estimated $\|\btheta_g\|_2$'s when the error variance is $1$ for the sample size of 50 in each group. }
\label{fig::esticoefnon}
\end{figure}

\section{Computation time}
Figure~\ref{fig:timecompare} illustrates the computation times as a function of the number of locations $(p)$ and the sample size $(n)$ in two sub figures.
As a follow-up analysis, we further estimate $a$, and $b$ assuming that the computation times are $O(p^a), $ and $O(n^b)$. Thus, we fit separate regression models in the log-scale as $\log(t)$ on $\log(p)$ on $\log(n)$ and $\log(t)$ on $\log(p)$. Our estimates suggest that the orders for $p$ and $n$ are close to quadratic and linear, respectively, based on the two regression coefficients 1.6 and 0.8 for the two cases. Furthermore, we fitted a few additional models in the original scale 1) $t$ on $p$, 2) $t$ on $p$ and $p^2$, 3) $t$ on $p,p^2$ and $p^3$, 4) $t$ on $n$ and 5) $t$ on $n$ and $n^2$. On comparing the Akaike information criterion (AIC)s, for $p$ the quadratic model and $n$ the linear model turn out to be best models. These results agree with the log-scale-based results.

\begin{figure}[htbp]
\centering
\subfigure[]{\includegraphics[width = 0.4\textwidth]{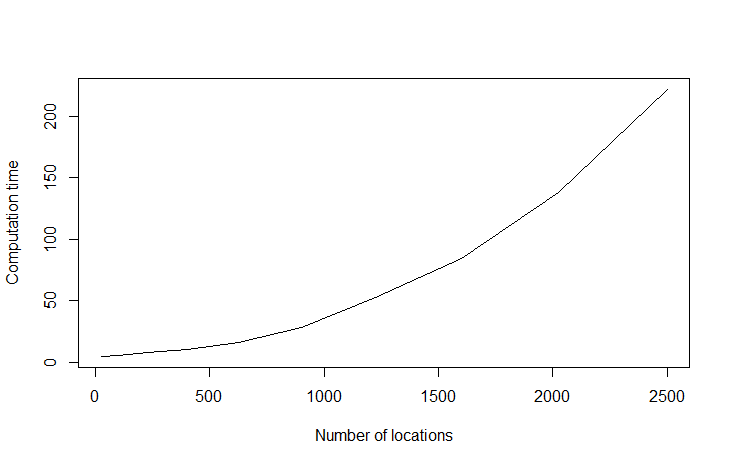}}
\subfigure[]{\includegraphics[width = 0.4\textwidth]{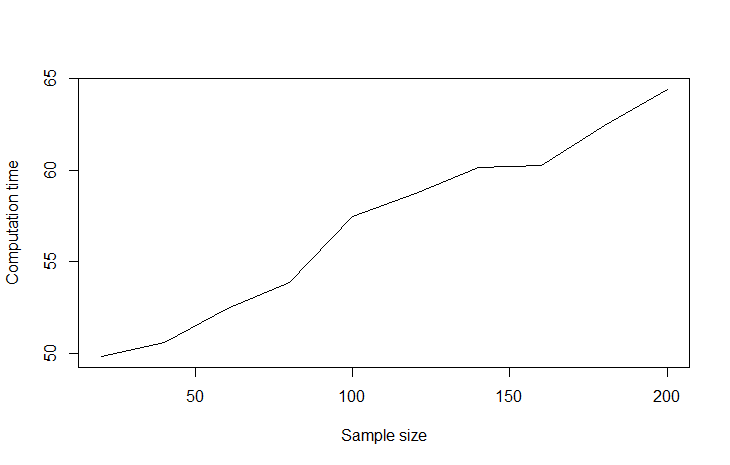}}
\caption{Logarithm of computation time comparison with increasing a) number of locations $p$, and b) number of samples $n$. The unit of time here is in minutes for 50 MCMC iterations.}
\label{fig:timecompare}
\end{figure}

\section{Posterior computation}


For simplicity, we represent $\widetilde{\bbeta}$ and $\widetilde{\balpha}_g$ as $q\times p$ matrices. Then, the priors can be represented as matrix-normals.
Specifically, $\widetilde{\bbeta}\sim$Matrix-Normal$(0,\bSigma,\bkappa)$
and $\widetilde{\balpha}_g\sim$Matrix-Normal$(0,\bSigma_g,\bkappa_g)$.

The complete proportional log-posterior for our model excluding the hyper-priors is,
\begin{align*}
    &-\frac{n}{2}\log(\sigma^2)-\sum_g\sum_{i\in g}  \{y_i-b_{0,g}-p^{-1/2}\sum_{j=1}^{p}\bD_i(\bv_j)^T\bbeta_g(\bv_j)\}^2/(2\sigma^2)\\&- \frac{1}{2} \, \mathrm{tr}\left[ {\bkappa}^{-1} \widetilde{\bbeta}^{T} \mathbf{\bSigma}^{-1} \widetilde{\bbeta} \right] - \frac{p}{2}\log(|{\bkappa}|) -\frac{q}{2}\log(|\bSigma|)\\&- \frac{1}{2} \, \mathrm{tr}\left[ {\bkappa}_g^{-1} \widetilde{\balpha}_g^{T} \mathbf{\bSigma}_g^{-1} \widetilde{\balpha}_g \right] - \frac{p}{2}\log(|{\bkappa}_g|) -\frac{q}{2}\log(|\bSigma_g|)\\&
    -\{(c_1-1)\log(\sigma^{-2})-c_2\sigma^{-2}\}.
\end{align*}
Our sampler iterates between the following steps. 

\begin{enumerate}
    \item[(1)] {\bf Updating $\sigma$:}  The full conditional distribution for $\sigma^{-2}$ is given by $\sigma^{-2}\sim \Ga\{c_1+n/2,c_2+\sum_g\sum_{i\in g} \{y_i-b_{0,g}-p^{-1/2}\sum_{j=1}^{p}\bD_i(\bv_j)^T\bbeta_g(\bv_j)\}^2/2$.
    \item [(2)] {\bf Updating intercepts $b_{0,g}$:} The full conditional distribution for $b_{0,g}$ is $\Normal(\mu_g,s^2_g)$, where $\mu_g=s^2_g\sum_{i\in g}\{y_i-p^{-1/2}\sum_{j=1}^{p}\bD_i(\bv_j)^T\bbeta_g(\bv_j)\}/\sigma^2$ and $s_g^2=(n_g/\sigma^2+1/\sigma_b^2)^{-1}$, where $n_g$ is the number of subjects in group $g$. 
    \item [(3)] {\bf Updating shared $\widetilde{\bbeta}$:} We update this parameter using HMC. Thus, we provide the derivative of the negative log-likelihood with respect to $\widetilde{\bbeta}(\bv_j)$,
    $$
    -\sum_g\sum_{i\in g}  \{y_i-b_{0,g}-p^{-1/2}\sum_{j=1}^{p}\bD_i(\bv_j)^T\bbeta_g(\bv_j)\}\bD_i(\bv_j)/(\sigma^2)\frac{\partial \bbeta_g(\bv_j)}{\partial \widetilde{\bbeta}(\bv_j)} - \left[ \mathbf{\bSigma}^{-1}\widetilde{\bbeta}{\bkappa}^{-1} \right]_{.,j}, 
    $$
    where $\frac{\partial \bbeta_g(\bv_j)}{\partial \widetilde{\bbeta}(\bv_j)}=h'_{\lambda}(\widetilde{\bbeta}_g(\bv_j))$ where, $h'_{\lambda}(\bx)=\left(1-\frac{\lambda}{\|\bx\|_2}\right)_{+}+\left(1-\frac{\lambda}{\|\bx\|_2}\right)_{+}\frac{\lambda}{\|\bx\|^3_2}\bx\odot\bx$ with $\odot$ standing for element-wise product.
    \item [(4)] {\bf Updating $\bSigma$:} Posterior of $\bSigma^{-1}$ is Wishart$(\nu+p,\widetilde{\bbeta}\bkappa^{-1}\widetilde{\bbeta}^T+\bI_q)$.
    \item [(5)] {\bf Updating parameters in $a$ and $a_{g}$:} These are updated following the similar Metropolis-Hastings (MH) strategy as in \cite{roy2021spatial}.
    \item[(6)] {\bf Updating $\widetilde{\balpha}_g$:} We update this parameter using HMC. Thus, we provide the derivative of the negative log-likelihood with respect to $\widetilde{\balpha}(\bv_j)$,
    $$
    -\sum_{i\in g}  \{y_i-b_{0,g}-p^{-1/2}\sum_{j=1}^{p}\bD_i(\bv_j)^T\bbeta_g(\bv_j)\}\bD_i(\bv_j)/(\sigma^2)\frac{\partial \bbeta_g(\bv_j)}{\partial \widetilde{\balpha}_g(\bv_j)} + \left[ \mathbf{\bSigma_g}^{-1}\widetilde{\balpha}_g{\bkappa_g}^{-1} \right]_{.,j}, 
    $$
    where $\frac{\partial \bbeta_g(\bv_j)}{\partial \widetilde{\balpha}_g(\bv_j)}=h'_{\lambda}(\widetilde{\bbeta}_g(\bv_j))h'_{\lambda_g}(\widetilde{\alpha}_g(\bv_j))$. The suffix ${.,j}$ stands for the $j$-$th$ column.

    \item[(6)] {\bf Updating $\bSigma_g$:} Posterior of $\bSigma_g^{-1}$ is Wishart$(\nu+p,\widetilde{\balpha}\bkappa_g^{-1}\widetilde{\balpha}_g^T+\bI_q)$.
    \item[(7)] {\bf Updating $\lambda$ and $\lambda_g$}: We use random walk MH for updating the thresholding parameters as in \cite{kang2018scalar}.
\end{enumerate}

\bibliographystyle{plainnat}
\bibliography{main}

\end{document}